\PassOptionsToPackage{table}{xcolor}
\documentclass{lmcs} %
\pdfoutput=1
\usepackage[utf8]{inputenc}

\usepackage{lastpage}
\lmcsdoi{22}{2}{31}
\lmcsheading{}{\pageref{LastPage}}{}{}%
{Jul.~24,~2025}{Jun.~29,~2026}{}

\usepackage{iftex}

\ifpdf
  \usepackage{underscore}         %
  \usepackage[T1]{fontenc}        %
\else
  \usepackage{breakurl}           %
\fi

\usepackage{thm-restate}
\usepackage{xspace}
\usepackage{adjustbox}
\usepackage{amsmath}
\usepackage{amssymb} %

\definecolor{USPNcobalt}{HTML}{293358}
\definecolor{USPNocre}{HTML}{8b7d6d}
\definecolor{USPNblanc}{HTML}{ffffff}
\definecolor{USPNceruleen}{HTML}{354878}
\definecolor{USPNsable}{HTML}{ad947e}

\definecolor{blueColorBlind} {RGB}{68 ,119,170}
\definecolor{greenColorBlind}{RGB}{34 ,136,51 }
\definecolor{redColorBlind}  {RGB}{238,102,119}

\usepackage{subcaption} %
\makeatletter         
\def\figurecaption#1#2{\noindent\hangindent 40pt
                       \hbox to 36pt {\small\sl #1 \hfil}
                       \ignorespaces {\small #2}}

\long\def\@makecaption#1#2{
  \vskip 10pt 
  \settowidth{\@tempdima}{#2}
  \ifdim\@tempdima>0pt
       \setbox\@tempboxa\hbox{#1: #2}
     \else
       \setbox\@tempboxa\hbox{#1 #2}
   \fi
   \ifdim \wd\@tempboxa >\hsize               
       \begin{list}{#1:}{
       \settowidth{\labelwidth}{#1:}
       \setlength{\leftmargin}{\labelwidth}
       \addtolength{\leftmargin}{\labelsep}
        }\item #2 \end{list}\par   
     \else                                    
       \hbox to\hsize{\hfil\box\@tempboxa\hfil}  
   \fi}
\makeatother

\usepackage[misc,geometry]{ifsym} %

\usepackage{pifont} %

\usepackage{pdflscape} %

\newenvironment{oneenumerate}
	{\begin{enumerate}}
	{\end{enumerate}}

\usepackage{multirow}

\newcommand{\cellHeader}[0]{\cellcolor{USPNsable!40}\bfseries}
\newcommand{\rowHeader}{\rowcolor{USPNsable!40}\bfseries}

\newcommand{\cellDecidable}{\cellcolor{greenColorBlind!75}$\surd$}
\newcommand{\cellUndecidable}{\cellcolor{redColorBlind!75}$\times$}

\usepackage{tikz}
\usetikzlibrary{arrows,automata,fit,positioning,shapes} %

\tikzstyle{pta}=[auto, ->, >=stealth']
\tikzstyle{every node}=[initial text=]
\tikzstyle{RA-location}=[rectangle, minimum size=12pt, draw=black, fill=green!10, inner sep=2pt] %
\tikzstyle{location}=[rounded rectangle, minimum size=12pt, draw=black, fill=blue!10, inner sep=2pt]
\tikzstyle{final}=[double, fill=blueColorBlind!40]

\tikzstyle{private}=[fill=redColorBlind!30,thick]

\tikzset{
  location2gen/.style={
    draw=black,
    rounded corners,
	align=center,
    inner sep=2pt,
    font=\small
  }
}

\tikzset{
  location2/.style={
	location2gen,
    rectangle split,
    rectangle split parts=2,
    rectangle split horizontal=false,
	rectangle split part fill={blue!5, blue!15},
  }
}

\newcommand{\location}[3][]{%
  \node[location2,#1] {%
    \scriptsize $\ensuremath{$#2$}$
    \nodepart{second}
    #3
  };
}

\definecolor{coloract}{rgb}{0, 0.3, 0}
\definecolor{colorclock}{rgb}{0.4, 0, 0}
\definecolor{colordisc}{rgb}{1, 0, 1}
\definecolor{colorloc}{rgb}{0.4, 0.4, 0.65}
\definecolor{colorparam}{rgb}{1, 0.6, 0.0}

\newcommand{\styleact}[1]{\ensuremath{\textcolor{coloract}{{\textstyleact{#1}}}}}
\newcommand{\styleclock}[1]{\ensuremath{\textcolor{colorclock}{{#1}}}}

\newcommand{\textstyleact}[1]{\ensuremath{\mathit{#1}}}
\newcommand{\textstyleclock}[1]{\ensuremath{\mathit{#1}}}
\newcommand{\textstyleloc}[1]{\ensuremath{#1}}
\definecolor{USPNcobalt}{HTML}{293358}
\definecolor{USPNocre}{HTML}{8b7d6d}
\definecolor{USPNblanc}{HTML}{ffffff}
\definecolor{USPNceruleen}{HTML}{354878}
\definecolor{USPNsable}{HTML}{ad947e}

\usepackage[
		pdfauthor={Etienne Andre, Sarah Depernet and Engel Lefaucheux},%
		pdftitle={The Bright Side of Timed Opacity},
		breaklinks  = true,
		colorlinks  = true,
		citecolor   = USPNsable,
		linkcolor   = USPNocre,
		urlcolor    = USPNceruleen,
	]{hyperref}

\newcommand{\defProblem}[3]
{%
	\noindent\fcolorbox{black}{USPNsable!20}{
	\begin{minipage}{.95\columnwidth}
		\textbf{#1:}\\
		\textsc{Input}: #2\\
		\textsc{Problem}: #3
	\end{minipage}
}

	\smallskip

}

\newcommand{\assign}{\leftarrow}

\newcommand{\checkUseMacro}[1]{#1}

\newcommand{\Time}{\ensuremath{\mathbb{T}}} %
\newcommand{\setsmall}[1]{\ensuremath{\{#1\}}}
\newcommand{\set}[1]{\ensuremath{\left\{#1\right\}}}

\newcommand{\setN}{\ensuremath{\mathbb{N}}}
\newcommand{\setQ}{\ensuremath{\mathbb{Q}}}
\newcommand{\setQgeqzero}{\ensuremath{\setQ_{\geq 0}}}
\newcommand{\setR}{\ensuremath{\mathbb{R}}}
\newcommand{\setRgeqzero}{\ensuremath{\setR_{\geq 0}}}
\newcommand{\setZ}{\ensuremath{\mathbb{Z}}}
\newcommand{\BTrue}{\ensuremath{\mathit{true}}}

\newcommand{\compOp}{\bowtie}

\newcommand{\intpart}[1]{\ensuremath{\lfloor#1\rfloor}}
\newcommand{\fract}[1]{\ensuremath{\text{frac}(#1)}}
\newcommand{\init}{\ensuremath{0}}
\newcommand{\priv}{\ensuremath{{\mathit{priv}}}}

\newcommand{\final}{\ensuremath{f}}

\newcommand{\styleAutomaton}[1]{\ensuremath{\mathcal{#1}}}

\newcommand{\clock}{\ensuremath{\textstyleclock{x}}}

\newcommand{\clockx}{\ensuremath{\textstyleclock{x}}}
\newcommand{\clocky}{\ensuremath{\textstyleclock{y}}}
\newcommand{\clockz}{\ensuremath{\textstyleclock{z}}}
\newcommand{\clocki}[1]{\ensuremath{\textstyleclock{\clock_{#1}}}}
\newcommand{\ClockCard}{\ensuremath{H}} %
\newcommand{\clockval}{\ensuremath{\mu}}
\newcommand{\ClockSet}{\ensuremath{\mathbb{X}}} %
\newcommand{\ClocksZero}{\ensuremath{\vec{0}}}

\newcommand{\resets}{\ensuremath{R}}
\newcommand{\reset}[2]{\ensuremath{[#1]_{#2}}}

\newcommand{\TA}{\ensuremath{\checkUseMacro{\styleAutomaton{A}}}}
\newcommand{\TB}{\ensuremath{\checkUseMacro{\styleAutomaton{B}}}} %
\newcommand{\TAprivextend}{\ensuremath{\left(\ActionSet, \LocSet, \locinit, \PrivSet, \FinalSet, \ClockSet, \invariant, \EdgeSet\right)}}

\newcommand{\action}{\ensuremath{\textstyleact{a}}}
\newcommand{\ActionSet}{\ensuremath{\Sigma}}
\newcommand{\constraint}{\ensuremath{C}}

\newcommand{\edge}{\ensuremath{\checkUseMacro{e}}}
\newcommand{\edgei}[1]{\ensuremath{\checkUseMacro{\edge_{#1}}}}
\newcommand{\EdgeSet}{\ensuremath{E}}

\newcommand{\guard}{\ensuremath{g}}

\newcommand{\invariant}{\ensuremath{I}} %

\newcommand{\loc}{\ensuremath{\textstyleloc{\ell}}}
\newcommand{\loci}[1]{\ensuremath{\textstyleloc{\loc_{#1}}}}
\newcommand{\locinit}{\ensuremath{\textstyleloc{\loc_\init}}}
\newcommand{\locfinal}{\ensuremath{\textstyleloc{\loc_\final}}}
\newcommand{\FinalSet}{\ensuremath{L_f}}
\newcommand{\PrivSet}{\ensuremath{L_{\mathit{priv}}}}
\newcommand{\locpriv}{\ensuremath{\textstyleloc{\loc_\priv}}}

\newcommand{\LocSet}{\ensuremath{L}}
\newcommand{\longuefleche}[1]{\stackrel{#1}{\longrightarrow}}

\newcommand{\run}{\checkUseMacro{\rho}} %

\newcommand{\TimedWords}[1]{\ensuremath{\textit{TW}(#1)}}
\newcommand{\Language}{\ensuremath{\mathcal{L}}}

\newcommand{\sqexp}{\ensuremath{\textit{exp}}}

\newcommand{\LargestConstant}{\ensuremath{M}}

\newcommand{\region}{\ensuremath{r}}
\newcommand{\regioni}[1]{\ensuremath{\region_{#1}}}

\newcommand{\Regions}[1]{\ensuremath{\styleAutomaton{R}_{#1}}}
\newcommand{\NbClockRegions}{\ensuremath{\vert \ClockSet \vert ! \cdot 2^{\vert \ClockSet \vert} \cdot \prod\limits_{x \in \ClockSet}(2 M(x) + 2)}}
\newcommand{\RegionAutomaton}[1]{\ensuremath{\styleAutomaton{RA}_{#1}}}

\newcommand{\silentaction}{\ensuremath{\varepsilon}}

\newcommand{\Subwords}{\ensuremath{\mathit{Subwords}}}

\newcommand{\strategy}[1]{\ensuremath{\sigma_{#1}}}
\newcommand{\projection}[1]{\ensuremath{\pi_{\strategy{#1}}}}
\newcommand{\tweak}[2]{\ensuremath{\mathit{tweak}_{#1 \rightarrow #2}}}

\newcommand{\semantics}[1]{\ensuremath{\mathfrak{T}_{#1}}}
\newcommand{\semanticsextend}{\ensuremath{\left(\StateSet, \concstateinit, \ActionSet \cup \{ \silentaction \} \cup \setRgeqzero, \transition\right)}}

\newcommand{\transition}{{\ensuremath{\rightarrow}}}
\newcommand{\StateSet}{\ensuremath{\mathfrak{S}}}
\newcommand{\concstate}{\ensuremath{\mathfrak{s}}}

\newcommand{\concstateinit}{\ensuremath{\concstate_\init}}

\newcommand{\PrivVisit}[1]{\ensuremath{\mathit{Visit}^{\mathit{priv}}(#1)}}
\newcommand{\PubVisit}[1]{\ensuremath{\mathit{Visit}^{\mathit{pub}}(#1)}}

\newcommand{\PrivateTr}[1]{\ensuremath{\mathit{Tr}^\mathit{priv}(#1)}}
\newcommand{\PublicTr}[1]{\ensuremath{\mathit{Tr}^{\mathit{pub}}(#1)}}
\newcommand{\Trace}[1]{\ensuremath{\mathit{Tr}(#1)}}

\newcommand{\AllRuns}{\ensuremath{\Omega}}

\newcommand{\APriv}{\ensuremath{\TA_{\mathit{priv}}}}
\newcommand{\APub}{\ensuremath{\TA_{\mathit{pub}}}}
\newcommand{\AMemo}{\ensuremath{\TA_{\mathit{memo}}}}

\newcommand{\Unfold}[1]{\ensuremath{\mathit{Unfold}_N(#1)}}

\newcommand{\PrivSetOn}{\ensuremath{L_{\mathit{On, priv}}}} %
\newcommand{\PrivSetOff}{\ensuremath{L_{\mathit{Off, priv}}}} %

\newcommand{\untimed}[1]{\ensuremath{\mathit{untimed}(#1)}}
\newcommand{\timed}[1]{\ensuremath{\mathit{time}(#1)}}

\newcommand{\Tick}[1]{\ensuremath{\mathit{Tick}(#1)}}
\newcommand{\TickN}[1]{\ensuremath{\mathit{Tick}_N(#1)}}

\newcommand{\interval}[2]{\ensuremath{[\![#1 ; #2]\!]}}

\newcommand{\tausimple}{\ensuremath{\tilde{\tau}}} %

\newcommand{\mathnth}[1]{\ensuremath{#1}th}

\usepackage{pifont} %
\newcommand{\circleone}{\ding{172}}
\newcommand{\circletwo}{\ding{173}}
\newcommand{\circlethree}{\ding{174}}

\newcommand{\ComplexityFont}[1]{{\sffamily\upshape #1}}

\newcommand{\NEXPTIME}{\ComplexityFont{NEXPTIME}\xspace}
\newcommand{\coNEXPTIME}{\ComplexityFont{Co-NEXPTIME}\xspace}
\newcommand{\PSPACE}{\ComplexityFont{PSPACE}\xspace}
\newcommand{\NLOGSPACE}{\ComplexityFont{NLOGSPACE}\xspace}
\newcommand{\EXPSPACE}{\ComplexityFont{EXPSPACE}\xspace}
\newcommand{\twoEXPSPACE}{\ComplexityFont{2-EXPSPACE}\xspace}

\newcommand{\eg}{e.g.,\xspace}

\newcommand{\ie}{i.e.,\xspace}
\newcommand{\st}{s.t.}

\newcommand{\wrt}{w.r.t.\xspace}

\theoremstyle{plain}
\usepackage[capitalise,english,nameinlink]{cleveref} %
\crefname{line}{\text{line}}{\text{lines}} %
\crefname{thm}{\text{Theorem}}{\text{Theorems}}
\crefname{cor}{\text{Corollary}}{\text{Corollaries}}
\crefname{lem}{\text{Lemma}}{\text{Lemmas}}
\crefname{prop}{\text{Proposition}}{\text{Propositions}}
\crefname{rem}{\text{Remark}}{\text{Remarks}}
\crefname{exa}{\text{Example}}{\text{Examples}}
\crefname{defi}{\text{Definition}}{\text{Definitions}}
\crefname{thmC}{\text{Theorem}}{\text{Theorems}}

\usepackage{fontawesome}
\newcommand{\homepage}[1]{\href{#1}{\color{gray}\faHome}}
\usepackage{graphicx}


\begin{document}
\sloppy

\title{The Bright Side of Timed Opacity} %
\titlecomment{{\lsuper*}%
This is the extended version of the manuscript of the same name published in the proceedings of the 25th International Conference on Formal Engineering Methods (ICFEM 2024).
}
\thanks{This work is partially supported by ANR BisoUS (ANR-22-CE48-0012) and by ANR GUMMIS (ANR-25-CE48-5096).}	%

\author[\'E.~Andr\'e]{\'Etienne Andr\'e\lmcsorcid{0000-0001-8473-9555}}[a,b,c]
\author[S.~Dépernet]{Sarah Dépernet\lmcsorcid{0009-0003-8710-7934}}[d]
\author[E.~Lefaucheux]{Engel Lefaucheux\lmcsorcid{0000-0003-0875-300X}}[d]

\address{Université Sorbonne Paris Nord, CNRS, Laboratoire d’Informatique de Paris Nord, LIPN, F-93430 Villetaneuse, France}	%
\address{Institut universitaire de France (IUF)}

\address{Nantes Université, LS2N, CNRS, Bretagne, France}

\address{Université de Lorraine, CNRS, Inria, LORIA, F-54000 Nancy, France}	%

\begin{abstract}
Timed automata (TAs) are an extension of finite automata that can measure and react to the passage of time, providing the ability to handle real-time constraints using clocks.
In~2009, Franck~Cassez showed that the timed opacity problem, where an attacker can observe some actions with their timestamps and attempts to deduce information, is undecidable for~TAs.
Moreover, he showed that the undecidability holds even for subclasses such as event-recording automata.
In this article, we consider the same definition of opacity, by restricting either the system or the attacker.
Our first contribution is to prove the inter-reducibility of two variants of opacity: full opacity (for which the observations should be the same regardless of the visit of a private location) and weak opacity (for which it suffices that the attacker cannot deduce whether the private location was visited, but for which it is harmless to deduce that it was not visited); we also prove further results including a connection with timed language inclusion.
Our second contribution is to study opacity for several subclasses of TAs:
with restrictions on the number of clocks, the number of actions, the nature of time, or a new subclass called observable event-recording automata.
We show that opacity is mostly decidable in these cases, except for one-action TAs and for one-clock TAs with $\varepsilon$-transitions, for which undecidability remains.
Our third (and arguably main) contribution is to propose a new definition of opacity in which the number of observations made by the attacker is limited to the first $N$ observations, or to a set of $N$ timestamps after which the attacker observes the first action that follows immediately.
This set can be defined either \emph{a~priori} or at runtime; all three versions yield decidability for the whole TA class.
\end{abstract}

\maketitle{}
\section{Introduction}\label{section:introduction}

The concept of opacity~\cite{Mazare04,BKMR08} formalizes the absence of information leakage about designated secrets in partially observed systems. Intuitively, a system is opaque if, from the perspective of an external observer (attacker) who sees only a subset of system events, every execution containing secret behaviour is observationally indistinguishable from at least one non-secret execution. Thus, opacity represents a security guarantee: the observer cannot deduce with certainty that a protected sequence of actions occurred.
Time plays a critical role in this deduction power.
As demonstrated in~\cite{GMR07}, systems that appear opaque when abstracting away timing constraints may become non-opaque once real-time information—such as event timestamps—is taken into account.

These timing attacks belong to the broader landscape of side-channel attacks, in which attackers infer sensitive information from secondary physical or behavioural signals (such as power consumption, electromagnetic emanations, or cache use) rather than from the intended functional interface~\cite{Standaert2010}.
For instance, a timing attack against the Chinese public-key cryptography standard revealed that leakage of leading zero bits, inferred from measured runtimes, suffices to reconstruct the secret key~\cite{CHSJX22}.

In order to specify and verify systems and their vulnerability to timing leaks, one needs appropriate formalisms.
Timed automata (TAs) are an extension of finite automata that can measure and modify their behaviour through real-time constraints.
They are equipped with a finite set of clocks that can be compared with integer constants, and reset.
As such, they are  widely used in order to model real-time reactive and concurrent systems with timing constraints.
Consequently, analysing whether a TA preserves opacity is important to mitigate timing leaks in safety- and security-critical real-time systems.

\subsection{Related work}
There are several ways to define opacity problems in TAs, depending on the power of the attacker.
The common idea is to ensure that the attacker cannot deduce from the observation of a run whether it was a private or a public run.
The attacker in~\cite{Cassez09} is able to observe a subset $\ActionSet_o \subseteq \ActionSet$ of actions with their timestamps.
In this context, a timed word~$w$ is said to be opaque if there exists a public run that produces the projection of~$w$ following $\ActionSet_o$ as an observed timed word.
In this configuration, one can consider the opacity problem consisting of determining, knowing a TA~$\TA$ and a set of timed words, whether all words in this set are opaque in~$\TA$.
This problem has been shown to be undecidable for TAs~\cite{Cassez09}.
This notably relates to the undecidability of timed language inclusion for TAs~\cite{AD94}.
However, the undecidability holds in~\cite{Cassez09} even for the restricted class of event-recording automata (ERAs)~\cite{AFH99} (a subclass of~TAs), for which language inclusion is decidable.
The aforementioned negative results leave hope only if the definition or the setting is changed, which was done in four main lines of work.

First, in~\cite{WZ18,WZA18}, the input model is simplified to \emph{real-time automata}~\cite{Dima01}, a restricted formalism compared to~TAs.
In this setting, (initial-state) opacity becomes decidable~\cite{WZ18,WZA18}.
In~\cite{Zhang24}, Zhang studies labelled real-timed automata (a subclass of labelled TAs); in this setting, state-based (at the initial time, the current time, etc.)\ opacity is proved to be decidable by extending the observer (that is, the classical powerset construction) from finite automata to labelled real-timed automata.

Second, in~\cite{AEYM21}, the authors consider a time-bounded notion of the opacity of~\cite{Cassez09}, where the attacker has to disclose the secret before an upper bound, using a partial observability.
This can be seen as a secrecy with an \emph{expiration date}.
The rationale is that retrieving a secret ``too late'' can be considered as useless; this is understandable, \eg{} when the secret is the value in a cache; if the cache was overwritten since, then knowing the secret is probably useless in many situations.
In addition, the analysis is carried over a time-bounded horizon; this means there are two time bounds in~\cite{AEYM21}: one for the secret expiration date, and one for the bounded-time execution of the system.
The authors prove that this problem is decidable for~TAs.
A somehow similar framework is considered in~\cite{SLR23}, in which the attacker has a bounded memory, and a finite duration between distinct observations is required, in which case the problem is decidable and \PSPACE{}-complete.

Third, in~\cite{ALMS22,ALM23}, the authors present an alternative definition to
Cassez's opacity by studying \emph{execution-time opacity}: the attacker has only access to the execution time of the system, as opposed to Cassez' partial observations with some observable events (with their timestamps).
In that case, most problems become decidable (see \cite{ALLMS23} for a survey).
	Untimed control in this setting was considered in~\cite{ABLM22}, while timed control was considered in~\cite{ADLL25journal}.
	In addition, \cite{ALMS22,ALM23,AAL24} consider also \emph{parametric} versions of the execution-time opacity problems, in which timing parameters~\cite{AHV93} can be used in order to make the system execution-time opaque, while execution-time opacity for an extension of timed automata with \emph{energy variables} is studied in~\cite{AB26}.

Finally, opacity was considered in the discrete time settings (where clock valuations are restricted to $\setN$)~\cite{CG00,KKG24,AGWZH24} where opacity, including 
additional variants regarding current-location timed opacity and initial-location timed opacity, was shown to be decidable.
We complement their researches by providing the exact complexity of this problem (\cref{theorem:discrete-time}).
In~\cite{DQY25}, another interpretation of finite accuracy was shown decidable for opacity, considering an attacker which has a finite but unknown observation precision.

Regarding non-interference for TAs, some decidability results are proved in~\cite{BDST02,BT03}, while control was considered in~\cite{BCLR15} and a parametric timed extension in~\cite{AK20}.
General security problems for TAs are surveyed in~\cite{AA23survey}.
\subsection{Contributions}

Considering the negative decidability results from~\cite{Cassez09}, there are mainly two directions: one can consider more restrictive classes of automata, or one can limit the capabilities of the attacker---we address both directions in this work.

In this work, we consider the secret to be a set of private locations.
Our attacker model is as follows: the attacker knows the TA modelling the system and can observe (some) actions together with their timestamp, but never gains access to the values of the clocks, nor knows in which location the system is.
Their goal is to deduce from these observations whether a private location was visited.

We address here three variants\footnote{%
	We adapt from~\cite{ALLMS23} the concepts of these three different levels of opacity, even though they were initially introduced in~\cite{ALLMS23} in the context of ``execution-time opacity'' in which the attacker has access only to the total execution time.
} of opacity:
\begin{enumerate}
	\item $\exists$-opacity (``there exists a pair of runs, one visiting and one not visiting the private locations set, that cannot be distinguished''),
	\item weak opacity (``for any run visiting the private locations set, there is another run not visiting it and the two cannot be distinguished'') and
	\item full opacity (weak opacity, with the other direction holding as well).
\end{enumerate}

Our set of contributions is threefold.

\subsubsection{Inter-reducibility}
Our first contribution (\cref{section:inter-reduc}) is to position full, weak and $\exists$-opacity with respect to more classical problems of~TAs.
More precisely, we establish the following results:
\begin{itemize}
\item Weak and full opacity are inter-reducible;
\item $\exists$-opacity and reachability in~TAs are inter-reducible; and
\item Weak opacity and language inclusion of~TAs are inter-reducible.
\end{itemize}

The fact that we prove that weak opacity and full opacity are inter-reducible is not only interesting \emph{per~se}, but also allows us to consider only one of both cases in the remainder of the paper.
The two other inter-reducibility results allow us to rely on the vast existing literature  on reachability and language inclusion for~TAs.

\subsubsection{Opacity in subclasses of TAs}
Throughout the second part of this paper (\cref{section:opacity:TA:res}), we consider the same attacker settings as in~\cite{Cassez09} but for natural subclasses of~TAs: first we deal with one-action~TAs, then with one-clock~TAs (both with and without $\silentaction$-transitions---a mostly technical consideration which makes a difference in decidability), TAs over discrete time, and a new subclass which we call observable ERAs.
Precisely, we show that:
\begin{enumerate}
	\item The problem of  $\exists$-opacity is decidable for general TAs and thus for all subclasses of TAs we consider as well (\cref{section:opacity:exists}).
	\item The problems of weak and full opacity are both undecidable for TAs with only one action even in the absence of $\silentaction$-transitions (\cref{section:opacity:one-action}) or two clocks (\cref{section:opacity:one-clock}).
	\item These two problems are also undecidable for TAs with a single clock, unless we forbid $\silentaction$-transitions, in which case the problems become decidable (\cref{section:opacity:one-clock}).
	\item These two problems are decidable for unrestricted TAs over discrete time (\cref{section:opacity:discrete-time}), as well as for observable ERAs (\cref{section:opacity:era}).
\end{enumerate}
These results overall build on existing results from the literature. %
They however allow us to draw a clear border between decidability and undecidability.
Moreover, we provide the exact complexity for most of the decidable results which, in some cases, complicates the proofs.

\subsubsection{Reducing the attacker power}
Then, in the third part (\cref{section:finite}), we introduce a new approach in which we reduce the visibility of the attacker to a \emph{finite} number of actions.
We consider three different settings:
\begin{enumerate}%
	\item when the attacker can only see the first $N$ events (letters) of the system, \ie{} the beginning of the run;
	\item when the attacker can decide \emph{a~priori} (statically) a set of timestamps, and they will observe the first event following each of these timestamps; and
	\item when the attacker can decide this set of timestamps at runtime depending on what observations they made until now.
\end{enumerate}%
This models the case of an attacker with a limited attack budget, while considering the full class of~TAs.
We prove that all three settings are decidable for the full TA formalism and \coNEXPTIME{}-complete.

We finally study in \cref{ss:combined} the combination of the reductions of the attacker power introduced in \cref{section:finite} on the one hand with the restrictions of the model from \cref{section:opacity:TA:res} on the other hand.

\paragraph{Side results on language inclusion for timed automata}

As a proof ingredient for \cref{section:opacity:discrete-time}, we also show that timed language inclusion is \EXPSPACE{}-complete for TAs over discrete time; while this result was proved in the past~\cite{CG00,KKG24,AGWZH24}, the exact complexity was---surprisingly---never investigated.

Furthermore, as a proof ingredient for \cref{sec:Nfirst}, we prove that inclusion of $N$-bounded timed languages of TAs is a decidable and \coNEXPTIME{}-complete problem.

\subsection{About this manuscript}

This manuscript is the extended version of~\cite{ADL24}.
In addition to reorganizing the presentation for clarity, we significantly revised and extended several technical results. 
First, the opacity results on subclasses of timed automata have been rewritten to rely more 
on the established inter-reducibility results, providing a more unified framework. In particular, the undecidability of weak/full opacity for TAs over one action was strengthened
by removing the need for $\silentaction$-transitions.
Second, the proof for the opacity problems for an attacker limited to the first $N$ observations has been completely redesigned, reducing the upper bound from \twoEXPSPACE{} to \coNEXPTIME{}---we also established the corresponding hardness result.
Finally, we extended the attacker model by allowing the attacker to strategically choose when to perform these $N$~observations, thereby strengthening the results and providing a more realistic attacker model.

\subsection{Outline}
\cref{section:preliminaries} recalls necessary preliminaries.
\cref{section:opacity:TA} defines the problems of interest.
\cref{section:inter-reduc} proves inter-reducibility between several definitions of opacity together with connections with reachability and language inclusion problems for TAs.
\cref{section:opacity:TA:res} studies opacity problems for subclasses of~TAs, while \cref{section:finite} reduces the power of the attacker to a finite set of observations.
\cref{section:conclusion} concludes.

\section{Preliminaries}\label{section:preliminaries}

We denote by $\setN, \setZ, \setQgeqzero, \setRgeqzero$ the sets of non-negative integers, integers, non-negative rationals and non-negative reals, respectively.
If $a$ and~$b$ are two integers with $a \leq b$, the set $\set{a, a+1, \dots, b-1, b}$ is denoted by $\interval{a}{b}$.

\subsection{Clock constraints}

We let $\Time$ be the domain of the time, which will be either non-negative reals $\setRgeqzero$ (continuous-time semantics) or naturals $\setN$ (discrete-time semantics).
Unless otherwise specified, we assume $\Time = \setRgeqzero$.

\emph{Clocks} are real-valued variables that all evolve over time at the same rate.
Throughout this paper, we assume a set~$\ClockSet = \{ \clocki{1}, \dots, \clocki{\ClockCard} \} $ of \emph{clocks}.
A \emph{clock valuation} is a function
$\clockval : \ClockSet \rightarrow \Time$, assigning a non-negative value to each clock.
We write $\ClocksZero$ for the clock valuation assigning $0$ to all clocks.
Given a constant $d \in \Time$, $\clockval + d$ denotes the valuation \st\ $(\clockval + d)(\clock) = \clockval(\clock) + d$, for all $\clock \in \ClockSet$. If $R$ is a subset of $\ClockSet$ and $\clockval$ a clock valuation, we call \emph{reset} of the clocks of $\resets$ and denote by $[\clockval]_R$ the valuation \st\ for all clock $x \in \ClockSet$, $\reset{\clockval}{\resets}(x) = 0$ if $x \in \resets$ and $\reset{\clockval}{\resets}(x) = \clockval(x)$ otherwise.

We assume ${\compOp} \in \{<, \leq, =, \geq, >\}$.
A \emph{constraint}~$\constraint$ is a conjunction of inequalities over~$\ClockSet $ of the form
$\clock \compOp d$, with
$d \in \setZ$.
Given~$\constraint$, we write~$\clockval \models \constraint$ if the expression obtained by replacing each~$\clock$ with~$\clockval(\clock)$ in~$\constraint$ evaluates to true.

\subsection{Timed Automata}
A TA is a finite automaton extended with a finite set of clocks with values in~$\Time$. %
We also add to the standard definition of TAs a special private locations set, which will be used to define our subsequent opacity concepts.

\begin{defi}[TA~\cite{AD94}]\label{def:TA}
	A TA~$\TA$ is a tuple \mbox{$\TA = \TAprivextend$}, where:
	\begin{oneenumerate}%
		\item $\ActionSet$ is a finite set of actions,
		\item $\LocSet$ is a finite set of locations,
		\item $\locinit \in \LocSet$ is the initial location,
		\item $\PrivSet \subseteq \LocSet$ is a set of private locations,
		\item $\FinalSet \subseteq \LocSet$ is a set of final locations,
		\item $\ClockSet$ is a finite set of clocks,
		\item $\invariant$ is the invariant, assigning to every $\loc\in \LocSet$ a constraint $\invariant(\loc)$ %
			(called \emph{invariant}),
		\item $\EdgeSet$ is a finite set of edges  $\edge = (\loc,\guard,\action,\resets,\loc')$
		where~$\loc,\loc'\in \LocSet$ are the source and target locations, $\action \in \ActionSet \cup \{ \silentaction \}$ (where $\silentaction$ denotes an unobservable action),
		$\resets\subseteq \ClockSet$ is a set of clocks to be reset, and $\guard$ is a constraint %
			(called \emph{guard}).
	\end{oneenumerate}%
\end{defi}
\begin{figure}[tb]

	\centering

	\begin{tikzpicture}[pta, scale=2, xscale=1, yscale=.5]

		\location[initial,name=s0,at={(0, -1)}]{$\styleclock{\clock} \leq 3$}{$\loci{0}$};

		\location[private,name=s2,at={(1.5, 0)}]{$\styleclock{\clock} \leq 2$}{$\loci{2}$};

		\node[location, final] at (3, -1) (s1) {$\loci{1}$};

		\path (s0) edge[bend left] node[align=center]{$\styleclock{\clock} \geq 1$\\$\styleact{\silentaction}$} (s2);
		\path (s0) edge[loop above] node[]{$\styleact{a}$} (s0);
		\path (s0) edge[] node[below]{$\styleact{b}$} (s1);
		\path (s2) edge[bend left] node[]{$\styleact{b}$} (s1);

	\end{tikzpicture}
	\caption{A TA example}
	\label{figure:example-TA}

\end{figure}
\begin{exa}
	In \cref{figure:example-TA}, we give an example of a TA with three locations $\loci{0}$, $\loci{1}$ and $\loci{2}$,
	three edges, two observable actions~$\{a, b\}$,
	and one clock $\clock$.
	$\loci{0}$ is the initial location, $\loci{2}$ is the (unique) private location, and~$\loci{1}$ is the (unique) final location.
	$\loci{0}$ has an invariant ``$\textstyleclock{\clock} \leq 3$'' and the edge from $\loci{0}$ to $\loci{2}$ is labelled by the unobservable action $\silentaction$ and has a guard ``$\textstyleclock{\clock} \geq 1$''.
\end{exa}
\paragraph{Semantics of timed automata}

We recall the semantics of a TA using a timed transition system (TTS).

\begin{defi}[Semantics of a TA]\label{def:semantics}
	Given a TA $\TA = \TAprivextend$,
	the semantics of~$\TA$ is given by the TTS $\semantics{\TA} = \semanticsextend$, with
	\begin{enumerate}
		\item $\StateSet = \big\{ (\loc, \clockval) \in \LocSet \times \setRgeqzero^\ClockSet \mid \clockval \models \invariant(\loc) \big\}$,
		\item $\concstateinit = (\locinit, \ClocksZero) $,
		\item  $\transition \subseteq \StateSet \times \EdgeSet \times \StateSet \cup \StateSet \times \setRgeqzero \times \StateSet$ consists of the discrete and (continuous) delay transition relations:
		\begin{enumerate}
			\item discrete transitions:
$\big((\loc,\clockval), \edge, (\loc',\clockval') \big) \in \transition$,
		if $(\loc, \clockval) , (\loc',\clockval') \in \StateSet$, $\edge = (\loc,\guard,\action,\resets,\loc') \in \EdgeSet$, $\clockval'= \reset{\clockval}{\resets}$, and $\clockval\models \guard$.
			\item delay transitions:
$\big((\loc,\clockval), d, (\loc,\clockval +d)\big) \in \transition$,
		if $d \in \setRgeqzero$ and
$\forall d' \in [0, d], (\loc, \clockval+d') \in \StateSet$.
		\end{enumerate}
	\end{enumerate}
\end{defi}

Moreover we write $(\loc, \clockval)\longuefleche{(d, \edge)} (\loc',\clockval')$ for a combination of a delay and discrete transition if
$\exists  \clockval'' :  \big( (\loc,\clockval), d, (\loc,\clockval'') \big) \in \transition$
and
$\big( (\loc,\clockval'') , \edge,  (\loc',\clockval') \big) \in \transition$.

Given a TA~$\TA$ with semantics $\semanticsextend$, we refer to the elements of~$\StateSet$ as the \emph{configurations} of~$\TA$.
A (finite) \emph{run} of~$\TA$ is an alternating sequence of configurations of~$\TA$ and pairs of delays and edges starting from the initial configuration $\concstateinit$ and ending in a final configuration (\ie{} whose location is final),
of the form
$(\loci{0}, \clockval_{0}), (d_0, \edge_0), (\loci{1}, \clockval_{1}), \ldots (\loci{n}, \clockval_{n})$ for some $n \in \setN$,
with
$\loci{n} \in \FinalSet$ and for $i = 0, 1, \dots n-1$, $\loci{i} \notin \FinalSet$, $\edge_i \in \EdgeSet$, $d_i \in \setRgeqzero, $ and
$(\loci{i}, \clockval_{i}) \longuefleche{(d_i, \edge_i)} (\loci{i+1}, \clockval_{i+1})$.
We denote by $\AllRuns$ the set of runs of~$\TA$.
A \emph{path} of~$\TA$ is a prefix of a run ending with a configuration.
\subsection{Region Automaton}\label{section:tools:RA}

We recall that the region automaton is obtained by quotienting the set of clock valuations out by an equivalence relation $\simeq$ recalled below.
Given a TA~$\TA$ and its set of clocks $\ClockSet$, we define $\LargestConstant : \ClockSet \rightarrow \setN$ the map that associates to a clock $\clock$ the greatest value to which the interpretations of $\clock$ are compared within the guards and invariants; if $\clock$ appears in no constraint, we set $\LargestConstant(\clock) = 0$.

Given $\alpha \in \setR$, we write $\intpart{\alpha}$ and $\fract{\alpha}$ for the integral and fractional parts of $\alpha$, respectively.

\begin{defi}[Equivalence relation $\simeq$ for valuations~\cite{AD94}]
Let $\clockval$, $\clockval'$ be two clock valuations. %
We say that $\clockval$ and $\clockval'$ are equivalent, denoted by $\clockval \simeq \clockval'$ when, for each $\clock \in \ClockSet$, either $	\clockval (\clock) > \LargestConstant(\clock)$ and $\clockval' (\clock) > \LargestConstant(\clock)$ or the three following conditions hold:

\begin{enumerate}
\item $\intpart{\clockval (\clock)} = \intpart{\clockval' (\clock)}$;
\item $\fract{\clockval(\clock)} = 0 \text{ if and only if }  \fract{\clockval'(\clock)} = 0$;
\item for each $\clocky \in \ClockSet$, $\fract{\clockval(\clock)} \leq  \fract{\clockval(\clocky)} \text{ if and only if } \fract{\clockval'(\clock)} \leq  \fract{\clockval'(\clocky)}$.
\end{enumerate}
\end{defi}

The equivalence relation is extended to the configurations of~$\TA$: let $\concstate = (\loc, \clockval)$ and $\concstate' = (\loc', \clockval')$ be two configurations in~$\TA$, then
\(\concstate \simeq \concstate' \text{ if and only if } \loc = \loc' \text{ and } \clockval \simeq \clockval'. \)

The equivalence class of a valuation $\clockval$ is denoted $[ \clockval ]$ and is called a \emph{clock region}, and the equivalence class of a configuration $\concstate = (\loc, \clockval)$ is denoted $[ \concstate ]$ and called a \emph{region} of~$\TA$.
In other words, a region is a pair made of a location and of a clock region.
Clock regions are denoted by the enumeration of the constraints defining the equivalence class. Thus, values of a clock $\clock$ that go beyond $\LargestConstant(\clock)$ are merged and described in the regions by the inequality ``$\clock > \LargestConstant(\clock)$''.

The set of regions of~$\TA$ is denoted by~$\Regions{\TA}$.
These regions are of finite number: this allows us to construct a finite ``untimed'' regular automaton, the \emph{region automaton}~$\RegionAutomaton{\TA}$.
Locations of $\RegionAutomaton{\TA}$ are regions of~$\TA$, and the transitions of $\RegionAutomaton{\TA}$ convey the reachable valuations associated with each configuration in~$\TA$.

To formalize the construction, we need to transform discrete and time-elapsing transitions of~$\TA$ into transitions between the regions of~$\TA$. To do that, we define a ``time-successor'' relation that corresponds to time-elapsing transitions.

\begin{defi}[Time-successor relation \cite{ALM23}]
Let $\region = (\loc, [\clockval]), \region' = (\loc', [\clockval']) \in \Regions{\TA}$.
We say that $\region'$ is a time-successor of $\region$ when $\region \neq \region'$, $\loc = \loc'$ and for each configuration $(\loc, \clockval)$ in $\region$, there exists $d \in \setRgeqzero$ such that $(\loc, \clockval + d)$ is in~$\region'$ and for all $d' < d, (\loc, \clockval + d') \in \region \cup \region'$.
\end{defi}

A region $\region = (\loc, [\clockval])$ is \emph{unbounded} when, for all $\clock$ in $\ClockSet$ and all $\clockval' \in [\clockval]$, $\clockval'(\clock) > \LargestConstant(\clock)$.

\begin{defi}[Region automaton~\cite{AD94}\label{definition:region-automaton}]
Given a TA $\TA = \TAprivextend$, the region automaton is the
tuple $\RegionAutomaton{\TA} = (\ActionSet_{\Regions{}}, \Regions{}, \region_0, {\Regions{}}_f, \EdgeSet_{\Regions{}})$ where
\begin{oneenumerate}%
\item $\ActionSet_{\Regions{}} = \ActionSet \cup \lbrace \silentaction \rbrace$;
\item ${\Regions{}} = \Regions{\TA}$;
\item $\region_0 = [\concstateinit]$;
\item ${\Regions{}}_f$ is the set of regions whose first component is a final location $\locfinal \in \FinalSet$;
\item $\EdgeSet_{\Regions{}}$ is defined as follows:
\begin{enumerate}%
	\item \emph{(discrete transitions)} For every $\region = (\loc, [\clockval])$ with $\loc\notin \FinalSet$, $\region' = (\loc', [\clockval']) \in \Regions{\TA}$ and $\action \in \ActionSet \cup \lbrace \silentaction \rbrace$:
\[ (\region, \action, \region') \in \EdgeSet_{\Regions{}} \text{ if } \exists \clockval'' \in [\clockval], \exists \clockval''' \in [\clockval'], \big( (\loc, \clockval'') , \edge , (\loc', \clockval''') \big) \in \transition \]
\noindent{}with $\edge = (\loc, \guard, \action, \resets, \loc') \in \EdgeSet$ for some guard~$\guard$ and some clock set~$\resets$.\\
\item \emph{(delay transitions)}
	For every $\region = (\loc, [\clockval])$ with $\loc\notin \FinalSet$, $\region' \in \Regions{\TA}$:
\[ (\region, \silentaction, \region') \in \EdgeSet_{\Regions{}} \text{ if } \region' \text{ is a time-successor of } \region \text{ or if } \region=\region'\text{ is unbounded.}\]
\end{enumerate}%
\end{oneenumerate}%
\end{defi}

A \emph{run} of~$\RegionAutomaton{\TA}$ is an alternating sequence of regions of~$\RegionAutomaton{\TA}$ and actions starting from the initial region $\regioni{0}$ and ending in a final region, of the form
$\regioni{0}, \action_0, \regioni{1}, \action_1, \ldots \regioni{n-1}, \action_{n-1}, \regioni{n}$ for some $n \in \setN$,
with $\regioni{n} \in {\Regions{}}_f$ and for each $i \in \interval{0}{n-1}$, $\regioni{i} \notin {\Regions{}}_f$, and $(\regioni{i}, \action_i, \regioni{i+1}) \in \EdgeSet_{\Regions{}}$.
A~\emph{path} of $\RegionAutomaton{\TA}$ is a prefix of a run ending with a region;
and the \emph{trace} of a path of~$\RegionAutomaton{\TA}$ is the sequence of actions ($\silentaction$~excluded) contained in this path.

\section{Opacity Problems in Timed Automata}\label{section:opacity:TA}
\subsection{Timed Words, Private and Public Runs}\label{sec:opacity:preliminaries:PrivPubVisit} %

Given a TA~$\TA$ and a run~$\run=(\loci{0}, \clockval_0), (d_0, \edgei{0}), (\loci{1}, \clockval_1), \ldots, (\loci{n}, \clockval_n)$ of~$\TA$, we say that $\PrivSet$ is \emph{visited in~$\run$} if there exists~$m \in \setN$ such that $\loci{m} \in \PrivSet$.
We denote by $\PrivVisit{\TA}$ the set of runs visiting $\PrivSet$, and refer to them as \emph{private} runs.
Conversely, we say that $\PrivSet$ is \emph{avoided in~$\run$}
if the run $\run$ does not visit~$\PrivSet$.
We denote the set of runs avoiding $\PrivSet$ by~$\PubVisit{\TA}$, referring to them as \emph{public} runs.

A timed word is a sequence of pairs made up of an action and a timestamp in~$\setRgeqzero$, with the timestamps being non-decreasing  over the sequence.
We denote by $\TimedWords{\ActionSet}$ the set of all finite timed words over the alphabet~$\ActionSet$.
A run $\run$ of a TA~$\TA$ defines a timed word: if $\run$ is of the form
\((\loci{0}, \clockval_0), (d_0, \edgei{0}), (\loci{1}, \clockval_1), \ldots, (\loci{n}, \clockval_n )\)
where for each  $i \in \interval{0}{n-1}$, $e_i = (\loc_i , g_i, a_i , R_i, \loc_{i+1})$ and $a_i \in \ActionSet \cup \set{\silentaction}$,
then it generates the timed word
$(a_{j_0}, \sum\limits_{i=0}^{j_0}d_i)(a_{j_1}, \sum\limits_{i=0}^{j_1}d_i)\cdots (a_{j_{m}}, \sum\limits_{i=0}^{j_{m}}d_i)$, where $j_0 < j_1 < \dots < j_{m}$ and $\set{j_k \mid k \in \interval{0}{m}} = \set{i \in \interval{0}{n-1} \mid a_i \neq \silentaction}$.
Timed words are therefore elements of $(\ActionSet \times \setRgeqzero)^*$.
We denote by $\Trace{\run}$ and call \emph{trace} of~$\run$ the timed word generated by the run~$\run$ and, by extension, given a set of runs $P \subseteq \AllRuns$, we denote by $\Trace{P}$ the set of the traces of runs in~$P$.

The set of timed words recognized by a TA~$\TA$ is the set of traces generated by its runs,
\ie{} $\Trace{\AllRuns}$, which we also denote $\Trace{\TA}$ and call \emph{language} of~$\TA$.
Similarly, we use $\PrivateTr{\TA} = \Trace{\PrivVisit{\TA}}$ to denote the set of traces of private runs of~$\TA$, and $\PublicTr{\TA} = \Trace{\PubVisit{\TA}}$ for the set of traces of public runs of~$\TA$.

In Cassez's original definition~\cite{Cassez09}, actions were partitioned into two sets,
depending on whether an attacker could observe them or not.
For simplicity, here we replace all unobservable transitions in~$\TA$ by $\silentaction$-transitions; put it differently, all actions in~$\ActionSet$ are observable.
Projecting the sequence of actions in a run onto the observable actions, as
done by Cassez, is equivalent to replacing these actions by $\silentaction$ and taking
the trace of the run.
Therefore, with respect to opacity, our model is equivalent to~\cite{Cassez09}.
\subsection{Defining Timed Opacity}

In this section, a definition of timed opacity based on the one from~\cite{Cassez09} is introduced, with three variants inspired by~\cite{ALLMS23}\footnote{%
	In~\cite{ALLMS23}, these variants are defined in the context of execution-time opacity, different from~\cite{Cassez09}.

}: existential, full and weak opacity---defined in the following.
If the attacker observes a set of runs of the system (\ie{} observes their associated traces), we do not want them to deduce whether~$\PrivSet$ was visited or not during these observed runs.
Opacity holds when the traces can be produced by both private and public runs.

Thus, we are first interested in the \emph{existence} of an opaque trace produced by the TA, that is, a trace that cannot allow the attacker to decide whether it was generated by a private or a public run.
$\exists$-opacity, which can be seen as the weakest form of opacity, is useful to check if there is at least one opaque trace; if not, the system cannot be made opaque by restraining the behaviours.%

\begin{defi}[$\exists$-opacity]\label{def:exists-opacity}
	A TA~$\TA$ is \emph{$\exists$-opaque}
	if $\PrivateTr{\TA} \cap \PublicTr{\TA} \neq \emptyset$.
\end{defi}
\defProblem{$\exists$-opacity decision problem}
{A TA~$\TA$}
{Is $\TA$ $\exists$-opaque?}

Ideally and for a stronger security of the system, one can ask the system to be opaque \emph{for all} possible traces of the system: a TA~$\TA$ is fully opaque whenever for any trace
in~$\Trace{\TA}$, it is not possible to deduce whether the run that generated this trace
visited~$\PrivSet$ or not.
Sometimes, a weaker notion is sufficient to ensure the required security in the system, \ie{}
when the compromising information solely comes from the identification of the private runs.

\begin{defi}[Full and weak opacity]\label{def:opacity}
	A TA~$\TA$ is \emph{fully opaque} if $\PrivateTr{\TA} = \PublicTr{\TA}$.
	A TA~$\TA$ is \emph{weakly opaque} if $\PrivateTr{\TA} \subseteq \PublicTr{\TA}$.
\end{defi}
\defProblem{Full (resp.\ weak) opacity decision problem}
{A TA~$\TA$}
{Is $\TA$ fully (resp.\ weakly) opaque?}
\begin{exa}
	Consider the TA~$\TA$ depicted in \cref{figure:example-TA}.
	We have:
\begin{align*}
\PrivateTr{\TA} &= \big\{ (a,\tau_1) \cdots (a,\tau_n) (b,\tau_{n+1}) \mid n \in \setN \land \forall i \in [\![1,n]\!], \tau_i \leq \tau_{i+1} \leq 2 \land \tau_{n+1} \geq 1 \big\} \\
\PublicTr{\TA} &= \big\{ (a,\tau_1) \cdots (a,\tau_n) (b,\tau_{n+1}) \mid n \in \setN \land \forall i \in [\![1,n]\!], \tau_i \leq \tau_{i+1} \leq 3 \big\}
\end{align*}
	First note that $\PrivateTr{\TA} \cap \PublicTr{\TA} \neq \emptyset$ since for example $(b, 1.5) \in \PrivateTr{\TA} \cap \PublicTr{\TA}$:
	therefore, $\TA$ is $\exists$-opaque.

	In addition, $\PrivateTr{\TA} \subseteq \PublicTr{\TA}$, hence $\TA$ is weakly opaque.

	However, $\PrivateTr{\TA} \neq \PublicTr{\TA}$ (for example $(b, 2.5) \in \PublicTr{\TA} \setminus \PrivateTr{\TA}$),
	and therefore $\TA$ is not fully opaque.
\end{exa}
\section{Reductions Between Opacity Notions and Classical Problems}\label{section:inter-reduc}

In this section, we position full, weak and $\exists$-opacity with respect to more classical problems of~TAs.
More precisely, we establish the following results:
\begin{itemize}
	\item \cref{sec:reducWFopac}: Weak and full opacity are inter-reducible;
	\item \cref{sec:reducRexist}: $\exists$-opacity and reachability in TA are inter-reducible; and
	\item \cref{sec:reducWincl}: Weak opacity and language inclusion of TA are inter-reducible.
\end{itemize}
All those inter-reductions are polynomial, hence establishing that the problems lie in general in the same complexity class.

\subsection{Inter-reducibility of Weak and Full Opacity}\label{sec:reducWFopac}

In order to relate weak and full opacity, we first introduce a construction separating private and public runs---that will also be useful to prove our subsequent results in \cref{section:opacity:TA:res,section:finite}.

\subsubsection*{$\APriv$ and $\APub$}%

First, we need a construction of two TAs called $\APub$ and~$\APriv$ (illustrated in \cref{figure:APub,figure:APriv}), that recognize timed words produced respectively by public and private runs of a given TA~$\TA$.

\paragraph{The public runs TA $\APub$}
The public runs TA $\APub$ is the easiest to build: it suffices to remove the private locations from~$\TA$ to eliminate every private run in the system.

\begin{defi}[Public runs automaton $\APub$]\label{definition:APub}

Let $\TA = \TAprivextend$ be a TA. We define the public runs TA \mbox{$\APub = (\ActionSet, \LocSet \setminus \PrivSet, \emptyset, \FinalSet\setminus \PrivSet, \ClockSet, \invariant', \EdgeSet')$} with $\invariant'$ and $\EdgeSet'$ precised as follows:
	\begin{enumerate}
		\item
$\invariant'$ is the restriction of $\invariant$ to the set of locations of $\APub$: for all $\loc \in \LocSet \setminus \PrivSet$, $\invariant'$ satisfies that $\invariant'(\loc) = \invariant(\loc)$;
		\item $\EdgeSet' = \EdgeSet \setminus \lbrace (\loc, g, a, R, \loc') \in \EdgeSet \mid \loc \in \PrivSet \lor \loc' \in \PrivSet \rbrace$ is the remaining set of transitions when private locations are removed from $\LocSet$.
	\end{enumerate}

\end{defi}

\paragraph{The private runs TA~$\APriv$}
The private runs TA~$\APriv$ is obtained by duplicating all locations and transitions of~$\TA$: one copy~$\TA_{S}$ corresponds to the paths that already visited the private locations set, and the other copy $\TA_{\bar{S}}$ corresponds to the paths that did not (this is a usual way to encode a Boolean, here ``$\PrivSet$ was visited'', in the locations of a~TA).
For each private location $\locpriv$ in~$\TA$, we redirect all transitions leading to the copy of
$\locpriv$ in~$\TA_{\bar{S}}$ towards the copy of~$\locpriv$ in~$\TA_{S}$.
The initial location is the one from~$\TA_{\bar{S}}$ and the final locations are the ones from~$\TA_{S}$.
Hence all runs need to go from $\TA_{\bar{S}}$ to $\TA_{S}$ before reaching a final location---which requires visiting a private location.
Formally:

\begin{defi}[Private runs TA $\APriv$]

Let $\TA = \TAprivextend$ be a TA.
	The private runs TA \mbox{$\APriv = (\ActionSet, \LocSet_S \uplus \LocSet_{\bar{S}}, \loci{0}^{\bar{S}}, \PrivSet^S, \FinalSet^S, \ClockSet, \invariant', \EdgeSet')$} is defined as follows:
	\begin{enumerate}
		\item $\LocSet_S  =  \lbrace \loc^S \mid \loc \in \LocSet \rbrace$ and $\LocSet_{\bar{S}}  =   \lbrace \loc^{\bar{S}} \mid \loc \in \LocSet \rbrace$.
		That is,
			$\LocSet_S$ and $\LocSet_{\bar{S}}$ are two disjoint copies of~$\LocSet$: to each location~$\loc$ of~$\LocSet$, we associate a location denoted $\loc^S$ in~$\LocSet_S$ and a location $\loc^{\bar{S}}$ in~$\LocSet_{\bar{S}}$ (for instance $\loci{0}^{\bar{S}}$ is in $\LocSet_{\bar{S}}$);

		\item $\PrivSet^S= \lbrace \locpriv^S \mid \locpriv \in \PrivSet \rbrace$ is the set of private locations;

		\item $\FinalSet^S = \lbrace \locfinal^S \mid \locfinal \in \FinalSet \rbrace$ is the set of final locations;

		\item $\invariant'$ is defined such as $\invariant'(\loc^S) = \invariant'(\loc^{\bar{S}}) = \invariant(\loc)$;
			and

		\item $\EdgeSet' = \EdgeSet_S \uplus \EdgeSet_{\bar{S}} \uplus \EdgeSet_{\bar{S} \rightarrow S}$ where $\EdgeSet_S$ and $\EdgeSet_{\bar{S}}$ are the two disjoint copies of $\EdgeSet$ respectively associated with the sets of locations $\LocSet_S$ and $\LocSet_{\bar{S}}$, and $\EdgeSet_{\bar{S} \rightarrow S}$ is a copy of the set of all transitions that go toward $\PrivSet^{\bar{S}}$ where the target location $\locpriv^{\bar{S}}$ has been changed into~$\locpriv^{S}$.
		More formally:
		\[\begin{array}{rcl}
		\EdgeSet_S & = & \big\lbrace (\loc^S, g, a, R, \loc'^S) \mid (\loc, g, a, R, \loc') \in \EdgeSet \big\rbrace \\
		\EdgeSet_{\bar{S}} & = &  \big\lbrace (\loc^{\bar{S}}, g, a, R, \loc'^{\bar{S}}) \mid (\loc, g, a, R, \loc') \in \EdgeSet \wedge \loc'\neq \locpriv \big\rbrace \\
		\EdgeSet_{\bar{S} \rightarrow S} & = & \big\lbrace (\loc^{\bar{S}}, g, a, R, \locpriv^S) \mid (\loc, g, a, R, \locpriv) \in \EdgeSet \big\rbrace\text{.}
		\end{array}\]
	\end{enumerate}
\end{defi}

The languages of $\APriv$ and~$\APub$ are respectively $\PrivateTr{\TA}$ and $\PublicTr{\TA}$.

\begin{figure}[tb]
	\centering
	\begin{subfigure}[b]{0.4\textwidth}
   \centering
\begin{tikzpicture}[pta, node distance=2cm, thin]

	\node[location, initial] (l0) at (-1, 1) {$\locinit$};
	\node[location, final] (lfS) at (+1, 1) {$\locfinal$};
	\node[rectangle, minimum width=3cm, minimum height=2cm, align=center, draw] (rectangle) at (0cm, 1cm) {$\TA$};

	\node[location, private] (lpriv) at (0, 0.5) {$\locpriv$};
	\path
	;

\draw[-, very thick, color = black!20] (-0.5,0.2) -- (0.5,0.8);
\draw[-, very thick, color = black!20] (0.5,0.2) -- (-0.5,0.8);

\end{tikzpicture}
  \caption{$\APub$}
  \label{figure:APub}
	\end{subfigure}
	\begin{subfigure}[b]{0.4\textwidth}
   \centering
	\begin{tikzpicture}[pta, node distance=2cm, thin]

	\node[location] (l0S) at (-1, 0) {$\locinit^S$};
	\node[location, final] (lfS) at (+1, 0) {$\locfinal^S$};
	\node[rectangle, minimum width=3cm, minimum height=1.5cm, align=center, draw] (rectangle) at (0cm, -.1cm) {$\TA_S$\\};

	\node[location, initial] (l0S) at (-1, -2) {$\locinit^{\bar{S}}$};
	\node[location] (lfS) at (+1, -2) {$\locfinal^{\bar{S}}$};
	\node[rectangle, minimum width=3cm, minimum height=1.5cm, align=center, draw] (rectangle) at (0cm, -2cm) {$\TA_{\bar{S}}$};

	\node[location, private] (lpriv) at (0, -0.5) {$\locpriv^S$};
	\path
	(-0.5, -1.5) edge (lpriv)
	;
	\end{tikzpicture}
  \caption{$\APriv$}
  \label{figure:APriv}
	\end{subfigure}
	\caption{Illustrating $\APub$ and $\APriv$ %
	}
	\label{fig:two-automata}
\end{figure}
\begin{exa}
	Consider again the TA~$\TA$ from \cref{figure:example-TA}.
	We give $\APub$ and~$\APriv$ in \cref{figure:APriv-APub:example-TA}.
\end{exa}
\begin{figure}[tb]

	\centering
	\begin{subfigure}[b]{0.3\textwidth}
		\begin{tikzpicture}[pta, xscale=1, yscale=.5]

		\location[initial,name=s0,at={(0, -1)}]{$\styleclock{\clock} \leq 3$}{$\loci{0}$};

		\node[location, final] at (1.5, -1) (s1) {$\loci{1}$};

		\path (s0) edge[loop above] node[]{$\styleact{a}$} (s0);
		\path (s0) edge[] node[below]{$\styleact{b}$} (s1);

		\end{tikzpicture}
		\caption{$\APub$}
		\label{fig:example-APub}
	\end{subfigure}
	\hfill
	\begin{subfigure}[b]{0.65\textwidth}
		\begin{tikzpicture}[pta, scale=2, xscale=1, yscale=.45]
		\node[rectangle, minimum width=6.25cm, minimum height=2.2cm, align=center, draw, dashed] (rectangleS) at (1.25cm, 2.25cm) {};

		\location[initial,name=s0S,at={(0, 1.75)}]{$\styleclock{\clock} \leq 3$}{$\loci{0}$};

		\location[private,name=s2S,at={(1.5, 2.75)}]{$\styleclock{\clock} \leq 2$}{$\loci{2}$};

		\node[location, final] at (2.5, 1.75) (s1S) {$\loci{1}$};

		\node[rectangle, minimum width=6.25cm, minimum height=2.2cm, align=center, draw, dashed] (rectangleSbar) at (1.25cm, -0.5cm) {};

		\location[initial,name=s0,at={(0, -1)}]{$\styleclock{\clock} \leq 3$}{$\loci{0}$};

		\location[name=s2,at={(1.5, 0)}]{$\styleclock{\clock} \leq 2$}{$\loci{2}$};

		\node[location] at (2.5, -1) (s1) {$\loci{1}$};

		\path (s0S) edge node[align=center, pos=0.6]{$\styleclock{\clock} \geq 1$\\$\styleact{\silentaction}$} (s2S);
		\draw (s0S) .. controls +(100:40pt) and +(130:15pt) .. (s0S) node[pos=0.5, above]{$\styleact{a}$};
		\path (s0S) edge[] node[below]{$\styleact{b}$} (s1S);
		\path (s2S) edge node[]{$\styleact{b}$} (s1S);
		\draw (s0) .. controls +(100:40pt) and +(130:15pt) .. (s0) node[pos=0.5, above]{$\styleact{a}$};
		\path (s0) edge[] node[below]{$\styleact{b}$} (s1);
		\path (s2) edge node[]{$\styleact{b}$} (s1);
		\path (s0) edge node[align=center, below, yshift=-1em]{$\styleclock{\clock} \geq 1$\\$\styleact{\silentaction}$} (s2S);
		\end{tikzpicture}
		\caption{$\APriv$}
		\label{fig:example-APriv}
	\end{subfigure}
\caption{$\APub$ and $\APriv$ with the example from \cref{figure:example-TA}}
\label{figure:APriv-APub:example-TA}
\end{figure}

By a minor modification on~$\APriv$, one can build a TA~$\AMemo$ that recognizes exactly the same language as $\TA$ and that stores in each location whether the private locations set has been visited.

\begin{defi}\label{definition:AMemo}
Given a TA~$\TA$, we let~$\AMemo$ be the copy of~$\APriv$ to which we add the set $\lbrace \locfinal^{\bar{S}} \mid \locfinal \in \FinalSet \rbrace$ to the set of final locations.
\end{defi}

Notably, there is a one-to-one trace-conserving correspondence between the public (resp.\ private) runs of~$\TA$ and the public (resp.\ private) runs of~$\AMemo$.
This implies that $\TA$ is weakly (resp.\ fully) opaque if and only if $\AMemo$ is weakly (resp.\ fully) opaque.

\subsubsection*{Inter-reducibility Proof}%

While the distinction between weak and full notions of opacity can lead to meaningful changes in the context of execution-time opacity~\cite{ALLMS23}, within our framework both associated problems are inter-reducible.

\begin{thm}\label{th:weak-full-eq}
The weak opacity decision problem and the full opacity decision problem are inter-reducible.
\end{thm}
\begin{proof}
Let us first show that the full opacity decision problem reduces to the weak opacity decision problem. Let $\TA$ be a TA.
In order to test whether $\TA$ is fully opaque, we can test both inclusions: $\PrivateTr{\TA} \subseteq \PublicTr{\TA}$ and $\PrivateTr{\TA} \supseteq \PublicTr{\TA}$.
The first inclusion can be decided directly by testing whether $\TA$ is weakly opaque.
In order to test the second inclusion, we need to build a TA $\TB$ where private and public runs are swapped.
To do so, we first build $\APub$ and~$\APriv$, and we then define $\TB$ (see \cref{figure:weak-to-full}) as the TA made of $\APub$ and~$\APriv$ as well as two new locations $\locinit'$ and~$\locpriv'$.
The location $\locinit'$ is the initial location of~$\TB$ and $\locpriv'$ is the only private location.
Both $\locinit'$ and $\locpriv'$ have the invariant $\clock = 0$ (picking any $\clock \in \ClockSet$), ensuring no time may elapse in those locations.
From~$\locinit'$, with a transition labelled by $\silentaction$, one may reach either the initial location of~$\APriv$ ($\loci{0}^S$) or~$\locpriv'$, from which an $\silentaction$-transition leads to the initial location of~$\APub$ ($\loci{0}$%
).
The final locations of~$\TB$ are the final locations of~$\APub$ and~$\APriv$.
The public runs of~$\TB$ are the ones starting in~$\locinit'$, going immediately to~$\loci{0}^S$, and then following a run of~$\APriv$ until a final location of~$\APriv$ is reached.
As the initial transition is labelled by~$\silentaction$ in 0-time, we have~$\PublicTr{\TB}=\PrivateTr{\TA}$.
Similarly, the private runs of~$\TB$ are the ones starting in~$\locinit'$, going immediately to~$\locpriv'$ then immediately to~$\locinit$, and then following a run of~$\APub$ until a final location of~$\APub$ is reached.
As the two initial transitions are labelled by~$\silentaction$ and occur in 0-time due to the invariants, we have~$\PrivateTr{\TB}=\PublicTr{\TA}$.
Hence, $\TA$ is fully opaque if and only if~$\TA$ and~$\TB$ are weakly opaque.

Let us now show the converse reduction.
Let $\TA$ be a~TA.
We will define a TA~$\TB$ such that $\TB$ is fully opaque if and only if $\TA$ is weakly opaque.
To do so, we want that $\PublicTr{\TB}=\PublicTr{\TA}$ and~$\PrivateTr{\TB}=\PublicTr{\TA}\cup \PrivateTr{\TA}$. Indeed, if these equalities hold, $\PublicTr{\TB}=\PrivateTr{\TB}$ would be equivalent to $\PublicTr{\TA}=\PublicTr{\TA}\cup \PrivateTr{\TA}$---which holds if and only if $\PrivateTr{\TA}\subseteq \PublicTr{\TA}$.
We build $\TB$ (see \cref{figure:full-to-weak}) using the same encoding as in the first reduction, except that $\locinit'$ is now connected to $\locinit$ (\ie{} the initial location of~$\APub$) and to~$\locpriv'$, while $\locpriv'$ is now connected to both $\loci{0}^S$ and~$\locinit$.
The public runs of $\TB$ are the ones starting in~$\locinit'$, going immediately to $\loci{0}$, and then following a run of $\APub$ until a final location of~$\APub$ is reached.
As the initial transition is labelled by~$\silentaction$ and must be taken in 0-time, we have~$\PublicTr{\TB}=\PublicTr{\TA}$.
Similarly, the private runs of $\TB$ are the ones starting in $\locinit'$, going immediately to~$\locpriv'$ and then either immediately to~$\loci{0}^S$ followed by a run of~$\APriv$, or to $\loci{0}$~followed by a run of~$\APub$.
As the two initial transitions are labelled by $\silentaction$ and must be taken in 0-time, we have $\PrivateTr{\TB}=\PrivateTr{\TA}\cup\PublicTr{\TA}$.
Hence, $\TA$ is weakly opaque if and only if $\TB$ is fully opaque.
\begin{figure}[tb]
	\centering
	\begin{subfigure}[b]{0.48\textwidth} %
   \centering
\scalebox{0.9}{
\begin{tikzpicture}[pta, node distance=2cm, thin]

	\location[initial,name=l0',at={(-5, -1)}]{$\styleclock{\clock} = 0$}{$\loci{0}'$};

	\location[private,name=lpriv',at={(-3, -1.5)}]{$\styleclock{\clock} = 0$}{$\locpriv'$};

	\node[location] (l0S) at (-1, 0) {$\locinit^S$};
	\node[location, final] (lfS) at (+1, 0) {$\locfinal^S$};
	\node[rectangle, minimum width=3cm, minimum height=1.5cm, align=center, draw] (rectangle) at (0cm, -.1cm) {$\APriv$};

	\node[location] (l0) at (-1, -2) {$\locinit$};
	\node[location, final] (lf) at (+1, -2) {$\locfinal$};
	\node[rectangle, minimum width=3cm, minimum height=1.5cm, align=center, draw] (rectangle) at (0cm, -2cm) {$\APub$};

	\path
	(l0') edge node[below]{$\styleact{\silentaction}$}(lpriv')
	(l0') edge node{$\styleact{\silentaction}$}(l0S)
	(lpriv') edge node[below]{$\styleact{\silentaction}$}(l0)
	;

\end{tikzpicture}}
  \caption{$\TB$ such that $\TA$ is fully opaque iff $\TA$ and $\TB$ are both weakly opaque}
  \label{figure:weak-to-full}
	\end{subfigure}
	\hfill
	\begin{subfigure}[b]{0.48\textwidth} %
   \centering
   \scalebox{0.9}{
	\begin{tikzpicture}[pta, node distance=2cm, thin]

	\location[initial,name=l0',at={(-4.5, -1)}]{$\styleclock{\clock} = 0$}{$\loci{0}'$};

	\location[private,name=lpriv',at={(-2.75, -0.5)}]{$\styleclock{\clock} = 0$}{$\locpriv'$};

	\node[location] (l0S) at (-1, 0) {$\locinit^S$};
	\node[location, final] (lfS) at (+1, 0) {$\locfinal^S$};
	\node[rectangle, minimum width=3cm, minimum height=1.5cm, align=center, draw] (rectangle) at (0cm, -.1cm) {$\APriv$};

	\node[location] (l0) at (-1, -2) {$\locinit$};
	\node[location, final] (lf) at (+1, -2) {$\locfinal$};
	\node[rectangle, minimum width=3cm, minimum height=1.5cm, align=center, draw] (rectangle) at (0cm, -2cm) {$\APub$};

	\path
	(l0') edge node[above]{$\styleact{\silentaction}$}(lpriv')
	(l0') edge node[below]{$\styleact{\silentaction}$}(l0)
	(lpriv') edge node[left]{$\styleact{\silentaction}$}(l0)
	(lpriv') edge node[above left]{$\styleact{\silentaction}$}(l0S)
	;
	\end{tikzpicture}}
  \caption{$\TB$ such that $\TA$ is weakly opaque iff $\TB$ is fully opaque}
  \label{figure:full-to-weak}
	\end{subfigure}
	\caption{Two TAs for the inter-reducibility of weak and full opacity}
	\label{fig:inter-reducibility}
\end{figure}
\end{proof}
\subsection{Inter-reducibility of Reachability and \texorpdfstring{$\exists$}{∃}-opacity}\label{sec:reducRexist}

We show here that, the $\exists$-opacity problem is inter-reducible with the reachability problem. As the latter is known to be \PSPACE{}-complete~\cite{AD94}, even for TAs with only two clocks~\cite{FJ15}, this establishes that the \PSPACE{}-completeness of the $\exists$-opacity problem.

\begin{thm}\label{theorem:decidability-exists-opacity}
The $\exists$-opacity problem is inter-reducible with the reachability problem for TAs.
\end{thm}

\begin{proof}
Let us first reduce the $\exists$-opacity problem to reachability.
Let $\TA$ be a~TA.
We build $\APriv$ and $\APub$ from~$\TA$ as described in \cref{sec:reducWFopac}.
Noting that the product of two TAs recognizes the intersection of their languages~\cite[Theorem~3.15]{AD94} (assuming the two TAs share no clock), we build
the TA $\APriv \times \APub$, product of $\APriv$ and~$\APub$,
whose language is $\PrivateTr{\TA} \cap \PublicTr{\TA}$.
To build this product, we can rename all clocks from $\APub$ so that $\APriv$ and $\APub$ share no clock.

The $\exists$-opacity problem is by definition the non-emptiness of the intersection of $\PrivateTr{\TA}$ and $\PublicTr{\TA}$.
Moreover, the reachability of a final location of $\APriv \times \APub$ is equivalent to the non-emptiness of the language of $\APriv \times \APub$, and thus of the set
$\PrivateTr{\TA} \cap \PublicTr{\TA}$.

Conversely, let us reduce the reachability problem for TAs to the $\exists$-opacity problem.
Let $\TA = (\ActionSet, \LocSet, \locinit, \emptyset, \FinalSet, \ClockSet, I, E)$ be a TA.
We suppose that $\ClockSet$ is not empty and define (see \cref{figure:exists-opacity:PSPACE-hardness})
$\TA' = (\ActionSet, \LocSet \cup \setsmall{\locinit', \loci{1}', \locfinal'}, \locinit', \FinalSet, \FinalSet \cup \setsmall{\locfinal'}, \ClockSet, I', E')$ where $I'$ is an invariant extending $I$ such that $I'(\locinit')=I'(\loci{1}') = I'(\locfinal')=\BTrue{}$ and
$E' = E \cup \big\{(\locinit', \silentaction, (\clock = 0), \emptyset, \locinit),
(\locinit', \silentaction, (\clock = 0), \emptyset, \loci{1}'), (\loci{1}', \silentaction, \BTrue{}, \emptyset, \locfinal')\big\} \cup \big\{(\loci{1}', a, \BTrue{}, \emptyset, \loci{1}') \mid a \in \ActionSet\big\}$ for some $\clock \in \ClockSet$.

\begin{figure}
\begin{center}
\begin{tikzpicture}[pta, node distance=2cm, thin]

	\node[location] (l0) at (-1, 1) {$\locinit$};
	\node[location, final, private] (lf) at (+1, 1) {$\locfinal$};
	\node[rectangle, minimum width=3cm, minimum height=2cm, align=center, draw] (rectangle) at (0cm, 1cm) {$\TA$};

	\node[location, initial] (l0') at (-2.5, -0.25) {$\locinit'$};
	\node[location] (l1') at (-1,-1.5) {$\loci{1}'$};
	\node[location, final] (lf') at (1, -1.5) {$\locfinal'$};
	\path
	(l0') edge node[align=center, pos=0.6]{$\styleclock{\clock} = 0$ \\ $\styleact{\silentaction}$} (l0)
	(l0') edge node[below left, align=center, pos=0.6]{$\styleclock{\clock} = 0$ \\ $\styleact{\silentaction}$} (l1')
	(l1') edge[loop below] node{$\styleact{\ActionSet}$} (l1')
	(l1') edge (lf')
	;
\end{tikzpicture}
\end{center}
  \caption{TA $\TA'$ for the reduction of the reachability problem to $\exists$-opacity}
  \label{figure:exists-opacity:PSPACE-hardness}
\end{figure}

The TA~$\TA'$ is $\exists$-opaque if and only if a final location is reachable in~$\TA$.
Indeed, the set $\PublicTr{\TA'}$ contains all the possible timed traces with the action set $\ActionSet$, and the private runs on $\TA'$ correspond exactly to runs on $\TA$. Hence $\PublicTr{\TA'} \cap \PrivateTr{\TA'} \neq \emptyset$ if and only if $\PrivateTr{\TA'} \neq \emptyset$, \ie{} if there is a run on $\TA$ that reaches a final location.
\end{proof}
\subsection{Inter-reducibility of Weak Opacity and language inclusion}\label{sec:reducWincl}

The timed language inclusion problem asks, given a pair of TAs $\TA$ and $\TB$, whether
$\Trace{\TA}\subseteq \Trace{\TB}$.
We now relate this problem to weak opacity.

\begin{thm}\label{lem:language-inclusion}
The timed language inclusion problem and the weak opacity problem are inter-reducible.
\end{thm}

\begin{proof}
The easiest direction is the reduction from weak opacity to language inclusion since a TA $\TA$ is weakly opaque if and only if the inclusion $\Trace{\APriv} \subseteq \Trace{\APub}$ holds.

Let us now show the opposite reduction.
Let $\TA$ and $\TB$ be two TAs. We build a new TA $\mathcal{O}$ such that $\mathcal{O}$ is weakly opaque if and only if $\Trace{\TA} \subseteq \Trace{\TB}$. As represented in \cref{fig:language-inclusion}, $\mathcal{O}$ consists in a copy of $\TA$ and~$\TB$ with two additional locations: a new initial location $\locinit$ and a location $\locpriv$ which is selected as the unique private location.
From $\locinit$, one can either go in 0 time and with an unobservable action to the initial
location of $\TB$ or to $\locpriv$, and from the latter, one can go (again in 0 time and with an unobservable action) to the initial location of~$\TA$.
As a consequence, 
$\PrivateTr{\mathcal{O}} = \Trace{\TA}$ and $\PublicTr{\mathcal{O}} = \Trace{\TB}$. It follows that $\Trace{\TA} \subseteq \Trace{\TB}$ if and only if $\mathcal{O}$ is weakly opaque.
\begin{figure}[tb]
\begin{tikzpicture}[pta, node distance=2cm, thin]

	\node[location] (l0A) at (-1, 0) {$\locinit^{\TA}$};
	\node[location, final] (lfA) at (+1, 0) {$\locfinal^{\TA}$};
	\node[rectangle, minimum width=3cm, minimum height=1.5cm, align=center, draw] (rectangle) at (0cm, 0cm) {$\TA$};

	\node[location] (l0B) at (-1, -2) {$\locinit^{\TB}$};
	\node[location, final] (lfB) at (+1, -2) {$\locfinal^{\TB}$};
	\node[rectangle, minimum width=3cm, minimum height=1.5cm, align=center, draw] (rectangle) at (0cm, -2cm) {$\TB$};

	\node[location, private] (lpriv) at (-3.5, 0) {$\locpriv$};
	\node[location, initial] (l0) at (-5,-1) {$\locinit$}; 
	\path
		(l0) edge node[align=center]{$\styleclock{x}=0$\\$\styleact{\silentaction}$} (lpriv)
		(l0) edge node[align=center, below left]{$\styleclock{x}=0$\\$\styleact{\silentaction}$} (l0B)
		(lpriv) edge node[align=center]{$\styleclock{x}=0$\\$\styleact{\silentaction}$} (l0A);
	\end{tikzpicture}
  \caption{The TA $\mathcal{O}$ of the reduction from language inclusion to weak opacity}
  \label{fig:language-inclusion}
\end{figure}
\end{proof}

\section{Opacity Problems for Subclasses of Timed Automata}\label{section:opacity:TA:res}

In this section, we consider the decidability status and complexities of the three opacity problems presented in \cref{section:opacity:TA} for several subclasses of TAs.
We first discuss the exact complexity of the $\exists$-opacity problem for TAs with one clock (\cref{section:opacity:exists}).
We then study full and weak opacity for the following subclasses:
	TAs with a single action (\cref{section:opacity:one-action}),
	TAs with a single clock (\cref{section:opacity:one-clock}),
	TAs under discrete time (\cref{section:opacity:discrete-time}),
	and
	observable ERAs (\cref{section:opacity:era}).
We prove each result either for weak opacity or for full opacity; thanks to \cref{th:weak-full-eq}, all such results immediately apply to both full and weak opacity.

We summarize these upcoming results in \cref{table-summary}.

\begin{table}[tb]
	\centering
	\caption{Summary of the results in \cref{section:opacity:TA:res} (``\cellDecidable{}'' denotes decidability, ``\cellUndecidable{}'' denotes undecidability)}
	\label{table-summary}
	\setlength{\tabcolsep}{2pt} %
	\begin{tabular}{| l | c | c | c |}
		\hline
		\rowHeader{} Subclass & $\exists$-opacity & weak opacity & full opacity \\
		\hline
		\cellHeader{}$|\ActionSet| = 1$ & \cellcolor{greenColorBlind!75} &
		\multicolumn{2}{c |}{\cellUndecidable{}\cref{th:undecidability:opacity:one-action}}\\
		\cline{1-1}\cline{3-4}
		\cellHeader{}$|\ClockSet| = 1$ without $\silentaction$-transitions & \cellcolor{greenColorBlind!75} & \multicolumn{2}{c |}{\cellDecidable{}\cref{theorem:one-clock-noepsilon} (non-primitive recursive-c)} \\
		\cline{1-1}\cline{3-4}
		\cellHeader{}$|\ClockSet| = 1$ & \cellDecidable{}\cref{theorem:decidability-exists-opacity} & \multicolumn{2}{c |}{\cellUndecidable{}\cref{th:undecidability:full-opacity:one-clock-epsilon}}\\
		\cline{1-1}\cline{3-4}
		\cellHeader{}$|\ClockSet| = 2$ &
		(\PSPACE{}-c)\cellcolor{greenColorBlind!75}
		& \multicolumn{2}{c |}{\cellUndecidable{}\cref{th:undecidability:full-opacity:two-clocks}} \\
		\cline{1-1}\cline{3-4}
		\cellHeader{}$\Time = \setN$ & \cellcolor{greenColorBlind!75}  & \multicolumn{2}{c |}{\cellDecidable{}\cref{theorem:discrete-time} (\EXPSPACE{}-c)}  \\
		\cline{1-1}\cline{3-4}
		\cellHeader{}oERAs & \cellcolor{greenColorBlind!75} & \multicolumn{2}{c |}{\cellDecidable{} \cref{th:decidability:oera} (\PSPACE{}-c) } \\
		\hline
	\end{tabular}

\end{table}

\subsection{\texorpdfstring{$\exists$}{∃}-opacity Problem for TAs with a single clock} \label{section:opacity:exists}

We showed in \cref{sec:reducRexist} that the $\exists$-opacity problem for TAs
was inter-reducible with the reachability problem and thus \PSPACE{}-complete.
It is interesting to note that the reduction from the $\exists$-opacity problem to reachability doubles the number of clocks, as it relies on the product of two TAs, while the opposite reduction conserves the number of clocks.
This observation is important as the reachability problem
is \NLOGSPACE{}-complete~\cite{LMS04} for TAs with one clock; the \PSPACE{}-hardness of reachability applies only from 2 clocks onward~\cite{FJ15}.
Thus, for TAs with one clock, \cref{theorem:decidability-exists-opacity} provides a 
\PSPACE{} algorithm for opacity, but only an \NLOGSPACE{} complexity lower bound.
We show now that $\exists$-opacity remains \PSPACE{}-complete in this setting.

\begin{thm}
The $\exists$-opacity problems for TAs with one clock is \PSPACE{}-complete.
\end{thm}
\begin{proof}
The inter-reduction of \cref{theorem:decidability-exists-opacity}
provides a \PSPACE{} algorithm for TAs in general, so we focus on the hardness here.

\newcommand{\TAx}{\ensuremath{\TA_{\clockx}}}
\newcommand{\TAy}{\ensuremath{\TA_{\clocky}}}
\newcommand{\TAz}{\ensuremath{\TA_{\clockz}}}
\newcommand{\TAxy}{\ensuremath{\TA_{\clockx,\clocky}}}

We reduce the reachability problem for TAs with two clocks, which is \PSPACE{}-hard~\cite{FJ15}.
Let $\TAxy = (\ActionSet, \LocSet, \locinit, \emptyset, \FinalSet, \{ \clockx, \clocky \}, I, \EdgeSet) $ be a TA with clocks $\clockx$ and~$\clocky$.
First, we relabel every transition (including silent transitions) of~$\TAxy$ with a new alphabet $\ActionSet' = \{a_i \mid 1 \leq i \leq |\EdgeSet| \}$ such that each letter of $\ActionSet'$ labels exactly one transition of~$\TAxy$.
We denote the obtained automaton by~$\TAxy'$.

Given a guard $g$, we define $g_{\clockx}$ and $g_{\clocky}$ as respectively the constraints in $g$ over $\clockx$ and~$\clocky$.
Hence, $g = g_{\clockx} \wedge g_{\clocky}$.
For $\clockz \in \{\clockx,\clocky \}$, we then define the TA
$\TAz = (\ActionSet, \LocSet, \locinit, \emptyset, \FinalSet, \{\clockz \}, I_{\clockz}, E_{\clockz})$ with
$E_{\clockz} = \{(\loc, a, g_{\clockz}, R \cap \{\clockz\}, \loc') \mid (\loc, a, g, R, \loc') \in \EdgeSet \}$
and $I_{\clockz}$ is similarly obtained by only keeping the $\clockz$ part of the invariant.

We have that a word accepted by $\TAxy'$ is also accepted by $\TAx'$ and by~$\TAy'$, as each of those TAs have less constraints. Moreover, if a word is accepted by $\TAx'$ and by~$\TAy'$, as the corresponding run is entirely characterized by its trace (since each transition has its own label) and satisfied the constraints on both clocks, then it is accepted by~$\TAxy'$.

Now, we build the TA~$\TB$ (see example in \cref{fig:exists-one-clock}) over the single clock~$\clockx$ as the classical union construction of $\TAx'$ and~$\TAy'$, and set as private locations the final locations of~$\TAx'$
(Recall that, by definition of union, the final locations of~$\TB$ are the final locations of~$\TAx'$ and~$\TAy'$.)
More precisely,
we add a new initial location $\locinit'$ from which one can reach the initial locations of $\TAx'$ and~$\TAy'$ via a transition labelled by a new letter~$\sharp$ and with the guard ``$\clockx=0$''.
Moreover, we relabel every occurrence of~$\clocky$ in the copy of $\TAy'$ into~$\clockx$ so that $\TB$ remains a one-clock~TA.

As the runs of $\TAx'$ (resp.~$\TAy'$) provide the private (resp.\ public) runs of~$\TB$,
$\TB$ is $\exists$-opaque if and only if there is a pair of runs of same trace accepted by
$\TAx'$ and $\TAy'$, thus a word accepted by $\TAxy'$ or equivalently a reachable final location in~$\TAxy$.
Moreover, $\TB$ is polynomial in the size of~$\TAxy$.
Therefore the $\exists$-opacity problem in one-clock automata is \PSPACE{}-hard. \qedhere

\begin{figure}[tb]
	\centering
	\begin{subfigure}[b]{0.38\textwidth}
   \centering
\begin{tikzpicture}[pta, node distance=2cm, thin]
	
	\node[location, initial] (l0) at (-5, 0) {$\loci{0}$};

	\location[name=l1,at={(-2.5, 0)}]{$\styleclock{\clock} \leq 2$}{$\loci{1}$};

	\node[location, final] (lf) at (0, 0) {$\locfinal$};

	\path
	(l0) edge node[above]{$\styleclock{\clock} \leq 1$} node[align=center,below]{$\styleact{a}$\\$\styleclock{y} \assign 0$} (l1)
	(l1) edge[loop above] node[align=center]{$\styleclock{\clock}>0$\\$\styleact{\silentaction}$\\$\styleclock{\clock} \assign 0$}(l1)
	(l1) edge node[above, align=center]{$\styleclock{y}=2$\\$\land \styleclock{\clock} > 1$} node[align=center, below]{$\styleact{a}$}(lf)
	;

\end{tikzpicture}
  \caption{A TA $\TA_{x,y}$}
	\end{subfigure}
	\hfill
	\begin{subfigure}[b]{0.55\textwidth} %
   \centering
	\begin{tikzpicture}[pta, node distance=2cm, thin]

	\node[location, initial] (l0') at (-5, 0) {$\loci{0}'$};
	\node[location] (l0x) at (-4, 1) {$\loci{0}^x$};
	\node[location] (l0y) at (-4, -1) {$\loci{0}^y$};

	\location[name=l1x,at={(-1.5, 1)}]{$\styleclock{\clock} \leq 2$}{$\loci{1}^x$};

	\node[location] (l1y) at (-1.5, -1) {$\loci{1}^y$};
	\node[location, final, private] (lfx) at (1, 1) {$\locfinal^x$};
	\node[location, final] (lfy) at (1, -1) {$\locfinal^y$};

	\node[rectangle, minimum width=6cm, minimum height=3.25cm, align=center, draw, dashed] (rectangle) at (-1.45, 1.75) {};
	\node at (2, 1.5) {$\TA_x$};
	\node[rectangle, minimum width=6cm, minimum height=2.25cm, align=center, draw, dashed] (rectangle) at (-1.45,-1) {};
	\node at (2,-1) {$\TA_y$};
	
	\path
	(l0') edge node[align=center]{$\styleclock{\clock} = 0$\\$\styleact{\sharp}$} (l0x)
	(l0') edge node[align=center, below left]{$\styleclock{\clock} = 0$\\$\styleact{\sharp}$} (l0y)
	(l0x) edge node[above]{$\styleclock{\clock} \leq 1$} node[below]{$\styleact{a_1}$} (l1x)
	(l1x) edge[loop above] node[align=center]{$\styleclock{\clock}>0$\\$\styleact{a_2}$\\$\styleclock{\clock} \assign 0$}(l1x)
	(l1x) edge node[above]{$\styleclock{\clock} > 1$} node[below]{$\styleact{a_3}$}(lfx)
	(l0y) edge node[align=center,below]{$\styleact{a_1}$\\$\styleclock{\clock} \assign 0$} (l1y)
	(l1y) edge[loop above] node[align=center]{$\styleact{a_2}$}(l1y)
	(l1y) edge node[above]{$\styleclock{\clock}=2$} node[below]{$\styleact{a_3}$}(lfy)
	;
	\end{tikzpicture}
  \caption{$\TB$ such that $\TA_{x,y}$ has a reachable final location \\if and only if $\TB$ is $\exists$-opaque}
	\end{subfigure}
	\caption{The TA $\TB$ for the reduction of reachability in two-clock TA to $\exists$-opacity in one-clock TA}
  \label{fig:exists-one-clock}
\end{figure}

\end{proof}

\subsection{Timed Automata with a Single Action}\label{section:opacity:one-action}

Recall that the universality problem consists in deciding whether a TA~$\TA$ accepts the set of all timed words~\cite{AD94}.
In~\cite{OW04}, it is shown that the class of one-action~TAs is one of the simplest cases for which the universality problem is undecidable among~TAs.
Therefore, noting that universality is a special case of language inclusion, relying on \cref{lem:language-inclusion} and \cref{th:weak-full-eq}, one immediately obtains that the full and weak opacity problems for TAs with one action are undecidable.
We refine slightly this result in the next theorem by showing undecidability remains even if the TA has no $\silentaction$-transition.

\begin{thm}\label{th:undecidability:opacity:one-action}
	The full and weak opacity problems for TAs with one action and no $\silentaction$-transition are undecidable.
\end{thm}

\begin{proof}
We first prove the undecidability of the full opacity problem.
Let $\TA$ be a TA with a single action `$a$' and no $\silentaction$-transition.
We want to build a new TA~$\TB$ with the same restrictions such that if we can answer the full opacity problem of this TA, then we can decide the universality problem for~$\TA$.
We define $\TB$ as follows (see illustration in \cref{figure:one-action:full-opacity}): 
we add an initial location exited by two
transitions labelled by `$a$' that must be taken urgently (\ie{} no time may elapse before taking them). The first transition leads to a private location which leads (again via another urgent transition labelled by `$a$') to the initial location of the TA~$\TA$ and the other leads in two urgent steps (aka, by reading a second `$a$') to a location $\loc_{\star}$ where every finite timed word on the alphabet $\{a\}$ can be read before reaching a final location.
From both the final locations of $\TA$ and $\loc_{\star}$, a transition labelled by `$a$' can be taken to reach the final location of $\TB$.

Ignoring the first two as well as the last `$a$' that must be read by every run of $\TB$,
the language recognized by $\TA$ corresponds exactly to the traces of private runs of $\TB$, and the traces of public runs of $\TB$ are all the finite timed words using `$a$'.
Therefore, $\TB$ is fully opaque if and only if $\PrivateTr{\TB} = \PublicTr{\TB}$ if and only if
$\Trace{\TA} = \TimedWords{\ActionSet}$ if and only if $\TA$ is universal.
Since universality for TAs with one action is undecidable~\cite{OW04}, we conclude that full opacity for one-action~TAs with no $\silentaction$-transition is undecidable.

Finally, as in the rest of this section, we deduce using \cref{th:weak-full-eq} that weak opacity for TAs with one action and no $\silentaction$-transition is undecidable too.
\end{proof}
\begin{figure}[tb]
   \centering
	\begin{tikzpicture}[pta, node distance=2cm, thin]

	\node[location, initial] (l0) at (-4, 0) {$\locinit$};

	\node[location, private] (lpriv) at (-2.5, 1) {$\locpriv$};
	\node[location, final] (lf) at (+3.5, 0.5) {$\locfinal$};
	\node[rectangle, minimum width=3cm, minimum height=1.1cm, align=center, draw] (rectangle) at (1cm, 1cm) {$\TA$};
	\node (exitA) at (2.35, 1) {};

	\node[location] (l1) at (-1.5,0) {$\loci{1}$};
	\node[location] (univ) at (1, 0) {$\loc_{\star}$};

	\path
	(l0) edge node[align=center]{$\styleclock{\clock} = 0$\\$\styleact{a}$} (lpriv)
	(lpriv) edge node[below]{$\styleact{a}$} node[above]{$\styleclock{\clock} = 0$} (rectangle)
	(l0) edge node[below]{$\styleact{a}$} node[above]{$\styleclock{\clock} = 0$} (l1)
	(univ) edge[loop below] node[align=center, right=-6pt]{$\styleact{a}$\\$\styleclock{\clock} \assign 0$} ()
	(univ) edge[bend right] node[align=center, below right=-4pt]{$\styleact{a}$} (lf)
	(exitA) edge[bend left] node[align=center, above right=-4pt]{$\styleact{a}$} (lf)
	(l1) edge node[below]{$\styleact{a}$} node[above]{$\styleclock{\clock} = 0$} (univ)
	;
	\end{tikzpicture}
  \caption{Automaton~$\TB$: Reduction from universality to full opacity}
  \label{figure:one-action:full-opacity}
\end{figure}

\begin{rem}
The problems of execution-time opacity introduced in~\cite{ALLMS23} are a particular \emph{decidable} subcase of these undecidable opacity problems with one-action~TAs.
Indeed, the execution time is equivalent to a \emph{unique} timestamp associated with the last action of the system, which can be seen as the \emph{only} observable action.
\end{rem}
\subsection{Timed Automata with few Clocks}\label{section:opacity:one-clock}

The decidability of the language inclusion problem (and the universality problem) have been studied for TAs with a low number of clocks in an attempt to regain decidability.
These results can immediately 
be translated to the weak/full opacity problems 
thanks to \cref{th:weak-full-eq} and \cref{lem:language-inclusion}.
More precisely, the universality problem for one-clock TAs has been shown to be undecidable
in~\cite{ADOQW08}, though decidability is regained when assuming the one-clock TAs
does not have any $\silentaction$-transitions~\cite{OW04}. This gives us the following results.

\begin{thm}\label{th:undecidability:full-opacity:one-clock-epsilon}
	The full and weak opacity problems for one-clock TAs are undecidable.
\end{thm}
\begin{thm}\label{theorem:one-clock-noepsilon}
The full and weak opacity problems are decidable (with non-primitive recursive complexity) for one-clock~TAs without $\silentaction$-transitions.
\end{thm}

The non-primitive recursive complexity established in~\cite{OW04} was proven to be tight  in~\cite{ADOQW08}.
Hence, while decidable, this problem cannot be effectively solved.

\begin{rem}\label{remark:noepsilon:nocontradiction}
	This result might seem to contradict the result of~\cite{AGWZH24} that proves undecidability of (language-based) opacity for one-clock TAs without $\silentaction$-transitions---but it does not.
	The discrepancy comes from the fact that, in the absence of $\silentaction$-transition, our attacker observes all actions, while their setting allows for unobservable actions---which can act as $\silentaction$-transitions even in the absence of syntactic $\silentaction$-transitions.
\end{rem}

This distinction regarding the presence of $\silentaction$-transitions disappears from two clocks onward as the language universality problem for TAs with at least two clocks is undecidable 
with and without $\silentaction$-transitions~\cite[Theorem~21]{OW04}. We thus have:

\begin{thm}\label{th:undecidability:full-opacity:two-clocks}
The full and weak opacity problems are undecidable for TAs with $\geq 2$ clocks.
\end{thm}
\subsection{Timed Automata over Discrete Time}\label{section:opacity:discrete-time}

In the general case, clocks are real-valued variables, with valuations thus ranging over $\Time = \setRgeqzero$.
TAs over discrete time restrict the clock's behaviour to valuations over $\Time = \setN$.
This restriction weakens significantly the expressiveness of TAs.
In particular, the proof of undecidability of the universality problem for TAs established in~\cite{AD94} strongly relies on continuous time, and thus cannot be translated to the 
discrete time framework. 
As a consequence of this loss of expressivity, many problems become decidable.

It is the case for opacity problems, which have been shown to be decidable for TAs over discrete time within multiple works~\cite{CG00,KKG24,AGWZH24}.
To our knowledge, none of the researches done on this problem provided a complete picture
of the complexity of the problem. Hence, 
relying on the region automaton (defined in \cref{section:tools:RA}) over discrete time and classical results on finite regular automata, we will first show that language inclusion for discrete-time TAs is decidable and \EXPSPACE{}-complete. The \EXPSPACE{}-completeness is immediately transfered to opacity by \cref{lem:language-inclusion}.

If $\clockval$, $\clockval'$ are two discrete clock valuations (\ie{} with values in~$\setN$), the definition of $\simeq$ from \cref{section:tools:RA} can be simplified into: $\clockval \simeq \clockval'$ holds if and only if, for each $\clock \in \ClockSet$, either $\clockval(\clock) = \clockval'(\clock)$ or $\clockval (\clock) > M(\clock)$ and $\clockval' (\clock) > M(\clock)$.

In continuous time, for each run of the~TA, there is a unique corresponding run of the region automaton.
In discrete time, thanks to the simplified form of the definition of~$\simeq$, the converse statement that a run of the region automaton corresponds to a unique run of the TA nearly holds.
Loss of information however remains when every clock goes beyond their maximum constant, as time elapsing is not measured beyond this point.
In order to avoid this loss of information, we add a letter~$t$ (for ticks) which occurs each time that an (integral) time unit passes in the region automaton.
This change can be operated directly on the TA~$\TA$ so that the correspondence between paths of~$\TA$ and $\RegionAutomaton{\TA}$ becomes immediate.

More precisely, we add a clock~$z$ and self-loop transitions $e_t = \big(\loc, (z=1), t, \{ z \}, \loc\big)$ on each location $\loc \in \LocSet$ of~$\TA$.
We also add the guard~``$z=0$'' to every other transition of~$\TA$.
Finally, we add an invariant ``$z \leq 1$'' to all locations of~$\TA$.
This transformation is linear in the size of~$\TA$.
We illustrate the resulting TA on a simple example in \cref{figure:example-RA}.
We depict a discrete-time TA~$\TA$, its transformation by the procedure we just described and finally its region automaton $\RegionAutomaton{\TA}$ (over discrete time).

\begin{figure}[tb]

	\centering

	\begin{subfigure}[b]{0.45\textwidth}
		\begin{tikzpicture}[pta, scale=2, xscale=.6, yscale=.5]

			\node[location, initial] at (0, 0) (s0) {$\loci{0}$};
			\node[location, private, final] at (2, 0) (s1) {$\locfinal$};

			\path (s0) edge[] node[above]{$\styleclock{\clock} > 2$} node[below]{$\styleact{a}$} (s1);
		\end{tikzpicture}
		\caption{$\TA$}
		\label{fig:discrete-time-example-TA}
	\end{subfigure}
	\hfill
	\begin{subfigure}[b]{0.45\textwidth}
	\begin{tikzpicture}[pta, scale=2, xscale=.8, yscale=.5]

		\location[initial,name=s0,at={(0, 0)}]{$\styleclock{z} \leq 1$}{$\loci{0}$};

		\location[private, final, name=s1,at={(2, 0)}]{$\styleclock{z} \leq 1$}{$\locfinal$};

		\path (s0) edge[] node[align=center]{$\styleclock{\clock} > 2$\\$\land \styleclock{z} = 0$\\$\styleact{a}$} (s1);
		\path (s0) edge[loop above] node[above, align=center]{$\styleclock{z} = 1$\\$\styleact{t}$\\$\styleclock{z} \assign 0$} (s0);
		\path (s1) edge[loop above] node[above, align=center]{$\styleclock{z} = 1$\\$\styleact{t}$\\$\styleclock{z} \assign 0$} (s1);
	\end{tikzpicture}
	\caption{$\TA$ augmented with ticks}
		\label{fig:discrete-time-example-TA:augmented}
	\end{subfigure}

	\medskip

	\begin{subfigure}[b]{\textwidth}
	\scalebox{.85}{
	\begin{tikzpicture}[pta, scale=2.7, xscale=.45, yscale=.5]

		\node[RA-location, initial] at (0, 0) (s0) {$\begin{array}{c}\loci{0} \\ \styleclock{\clock} = 0 \\ \styleclock{z} = 0 \end{array}$};

		\node[RA-location] at (1.5, 1) (s1z) {$\begin{array}{c}\loci{0} \\ \styleclock{\clock} = 1 \\ \styleclock{z} = 1 \end{array}$};

		\node[RA-location] at (3, 0) (s1) {$\begin{array}{c}\loci{0} \\ \styleclock{\clock} = 1 \\ \styleclock{z} = 0 \end{array}$};

		\node[RA-location] at (4.5, 1) (s2z) {$\begin{array}{c}\loci{0} \\ \styleclock{\clock} = 2 \\ \styleclock{z} = 1 \end{array}$};

		\node[RA-location] at (6, 0) (s2) {$\begin{array}{c}\loci{0} \\ \styleclock{\clock} = 2 \\ \styleclock{z} = 0 \end{array}$};

		\node[RA-location] at (7.5, 1) (s3z) {$\begin{array}{c}\loci{0} \\ \styleclock{\clock} > 2 \\ \styleclock{z} = 1 \end{array}$};

		\node[RA-location] at (9, 0) (s3) {$\begin{array}{c}\loci{0} \\ \styleclock{\clock} > 2 \\ \styleclock{z} = 0 \end{array}$};

		\node[RA-location, private, final] at (11, 0) (s4) {$\begin{array}{c}\locfinal \\ \styleclock{\clock} > 2 \\ \styleclock{z} = 0 \end{array}$};

		\node[RA-location, private, final] at (12.5, 1) (s4z) {$\begin{array}{c}\locfinal \\ \styleclock{\clock} > 2 \\ \styleclock{z} = 1 \end{array}$};

		\path (s0) edge[bend left] node[align=center, above]{$\styleact{\silentaction}$} (s1z);
		\path (s1z) edge[bend left] node[align=center]{$\styleact{t}$} (s1);
		\path (s1) edge[bend left] node[align=center, above]{$\styleact{\silentaction}$} (s2z);
		\path (s2z) edge[bend left] node[align=center]{$\styleact{t}$} (s2);
		\path (s2) edge[bend left] node[align=center, above]{$\styleact{\silentaction}$} (s3z);
		\path (s3z) edge[bend left] node[align=center]{$\styleact{t}$} (s3);

		\path (s4z) edge[bend left] node[align=center]{$\styleact{t}$} (s4);
		\path (s4) edge[bend left] node[align=center]{$\styleact{\silentaction}$} (s4z);
		\path (s3) edge[bend left] node[align=center]{$\styleact{\silentaction}$} (s3z);

		\path (s3) edge node[align=center, above]{$\styleact{a}$} (s4);

	\end{tikzpicture}
	}
	\caption{$\RegionAutomaton{\TA}$}
		\label{fig:zzzz}
	\end{subfigure}

	\caption{A discrete-time region automaton example}
	\label{figure:example-RA}

\end{figure}

With this construction, time information becomes superfluous in the TA as it can be deduced from the number of ticks that were produced, which also appears within a path of the region automaton. 
For instance, consider the run of the TA~$\TA$ of \cref{fig:discrete-time-example-TA} that remains 4~time units in~$\locinit$ before going to~$\locfinal$.
The timed word $(a, 4)$ on the original TA~$\TA$ becomes $(t,1)(t,2)(t,3)(t,4)(a,4)$ in our transformed~TA in \cref{fig:discrete-time-example-TA:augmented}.
The untimed word obtained in $\RegionAutomaton{\TA}$ is $tttta$, which means that 4~ticks occurred before the action $a$ was produced.
There is thus a direct correspondence between the original timed word $(a,4)$ and the untimed word $tttta$.
In the rest of this subsection, we only consider TAs augmented with ticks.
From the previous discussion, we have:

\begin{lem}\label{lemma:discrete-time:bijection-with-untimed-words}
The language of a discrete-time TA and the language of its region automaton are in bijection.
\end{lem}
\begin{proof}
Let $\TA$ be a discrete-time~TA.
We explicit the bijection of the lemma.

Let $\run$ be a path of~$\TA$, generating the timed word~$w$.
Since $\TA$ includes ticks, $w$ is of the form
\[(t,1) \dots (t,\tau_0) (a_0, \tau_0)~ (t, \tau_0 + 1) \dots (t, \tau_1) (a_1, \tau_1) ~ \dots ~ (t,\tau_{n-1}+1) \dots (t,\tau_n)(a_n, \tau_n).\]
(Note that consecutive actions can occur in 0-time, so some of the sequences of consecutive ``$t$'' can be empty, including the first one if an action occurs at time~0.)
To the timed word~$w$, we associate the untimed word produced within the region automaton by the path~$[\run]$ corresponding to~$\run$:
\[\underbrace{tt \dots t}_{\tau_0 \text{ times}} ~a_0 \underbrace{tt \dots t}_{(\tau_1 - \tau_0) \text{ times}} a_1 ~ \dots \underbrace{tt \dots t}_{(\tau_n - \tau_{n-1}) \text{ times}} ~a_n.\]
This association is injective as the sequence $(\tau_i)_{i\leq n}$ which was removed in
the transformation depends only on the number of ``$t$'' in the timed word.
Moreover, it is surjective as
given an untimed word in $\RegionAutomaton{\TA}$ $w' =\underbrace{tt \dots t}_{k_0 \text{ times}} ~a_0 \underbrace{tt \dots t}_{k_1 \text{ times}} a_1 ~ \dots \underbrace{tt \dots t}_{k_n \text{ times}} ~a_n$ produced by a path $[\run']$ of the region automaton,
defining
\[w = (t,1)\dots (t,k_0) (a_0, k_0) (t,k_0+1)\dots (t,k_0+k_1) (a_1, k_0 + k_1) \dots (a_n, \sum\limits_{i=0}^{n} k_i)\]
we have that $w$ is the timed word generated by the unique path of the TA corresponding to~$\run'$ and $w$ is associated with~$w'$.
\end{proof}

This bijection allows us to solve the language inclusion problem for discrete-time TAs
by considering the language inclusion of the region automata (which are 
finite regular automata).

\begin{prop}\label{th:decidability:inclusion:discrete-time}
	Language inclusion in discrete-time TAs is \EXPSPACE{}-complete.
\end{prop}

We separate both directions of the proof for sake of clarity. 
We start by showing that the language inclusion in discrete-time TAs can be achieved in \EXPSPACE{}.

\begin{proof}
Let $\TA$ and $\TB$ be two discrete-time TAs, and let $\RegionAutomaton{\TA}$ and $\RegionAutomaton{\TB}$ be their respective region automata. Then from \cref{lemma:discrete-time:bijection-with-untimed-words}, we have
\[ \Trace{\TA} \subseteq \Trace{\TB} \text{ if and only if } \Trace{\RegionAutomaton{\TA}} \subseteq \Trace{\RegionAutomaton{\TB}}\text{.}\]
Thus deciding the language inclusion in discrete-time TAs amounts to solving the language inclusion problem in the context of finite regular automata, which can be done in
\PSPACE{} in the size of the region automata.
Noting that the region automaton of the ticked TA is exponential in the size of the ticked~TA (itself linear in~$\TA$), this produces an \EXPSPACE{} algorithm.
\end{proof}

Let us now show that the language inclusion in discrete-time TAs is \EXPSPACE{}-hard.
To do so, we will reduce from a succinct variant of the equality of rational expressions.

\begin{defi}[Rational expressions with square]
The expressions $\emptyset$, $\silentaction$, and $\action$ with $\action \in \ActionSet$ are \emph{rational expressions with square}.
If $\sqexp$, $\sqexp_1$ and $\sqexp_2$ are rational expressions with square, then so are $\sqexp_1 + \sqexp_2$, $\sqexp_1 \cdot \sqexp_2$, $\sqexp^*$ and $\sqexp^2$.

A \emph{rational language with square} is a set of words on~$\ActionSet$ represented by a rational expression with square.
\end{defi}

The operators on the rational expressions are interpreted in the usual way.
For instance, the expression $(a+ab)^2$ represents the set of words $\set{aa, aab, aba, abab}$. There can be several expressions representing the same language.

The expressivity of rational languages with square is exactly the same as of rational languages since using the square is equivalent to concatenating an expression with itself. However, the description of a language with square may be exponentially more succinct.
This explains the following result.

\begin{thmC}[\cite{MS72}]\label{theorem:discrete-time:language-inclusion}
Let $\Language_1$ and $\Language_2$ be two rational languages with square.
Deciding whether $\Language_1 = \Language_2$ is \EXPSPACE{}-complete.
\end{thmC}

We can now proceed to the proof of the hardness result of \cref{th:decidability:inclusion:discrete-time}.

\begin{proof}[Proof of the hardness of language inclusion for discrete-time~TAs]
Let $\Language$ be a rational language with square and $\sqexp$ be the rational expression with squares that represents it.
From the structure of~$\sqexp$, we will build a TA whose untimed language is~$\Language$.

Since we will compare the timed language of TAs, and we only want to compare their untimed languages, we need to impose a standard for the timestamps of their words.
We choose that each action must occur at an even number of time units.
More precisely, if the TA recognizes words of at least one letter, it will read the first letter without any delay (at time~0), and wait two time units between each letter.
To do this, we use a clock $x$ which is reused for all the constructions, and which let 2 time units elapse before being reset, but only when a letter occurs.
Every operation, beside reading a letter, is then done in 0-time.
In particular, our constructions always start and end with $x=0$, and only allow time to elapse when a letter is read.
We present in \cref{figure:discrete-time:EXPSPACE-hardness} the inductive constructions corresponding to the basic rational expressions and the operators $+$, $\cdot$, and~$^*$.
The case of the square operator is explained separately.

\begin{table}
\caption{Timed automata constructions $\TA_{\sqexp}$ for regular expressions}
\label{figure:discrete-time:EXPSPACE-hardness}
{\centering
\begin{tikzpicture}[pta, node distance=2cm, thin]
	\node (p00) at ( 0, 0)  {};
	\node (p10) at ( 4, 0)  {};
	\node (p17) at ( 4, -12.25) {};
	\node (p20) at (12, 0)  {};
	\node (p27) at (12, -12.25) {};
	\node (p07) at ( 0, -12.25) {};
	\node (p01) at ( 0, -1) {};
	\node (p21) at (12, -1) {};
	\node (p02) at ( 0, -2) {};
	\node (p22) at (12, -2) {};
	\node (p03) at ( 0, -3.5) {};
	\node (p23) at (12, -3.5) {};
	\node (p04) at ( 0, -5.5) {};
	\node (p24) at (12, -5.5) {};
	\node (p05) at ( 0, -9) {};
	\node (p25) at (12, -9) {};

	\draw[-] (0,0) -- (12,0) -- (12,-12.25) -- (0, -12.25) -- (0,0);
	\draw[-] (p10) -- (p17);
	\draw[-] (p01) -- (p21);
	\draw[-] (p02) -- (p22);
	\draw[-] (p03) -- (p23);
	\draw[-] (p04) -- (p24);
	\draw[-] (p05) -- (p25);
	\node[align=center] at (2, -0.5) {Rational expression \\ with square $\sqexp$};
	\node at (8, -0.5) {Timed automaton $\TA_{\sqexp}$};

	\node at (2, -1.5) {$\styleact{\silentaction}$};
	\node at (2, -2.75) {$a \in \ActionSet$};
	\node at (2, -4.5) {$\sqexp_1 \cdot \sqexp_2$};
	\node at (2, -7.25) {$\sqexp_1 + \sqexp_2$};
	\node at (2, -10.5) {$\sqexp^*$};
	\node[location, initial] (a0) at (7, -1.5) {};
	\node[location, final] (a1) at (9, -1.5) {};
	\path (a0) edge node[align=center, above=-4.5mm]{$\styleclock{\clock} = 0$ \\ $\styleact{\silentaction}$} (a1);
	\node[location, initial] (b0) at (6, -2.65) {};
	\node[location] (b1) at (8, -2.65) {};
	\node[location, final] (b2) at (10, -2.65) {};
	\path (b0) edge node[align=center, above=-4.5mm]{$\styleclock{\clock} = 0$ \\ $\styleact{a}$} (b1);
	\path (b1) edge node[above]{$\styleclock{\clock} = 2$} node[align=center, below]{$\styleact{\silentaction}$\\$\styleclock{\clock} \assign 0$} (b2);
	\node[location, initial, rectangle, minimum width=1.5cm, minimum height=1cm] (c0) at (6.5, -4.5) {$\TA_{\sqexp_1}$};
	\node[location, rectangle, minimum width=1.5cm, minimum height=1cm] (c1) at (9.5, -4.5) {$\TA_{\sqexp_2}$};
	\path (c0) edge node[align=center]{$\styleclock{\clock} = 0$ \\ $\styleact{\silentaction}$} (c1);
	\node[location, initial] (d0) at (6, -7.25) {};
	\node[location, rectangle, minimum width=1.5cm, minimum height=1cm] (d1) at (9, -6.5) {$\TA_{\sqexp_1}$};
	\node[location, rectangle, minimum width=1.5cm, minimum height=1cm] (d2) at (9, -8) {$\TA_{\sqexp_2}$};
	\path (d0) edge node[align=center, above=+1mm]{$\styleclock{\clock} = 0$ \\ $\styleact{\silentaction}$} (d1);
	\path (d0) edge node[align=center, below]{$\styleclock{\clock} = 0$ \\ $\styleact{\silentaction}$} (d2);
	\node[location, initial] (e0) at (6.5, -11.5) {};
	\node[location, rectangle, minimum width=1.5cm, minimum height=1cm] (e1) at (6.5, -10) {$\TA_{\sqexp}$};
	\node[location, final] (e2) at (9.5, -11.5) {};
	\path (e0) edge[bend left] node[align=center, left]{$\styleclock{\clock} = 0$ \\ $\styleact{\silentaction}$} (e1);
	\path (e1) edge[bend left] node[align=center, right]{$\styleclock{\clock} = 0$ \\ $\styleact{\silentaction}$} (e0);
	\path (e0) edge node[align=center, right=-1mm]{$\styleclock{\clock} = 0$ \\ $\styleact{\silentaction}$} (e2);
\end{tikzpicture}

}
\end{table}

For the square construction $\TA_{\sqexp^2}$ (which we give in \cref{figure:discrete-time:EXPSPACE-hardness:square}), we need to add three additional clocks per square occurrence in the original expression~$\sqexp$: the first one, $z$,  manages the particular case of the empty word $\silentaction$ by detecting whether some time has passed during the traversal of $\TA_{\sqexp}$, while the two other clocks $y$ and~$v$ are used to force exactly two traversals in~$\TA_{\sqexp}$.

Indeed, the shift between the clocks $x$, $y$ and $v$ (with values kept between 0 and~2 all along the run) allows to keep in memory the number of remaining traversals in~$\TA_{\sqexp}$ by being modified once during the first traversal (\circleone{}), a second time between the first and second traversals (\circletwo{}), and being tested at the end of the second traversal (\circlethree{}).
These added clocks cannot be reused in nested squares constructions.
Thus we introduce a number of clocks equal to three times the maximal number of nested squares in the original expression~$\sqexp$.

More precisely, the previously built TA $\TA_{\sqexp}$ is modified into
$\tilde{\TA}_{\sqexp}$ by adding on every location silent loop transitions
resetting $y$ and $v$ when they reach $2$, as well as a
silent loop transition with guard $x=1 \wedge y=1$ and reset set $\setsmall{y,v}$.
At most one of the latter loops, denoted by~\circleone{}, is taken during an execution, and it requires at least one letter to be triggered.
This transition ensures that $y=v\neq x$ in the following.
This property is necessary to take the transition~\circletwo{}, which now ensures that $x=v\neq y$, which will allow taking the transition~\circlethree{}.
As mentioned, taking the transition~\circleone{} requires at least one letter to be read,
hence why, when $\sqexp$ contains the empty word, we need the clock~$z$ to give an alternative way to exit the gadget.
Formally, we have:

\begin{figure}
\begin{center}
\begin{tikzpicture}[pta, node distance=2cm, thin]
	\node[location, initial] (e0) at (-5,0) {};
	\node[location, rectangle, minimum width=3.5cm, minimum height=2cm] (e1) at (0, 0) {$\tilde{\TA}_{\sqexp}$};	\node[location, final] (e2) at (5,0) {};

	\node[align=center, below=+2mm] at (-5.5,0) {$\styleclock{\clock} = 0$ \\ $\styleact{\silentaction}$\\$ \styleclock{y} \assign 0$\\$\styleclock{z} \assign 0$\\$\styleclock{v} \assign 0$};

	\path (e0) edge node[align=center, above]{$\styleact{\silentaction}$}(e1);
	\path (e1) edge node[align=center, above]{$\styleclock{z} = 0$ \\ $\styleact{\silentaction}$}(e2);
	\path (e1) edge[bend right] node[align=center, below]{$\styleclock{\clock} = 0 \wedge \styleclock{y} = 1 \wedge \styleclock{v} = 0$ \\ $\styleact{\silentaction}$\\\Large{\circlethree{}}}(e2);
	\path (e1) edge[bend right] node[align=center, above]{\Large{\circletwo{}}\\$\styleclock{\clock} = 0 \wedge \styleclock{y} = 1  \wedge \styleclock{v} = 1$\\$\styleact{\silentaction}$ \\ $\styleclock{v} \assign 0$} (e0);

	\draw[-] (1.2,2) -- (1.2, 4.5) -- (5, 4.5) -- (5,2) -- (1.2, 2);
	\node at (1.75, 3.25) {\Large{\circleone{}:}};
	\node at (0.75, 0.25) {\Large{\circleone{}}};
	\node[location] (e5) at (2.5, 3.25) {};
	\path (e5) edge[loop right] node[align=center, right]{$\styleclock{x} = 1$\\$\land \styleclock{y}=1$ \\ $\styleact{\silentaction}$\\$\styleclock{y} \assign 0$\\$\styleclock{v} \assign 0$} (e5);

\end{tikzpicture}
\caption{TA $\TA_{\sqexp^2}$. $\tilde{\TA}_{\sqexp}$ is the automaton $\TA_{\sqexp}$ modified through~\circleone{}}

\label{figure:discrete-time:EXPSPACE-hardness:square}
\end{center}
\end{figure}

\begin{defi}[Square construction]
Let $\TA_{\sqexp} = (\ActionSet, \LocSet, \locinit, \emptyset, \FinalSet, \ClockSet, I, E)$ be the timed automaton corresponding to the rational expression with square $\sqexp$. Then, the TA corresponding to $\sqexp^2$ is $\TA_{\sqexp^2} = (\ActionSet, \LocSet \cup\{\locinit', \locfinal'\}, \locinit', \emptyset, \setsmall{\locfinal'}, \ClockSet \cup \setsmall{v,y,z}, I', E')$ where
$I'$ is the extension of $I$ such that $I'(\locinit')=I'(\locfinal')=\BTrue{}$,
and
\begin{align*}
E' = & E \cup \big \{(\locinit', \silentaction, \BTrue{}, \emptyset, \locinit) \big\} \cup  \\
& \bigcup\limits_{\loc \in \LocSet}^{} \big \{(\loc, \silentaction, y=2, \setsmall{y}, \loc), (\loc, \silentaction, v=2, \setsmall{v}, \loc), (\loc, \silentaction, x=1 \wedge y=1, \setsmall{v,y}, \loc) \big\} \cup \\
&\bigcup\limits_{\locfinal \in \FinalSet}^{} \big\{(\locfinal, \silentaction, x=0 \wedge y=1 \wedge v=1, \setsmall{v}, \locinit'), \\
&\quad\quad\quad(\locfinal, \silentaction, z=0, \emptyset, \locfinal'), (\locfinal, \silentaction, x=0 \wedge y=1 \wedge v=0, \emptyset, \locfinal') \big\}\text{.}
\end{align*}
\end{defi}

\begin{rem}
	Note that we could most probably save one clock in \cref{figure:discrete-time:EXPSPACE-hardness:square} by resetting the clocks everywhere every 3 instead of 2 time units, and therefore need only 2 instead of 3~additional clocks for each occurrence of a square in the original expression.
	But that would make the proof slightly more cumbersome, without changing the class of complexity.
\end{rem}

Two rational languages with square $\Language_1$ and~$\Language_2$, respectively represented by the expressions $\sqexp_1$ and~$\sqexp_2$, are equal if and only if the TAs $\TA_{\sqexp_1}$ and $\TA_{\sqexp_2}$ recognize the same timed language.
The obtained automata are discrete-time TAs of polynomial size in the rational expressions.
Thus from \cref{theorem:discrete-time:language-inclusion} follows the \EXPSPACE{}-hardness of language inclusion for discrete-time~TAs.
\end{proof}

As the inter-reductions provided by \cref{th:weak-full-eq} and \cref{lem:language-inclusion}
are polynomial, we immediately obtain the following.

\begin{cor}\label{theorem:discrete-time}
The weak and full opacity problems for discrete-time TAs are \EXPSPACE{}-complete.
\end{cor}
\subsection{Observable Event-Recording Automata}\label{section:opacity:era}
In~\cite{Cassez09}, the opacity problems were shown to be undecidable for Event-Recording Automata (ERAs)~\cite{AFH99}, a subclass of TAs where every clock is associated with a specific event, and every transition resets exactly the clock associated with the event of the transition.
For this reason, the valuation of a clock is entirely determined by the duration since the last occurrence of its associated event.
One of the main interest of ERAs is that they are determinisable~\cite{AFH99}.
This determinisation is carried out through the standard subset construction.

The undecidability result from~\cite{Cassez09} on ERAs required to make some events unobservable.
Hence, in our framework they would be replaced by $\silentaction$-transitions.
We define \emph{observable ERAs} (oERAs) as ERAs where the actions resetting the clocks must be observable.
This means that the information required for the determinisation now belongs to the trace that is observed.

\begin{defi}[oERAs]
	An \emph{observable event-recording automaton} (oERA) is a TA such that
	\begin{itemize}
		\item the set of clocks is $\ClockSet_\ActionSet = \{ \clock_\action \mid \action \in \ActionSet \}$;
		\item for any edge with action~$\action \in \ActionSet$, the set of clocks to reset is exactly $\{ \clock_\action \}$; and
		\item for any edge with action~$\silentaction$, the set of clocks to reset is empty.
	\end{itemize}
\end{defi}

ERAs do not include $\silentaction$-transitions. Therefore, while 
language inclusion is decidable and \PSPACE{}-complete for ERAs~\cite[Theorem~4]{AFH99},
we cannot directly use this result in conjunction with 
\cref{th:weak-full-eq} and \cref{lem:language-inclusion} to show the decidability 
of the full and weak opacity problems for oERAs.
We can however rely on the usual tools for ERAs in order to prove our result.

Given an oERA~$\TA$, as the information required for the determinisation now belongs to the trace that is observed, we can build through the subset construction a TA~$\mathit{Det}_\TA$ such that any path~$\run$ in~$\TA$ corresponds to a path~$\run_D$ in~$\mathit{Det}_\TA$ with the same trace and ending in a location labelled by the set of all the locations of~$\TA$ that can be reached with a run that has the same trace as~$\run$.
This information, combined with the construction of~$\AMemo$ (see \cref{definition:AMemo}) which stores in the state of the oERA whether the private location was visited or not, provides the following result.

\begin{thm}\label{th:decidability:oera}
The weak and full opacity problems for oERAs are \PSPACE{}-complete.
\end{thm}
\begin{proof}
	As said above, our algorithm to test weak opacity consists, given an oERA~$\TA$, in first building the corresponding~$\AMemo$ (see \cref{definition:AMemo}), then determinising it through the subset construction,  computing its region automaton, and finally testing reachability of a location including a final private location, but no final public location. 
If such a location is reachable, then the associated trace can be produced by a private run but no public run in $\AMemo$ (recall that $\AMemo$ has the same opacity properties as $\TA$), which constitutes a violation of weak opacity.

	\begin{itemize}
		\item Let us first explain why this algorithm is in \PSPACE{}.

	The determinisation of the oERA causes the number of locations to become exponential in the size of the entry~\cite{AFH99}, and the construction of the region automaton gives an exponential number of clock regions, bounded by $\NbClockRegions$~\cite{AD94}.
	The size of the region automaton is thus exponential in the number of locations and in the number of clocks of~$\TA$.
	On the region automaton, testing the reachability of a location can be done in \NLOGSPACE{}.
	Hence the problem of weak opacity in oERAs is in \PSPACE{}.

		\item Let us now explain why this problem is \PSPACE{}-hard.
	We reduce from the reachability problem for TAs, which is \PSPACE{}-complete~\cite{AD94}.

	Let $\TA$ be a TA with a set of final locations~$\FinalSet$. We can modify the actions labelling its transition so that it is an oERA without modifying reachability properties.
	We consider $\TA'$ the TA identical to~$\TA$ except that the set of private locations is set to~$\FinalSet$.
	This way, every run of~$\TA'$ is private.
	Thus $\TA'$ is weakly opaque if and only if no final location of~$\TA$ is reachable.
	Hence, the weak opacity problem in oERAs is \PSPACE{}-hard.
	\end{itemize}
	This result extends to full opacity thanks to \cref{th:weak-full-eq}.
\end{proof}

\section{Opacity with Limited Attacker Budget}\label{section:finite}

One of the causes for the undecidability of the opacity problems in~\cite{Cassez09}
stems from the unbounded memory the attacker might require to remember a run of the TA.
As a consequence, one can wonder whether the opacity problems remain undecidable when the attacker performs only a \emph{finite} number of observations.
This models the case of an attacker with a limited attack budget.

This notion of limited budget for a TA has previously been considered 
in~\cite{LRNA17,ALMS22,ALLMS23}. In~\cite{LRNA17}, a determinisation of TAs whose traces have their length bounded by a constant has been studied in~\cite{LRNA17}; here, contrarily to~\cite{LRNA17}, we do not forbid silent transitions loops in our TAs.
In~\cite{ALMS22,ALLMS23}, a notion of opacity, called execution time opacity, was considered and proven decidable for the full class of TAs. In their framework, the only action the attacker could observe was the duration of the full execution of the run. This can be represented by having a single observable action within the TA which labels the transitions reaching the final location.
We extend this line of research by assuming the attacker can observe the system only a \emph{finite} and \emph{a~priori} fixed number of times.
To the best of our knowledge, the results in this section represent 
the second result of the literature (after~\cite{ALMS22,ALLMS23}) providing a decidable opacity result for the full class of TAs over dense time.

We define a new setting for opacity as follows.
The attacker chooses a time~$t$ and observes the first observable action occurring after the selected time.
For instance, if the selected time is~$2.1$ and the trace is $(a,1.9)(b,2.2)(c,2.5)$ is about to be produced, then the attacker only observes $(b,2.2)$.
Repeating this operation several times during the same run, the attacker may constitute an observed timed word and adapt their strategy to select the next time in consequence, until the attacker spent their entire budget.
Such a strategy is given by a \emph{time selection} function.
In this framework, an observed timed word $w$ is \emph{opaque against the attacker's time selection function} iff there exists a private run and a public run which observed timed word 
under the time selection function is $w$.

In this section, we address two natural choices of strategies for the attacker, and the problem of opacity against all possible $N$-finite time selection function.
We establish the decidability and complexity of the opacity problems for each of them.

We consider three main settings:
\begin{enumerate}
	\item when the attacker chooses to observe the first actions of the system (\cref{sec:Nfirst});
	\item when the attacker chooses to follow a sequence of observation times fixed before the run (\cref{sec:Nfixed}); and
	\item when the attacker may adapt their strategy, and thus their times of observations, based on what has been observed up to now (\cref{ss:N-dynamic}).
\end{enumerate}

Formally, in order to represent the ability of the attacker to decide when they observe the
system, we define the time selection function~$\strategy{}$ associating to a timed word the next time when the attacker will start observing the system again.

We denote by $\Subwords(W)$ the set of subwords of timed words of the set $W \subseteq (\ActionSet\times\setRgeqzero)^*$.
Let $\preceq$ be the partial order on timed words such that by $v \preceq w$ if $v$~is a prefix of~$w$ (in particular, $\silentaction\preceq w$~and $w\preceq w$).

\begin{defi}[Time selection function]
A function $\strategy{}:\TimedWords{\ActionSet} \mapsto \setRgeqzero\cup\{+\infty\}$ is a \emph{time selection function} if it is non-decreasing, for the partial order $\preceq$ on timed words.
\end{defi}
\begin{defi}[Projection following a time selection function $\strategy{}$]\label{definition:projection}
Let $\strategy{}$ be a time selection function. We define the associated projection on timed words as follows.
\begin{itemize}
	\item $\projection{}(\silentaction) = \silentaction$, and
	\item for every letter $a\in \ActionSet, t \in \setRgeqzero $ such that $ w \cdot (a,t) $ is a timed word,
\[
\projection{}(w \cdot (a,t)) = \left\{
    \begin{array}{ll}
        \projection{}(w) \cdot (a,t) & \mbox{if } t \geq \strategy{}(\projection{}(w)) \\
        \projection{}(w) & \mbox{otherwise.}
    \end{array}
\right.
\]
\end{itemize}
\end{defi}

As said above, an important assumption of our setting for gaining decidability is that the number of observations performed by the attacker is finite.

\begin{defi}\label{definition:N-bounded}
Given $N \in \setN$, a timed language is \emph{$N$-bounded} if none of the words it contains has strictly more than $N$~letters.
A time selection function $\strategy{}$, defined for a TA~$\TA$, is \emph{$N$-finite} if the language $\projection{}(\Trace{\TA})$ is $N$-bounded.
\end{defi}

We extend the definitions of opacity to TAs %
with time selection functions. Recall that $\PrivateTr{\TA}$ and $\PublicTr{\TA}$ are the sets of traces produced by private and public runs on~$\TA$, respectively.

\begin{defi}[Opacity against time selection function $\strategy{}$]
Let $\TA$ be a TA.
Let $\strategy{}$ be a time selection function and $\projection{}$ its associated projection.
We say that $\TA$ is \emph{fully} (resp.\ \emph{weakly}) \emph{opaque against $\strategy{}$} if $\projection{}(\PrivateTr{\TA}) = \projection{}(\PublicTr{\TA})$ (resp.\ $\projection{}(\PrivateTr{\TA}) \subseteq \projection{}(\PublicTr{\TA})$).
It is \emph{$\exists$-opaque against $\strategy{}$} if $\projection{}(\PrivateTr{\TA}) \cap \projection{}(\PublicTr{\TA}) \neq \emptyset$.
\end{defi}
\subsection{Equivalence of timed words}\label{sec:techtools}

Our method relies on timed words grouping within equivalence classes, where two timed words belong to the same class if and only if there does not exist a TA that
can separate them (\ie{} by accepting one, but not the other). In this section, we detail this equivalence.

A \emph{time sequence} of length~$N$ is a non-decreasing sequence of $N$ non-negative reals.

If $w = (a_0, t_0) \dots (a_{N-1}, t_{N-1})$ is a timed word, we denote by~$\untimed{w}$ the word $a_0 \dots a_{N-1}$; and by~$\timed{w}$ the time sequence $(t_0, \dots, t_{N-1})$.
Note that we only consider \emph{finite} time sequences.

The following definition identifies groups of equivalent time sequences, and is inspired by the equivalence relation between clock valuations.

\begin{defi}[Equivalence of two time sequences]
Let $\tau$ and~$\tau'$ be two time sequences of same length~$N$.
We say they are \emph{equivalent}, denoted by $\tau \simeq \tau'$, if their integral parts are the same and their fractional parts are ordered in the same way and zero at the same time, \ie{} when the following three properties are all verified:
\begin{enumerate}
	\item $\forall i \in \interval{0}{N-1}~ \intpart{\tau_i} = \intpart{\tau'_i}$,
	\item $\forall i \in \interval{0}{N-1}~ \fract{\tau_i} = 0 \iff \fract{\tau'_i} = 0$, and
	\item $\forall i,j \in \interval{0}{N-1}~ \fract{\tau_i} \leq \fract{\tau_j} \iff \fract{\tau'_i} \leq \fract{\tau'_j}$.
\end{enumerate}
\end{defi}

We extend the equivalence relation of time sequences to words in the next definition.

\begin{defi}[Equivalence of two words]
Two timed words $w$ and~$w'$ are \emph{equivalent}, denoted by $w \simeq w'$, if their untimed words are the same and their time sequences are equivalent, i.e.:
\begin{enumerate}
	\item $\untimed{w} = \untimed{w'}$, and
	\item $\timed{w} \simeq \timed{w'}$.
\end{enumerate}
\end{defi}

We now show that equivalent words are accepted (or rejected) by exactly the same sets of TAs, and they can furthermore be produced by the same sequences of locations and transitions on each~TA.
In other words, no TA can distinguish equivalent words.
The two directions of this results are established through \cref{lem:folklore:equivalence,prop:folklore:rec}.

\begin{prop}\label{lem:folklore:equivalence}
Let $\TA$ be a TA, $w$ a timed word in~$\Trace{\TA}$ produced by a run~$\run$ of~$\TA$, and $w'$ a timed word such that $w \simeq w'$.
Then $w'$ is in~$\Trace{\TA}$ and is produced by a run following exactly the same sequence of locations and transitions in~$\TA$ as~$\run$.
\end{prop}
\begin{proof}
Given two equivalent timed words $w$ and~$w'$ of length~$N$ and a run $\run$ of trace~$w$ in some TA~$\TA$, we explain how to build a run~$\run'$ of trace~$w'$ in~$\TA$, visiting the same locations and taking the same transitions as~$\run$.
The idea, inspired by a technique from~\cite{ADLL25journal}, is to make a continuous distortion of the times occurring in~$\run$, mapping the timestamps of~$w$ onto those of~$w'$.
We denote by $\tau$ and~$\tau'$ the time sequences $\timed{w}$ and $\timed{w'}$, and $N$ their length.

Let $f$ (resp.~$f'$) be the increasing sequence whose elements are those of $\set{0, 1} \cup \set{\fract{\tau_i} \mid 0 \leq i < N}$ (resp.\ those of $\set{0, 1} \cup \set{\fract{\tau'_i} \mid 0 \leq i < N}$).

Since $w \simeq w'$, the fractional parts of the $\tau_i$ and the $\tau'_i$ are ordered in the same way in~$f$ and in~$f'$.
We also have $\vert f \vert = \vert f' \vert$, $f_0 = f'_0 = 0$ and $f_{\vert f \vert-1} = f'_{\vert f' \vert-1} = 1$.
We define the distortion function from $f$ to~$f'$ as follows:
\[\begin{array}{rrcl}
\gamma_{f \rightarrow f'}: & [0;1) & \longrightarrow & [0;1) \\
 & t & \longmapsto & f'_j + \frac{f'_{j+1}-f'_j}{f_{j+1}-f_j} (t-f_j) \text{ where $f_j,f_{j+1}$ are such that } f_j \leq t < f_{j+1}\text{.} \\
\end{array}\]
This distortion expands or shrinks each of the time intervals between fractional parts of timestamps of~$w$.
It is a continuous piecewise linear function on the interval~$[0;1)$.

\begin{exa}
	We depict an example of a distortion function in \cref{fig:folklore:proof}.
\end{exa}

\begin{figure}
\begin{tikzpicture}[scale=0.5]
	\draw[->,line width=1pt] (0,0) -- (10.3,0); %
	\draw[->,line width=1pt] (0,0) -- (0,10.3); %

	\def\A{10* 0.1}
	\def\B{10* 0.4}
	\def\C{10* 0.9}
	\def\a{10* 0.3}
	\def\b{10* 0.5}
	\def\c{10* 0.6}
	\def\ecart{-0.8}

	\draw (10 - 1.5*\ecart,0) node{t};
	\draw (0,10 - 1.5*\ecart) node{$\gamma_{f \rightarrow f'}(t)$};

	\node at (\ecart, 0) {0};
	\node at (\ecart, \A) {$f'_1$};
	\node at (\ecart, \B) {$f'_2$};
	\node at (\ecart, \C) {$f'_3$};
	\node at (\ecart, 10) {1};
	\node at (0, \ecart) {0};
	\node at (\a, \ecart) {$f_1$};
	\node at (\b, \ecart) {$f_2$};
	\node at (\c, \ecart) {$f_3$};
	\node at (10, \ecart) {1};
		\draw[dashed] (\a,0) -- (\a,\A);
		\draw[dashed] (0,\A) -- (\a,\A);
		\draw[dashed] (\b,0) -- (\b,\B);
		\draw[dashed] (0,\B) -- (\b,\B);
		\draw[dashed] (\c,0) -- (\c,\C);
		\draw[dashed] (0,\C) -- (\c,\C);
		\draw[dashed] (10,0) -- (10,10);
		\draw[dashed] (0,10) -- (10,10);

	\coordinate (O) at (0,0);
	\coordinate (aA) at (\a,\A);
	\coordinate (bB) at (\b,\B);
	\coordinate (cC) at (\c,\C);
	\coordinate (I) at (10,10);

	\draw[color=redColorBlind, ultra thick, densely dashed] (O) -- (I);
	\draw[color=blueColorBlind, ultra thick] (O) -- (aA) -- (bB) -- (cC) -- (I);

	\draw (O) node {$\bullet$};
	\draw (aA) node {$\bullet$};
	\draw (bB) node {$\bullet$};
	\draw (cC) node {$\bullet$};
	\draw (I) node {$\bullet$};

\end{tikzpicture}
\caption{Graph of the function $\gamma_{f \rightarrow f'}$ with $f = (0, 0.3, 0.5, 0.6, 1)$ and $f' = (0, 0.1, 0.4, 0.9, 1)$}
\label{fig:folklore:proof}
\end{figure}

Let $\tweak{f}{f'} : t \mapsto \intpart{t} + \gamma_{f \rightarrow f'}(\fract{t})$ be the distortion lifted to~$\setRgeqzero$.
It is strictly increasing, continuous on~$\setRgeqzero$ and its inverse function is $\tweak{f'}{f}$.
Moreover, it preserves the ceiling and floor of each time and maps the timestamps of $w$ on those of~$w'$: for each $i \leq N$, $\tweak{f}{f'}(\tau_i) = \intpart{\tau_i} + \gamma_{f \rightarrow f'}(\fract{\tau_i}) = \intpart{\tau_i} + \fract{\tau_i'} = \tau_i'$.
The strict monotonicity of $\gamma_{f \rightarrow f'}$ and the periodicity (period~1) of $\fract{\tweak{f}{f'}}$ ensure that the satisfaction of the guards, which compare clocks to integers, is preserved.
Indeed, when a clock is tested by a guard, its value corresponds to the difference of two times (the current time minus the time when this clock was last reset, or~$0$ if it has never been reset), let us say $t_1 - t_2$.
We have
\[t_1 - t_2 = \intpart{t_1} + \fract{t_1} - \intpart{t_2} - \fract{t_2} \text{ ;} \]
\[ \tweak{f}{f'}(t_1) - \tweak{f}{f'}(t_2) = \intpart{t_1} + \gamma_{f \rightarrow f'}(\fract{t_1}) - \intpart{t_2} - \gamma_{f \rightarrow f'}(\fract{t_2}).\]

Hence $\tweak{f}{f'}$ preserves a comparison $t_1 - t_2 \compOp k$ with $k \in \setN$ and ${\compOp} \in \set{<,=,>}$ if and only if we have
\[ \gamma_{f \rightarrow f'}(\fract{t_1}) \compOp \gamma_{f \rightarrow f'}(\fract{t_2}) \iff \fract{t_1} \compOp \fract{t_2} \] which is true because $\gamma_{f \rightarrow f'}$ is strictly increasing.

If $\run$ is 
$(\loci{0}, \clockval_0), (d_0, \edgei{0}), (\loci{1}, \clockval_1), \ldots, (\loci{n}, \clockval_n )$
where for each  $i \in \interval{0}{n-1}$, $e_i = (\loc_i , g_i, a_i , \resets_i, \loc_{i+1})$ and $a_i \in \ActionSet \cup \set{\silentaction}$, then we set $\run' = (\loci{0}, \clockval'_0), (d'_0, \edgei{0}), (\loci{1}, \clockval'_1), \ldots, (\loci{n}, \clockval'_n )$ with $d'_i = \tweak{f}{f'}(\sum\limits_{j=0}^{i} d_j) - \tweak{f}{f'}(\sum\limits_{j=0}^{i-1} d_j)$; $\clockval'_0 = \clockval_0$ and $\clockval'_{i+1} = \reset{\clockval'_i + d'_i}{\resets_i}.$

As $w \simeq w'$, both have the same untimed word, so $\run$ and~$\run'$ take a sequence of transitions whose labels form the same untimed word.
Moreover, if the action~$a_i$ is not~$\silentaction$, then we have $\sum\limits_{j=0}^{i} d'_j = \tweak{f}{f'}(\sum\limits_{j=0}^{i} d_j) = \tweak{f}{f'}(\tau_i) = \tau'_i$.
So $\run'$ has trace~$w'$.

The run~$\run'$ visits the same locations and takes the same transitions as~$\run$.
As we just proved above, all guards and invariants verified by the clock valuation within~$\run$ at a time $t \in \setRgeqzero$ are also verified in~$\run'$ at time $\tweak{f}{f'}(t)$ by the corresponding clock valuation.
Hence $\run'$ is a run of~$\TA$ of trace~$w'$.
\end{proof}

As a consequence of this result, no timed automaton can distinguish two equivalent words. 
We now show the converse result: given a timed word~$w$, there exists a timed automaton whose language is exactly the set of timed words equivalent to~$w$.

\begin{prop}\label{prop:folklore:rec}
The equivalence class of a timed word is recognizable by a~TA.
\end{prop}
\begin{proof}
Let $w = (a_1, \tau_1) \dots (a_n, \tau_n)$ be a timed word, and $[w]$ the set of timed words equivalent to~$w$.
We define the timed automaton $\TA_{[w]}$ as follow:

\begin{itemize}
\item The set of locations is $L = \set{\loc_i \mid i \in \interval{0}{|w|}}$;
\item The set of clocks is $\ClockSet = \set{\clock_i \mid i \in \interval{0}{|w|-1}}$;
\item For each~$i$ in $\interval{1}{|w|}$, there is a unique transition $(\loc_{i-1}, a_{i}, g_i, R ,\loc_{i})$ from~$\loc_{i-1}$ to~$\loc_{i}$ with $R = \set{x_{i}}$ if $i < |w|$ and $R = \emptyset$ otherwise, and denoting $\tau_0 = 0$:
\[g_i = \bigwedge\limits_{0 \leq j < i} g_{i,j}\]
\noindent{}with $g_{i,j} = (x_j = \tau_i - \tau_j)$ if $\tau_i - \tau_j \in \setN$ and $g_{i,j} = \lfloor \tau_i - \tau_j \rfloor < x_j < \lceil \tau_i - \tau_j \rceil$ otherwise.
\end{itemize}

Note that the timed word~$w$ is accepted by~$\TA_{[w]}$ (the value of each clock $\clock_j$, for $j < i$, is exactly $\tau_i - \tau_j$ before taking the transition labelled by~$a_i$).
We now show that the language of~$\TA_{[w]}$ is exactly~$[w]$.

Let $w'= (a'_1, \tau'_1) \dots (a'_n, \tau'_m)$ be a timed word produced by a run of~$\TA_{[w]}$.
First, as every run of~$\TA_{[w]}$ must take the same sequence of transitions,
$w'$~has the same untimed word as~$w$, in particular~$n=m$.
Moreover, the constraints $g_{i,0}$ compel the upper and lower integral parts of~$\tau'_i$  to be the same as for~$\tau_i$.
Consequently, the fractional parts of the timestamps in~$w'$ are zero if and only if the corresponding ones in~$w$ are so.

Finally, the constraints~$g_{i,j}$ on~$x_j$ determine the order of the fractional values  of~$\tau'_i$ and~$\tau'_j$.
Indeed, assuming that $\fract{\tau_i}= \fract{\tau_j}$,
then $\tau_i-\tau_j\in \setN$ and thus~$g_{i,j}$ ensures that $\tau'_i=\tau'_j$.
If $\fract{\tau_i}> \fract{\tau_j}$,
$g_{i,j}$ implies that 
\begin{align*}
&\lfloor \tau_i - \tau_j \rfloor <\tau'_i-\tau'_j < \lceil \tau_i - \tau_j \rceil \\
\iff&
0  <\fract{\tau_i}-\fract{\tau_j}< 1\\
\end{align*}
which ensures that $\fract{\tau'_i}> \fract{\tau'_j}$; and similarly, if $\fract{\tau_i}< \fract{\tau_j}$,
$g_{i,j}$ implies that $\fract{\tau'_i}< \fract{\tau'_j}$.

Therefore, $w'\simeq w$.
We have thus shown that every word accepted by~$\TA_{[w]}$ is equivalent to~$w$.
We can conclude thanks to \cref{lem:folklore:equivalence} that the traces of~$\TA_{[w]}$
are exactly those which timed word is in~$[w]$.
\end{proof}
\begin{figure}[tb]
   \centering
	\begin{tikzpicture}[pta, node distance=2cm, thin]
	
	\node[location, initial] (l0) at (-7.5, 0) {$\loci{0}$};
	\node[location] (l1) at (-5, 0) {$\loci{1}$};
	\node[location] (l2) at (-2.5, 0) {$\loci{2}$};
	\node[location] (l3) at (0, 0) {$\loci{3}$};
	\node[location] (l4) at (2.5, 0) {$\loci{4}$};
	\node[location, final] (l5) at (5, 0) {$\loci{5}$};

	\path
	(l0) edge node[align=center, above]{$0<\styleclock{\clock_0}<1$} node[align=center, below]{$\styleact{a}$\\$\styleclock{\clock_1} \assign 0$} (l1)
	(l1) edge node[align=center, above]{$\styleclock{\clock_0}=1$\\$\land~0<\styleclock{\clock_1}<1$} node[align=center, below]{$\styleact{b}$\\$\styleclock{\clock_2} \assign 0$} (l2)
	(l2) edge node[align=center, above]{$\styleclock{\clock_0}=1$\\$\land~0<\styleclock{\clock_1}<1$\\$\land~\styleclock{\clock_2}=0$} node[align=center, below]{$\styleact{c}$\\$\styleclock{\clock_3} \assign 0$} (l3)
	(l3) edge node[align=center, above]{$1<\styleclock{\clock_0}<2$\\$\land~0<\styleclock{\clock_1}<1$\\$\land~0<\styleclock{\clock_2}<1$\\$\land~0<\styleclock{\clock_3}<1$} node[align=center, below]{$\styleact{d}$\\$\styleclock{\clock_4} \assign 0$} (l4)
	(l4) edge node[align=center, above]{$3<\styleclock{\clock_0}<4$\\$\land~3<\styleclock{\clock_1}<4$\\$\land~2<\styleclock{\clock_2}<3$\\$\land~2<\styleclock{\clock_3}<3$\\$\land~2<\styleclock{\clock_4}<3$} node[align=center, below]{$\styleact{e}$} (l5)
	;
	\end{tikzpicture}
  \caption{The TA $\TA_{[w]}$ for $w = (a, 0.3) (b, 1) (c,1) (d,1.1) (e, 3.7)$}
  \label{figure:Aw}
\end{figure}
\begin{exa}
	The timed automaton $\TA_{[w]}$ whose language is the equivalence class of the timed word $w = (a, 0.3) (b, 1) (c,1) (d,1.1) (e, 3.7)$ is pictured in \cref{figure:Aw} to illustrate the construction presented in the proof of \cref{prop:folklore:rec}.
\end{exa}
\subsection{First $N$ observations}\label{sec:Nfirst}
We consider an attacker's strategy that involves wiretapping the first actions of the system until their entire budget is exhausted, allowing for a total of $N$~observations.
This setting can model an attack on a system whose security protocol prohibits any further access after a predefined number of suspected infiltrations.

Formally, this strategy is represented by the $N$-finite time selection function 
$\strategy{N}$ that associates $0$ to each timed word of length strictly smaller than~$N$ (hence authorizing the next observation), and $+\infty$ to every other words.
We thus adapt the opacity problems with respect to $\strategy{N}$.
\vspace*{0.3cm}

\defProblem{Weak / full / $\exists$-opacity \wrt{} first $N$ observations decision problem}
{A TA~$\TA$, $N\in \setN$}
{Is $\TA$ weakly / fully / $\exists$-opaque against $\strategy{N}$?}

\begin{exa}
If $w = (a,1.2)(b,1.4)(b,1.5)(a,2.1)$ is the trace of a public run of the system, and $N = 2$, then the attacker only observes the trace $\projection{2}(w) = (a,1.2)(b,1.4)$.
If $v = (a,1.2)(b,1.4)(c,1.6)$ is the trace of a private run, the trace observed by the attacker
is $\projection{2}(v) =(a,1.2)(b,1.4)$ again, and the attacker cannot conclude whether a private run occurred or not.
\end{exa}

To formally represent this framework, given a TA~$\TA$, we define a new TA (depicted in \cref{figure:finite-obs:unfolding}) which emulates the behaviour of~$\TA$ up to the \mathnth{N} observation.
This TA is an unfolding of~$\TA$ with $N+1$ copies of~$\TA$, where
$\silentaction$-transitions are taken within each copy, and transitions with an observable
action lead to the next copy.
A run ends when either a final location or the final copy is reached.

\begin{defi}[$N$-observation unfolded TA]

	Let $\TA = \TAprivextend$ be a TA and let $N \in \setN$.
	We call \emph{$N$-unfolded TA} of~$\TA$ the TA $\Unfold{\TA} = (\ActionSet, \LocSet', \loci{0}^0, \PrivSet', \FinalSet', \ClockSet, \invariant', \EdgeSet')$ where
	\begin{enumerate}
		\item $\LocSet' = \bigcup\limits_{i=0}^{N} \LocSet^i$ where the sets $\LocSet^i$ are $N+1$ disjoint copies of $\LocSet$ where each location $\loc \in \LocSet$ has been renamed $\loc^i$ in~$\LocSet^i$: that is, for $0 \leq i \leq N$, $\LocSet^i = \lbrace \loc^i \mid \loc \in \LocSet \rbrace$;
		\item $\loci{0}^0 \in \LocSet^0$ is the initial location;
		\item $\PrivSet' = \bigcup\limits_{i=0}^{N-1} \PrivSet^i$
where $\PrivSet^i$ are the copies within $\LocSet^i$ of the private locations of~$\TA$;
		\item $\FinalSet' = (\bigcup\limits_{i=0}^{N} \FinalSet^i)\cup \LocSet^N$ where $\FinalSet^i$ are the copies within $\LocSet^i$ of the final locations of~$\TA$;
		\item $\invariant' (\loc^i) = \invariant (\loc)$ for $l\in \LocSet$ and $i\leq N$; %
		\item $\EdgeSet' = \bigcup\limits_{i=0}^{N-1} \EdgeSet^i \cup  \EdgeSet^{i \rightarrow i+1}$ is the set of transitions where, given $0 \leq i<N$
		\begin{itemize}
		\item $\EdgeSet^i = \big\{ (\loc^i, \silentaction, g, R, \loc'^i) \mid (\loc, \silentaction, g, R, \loc') \in \EdgeSet \big\}$, and
		\item  $\EdgeSet^{i \rightarrow i+1} = \big\{ (\loc^i, a, g, R, \loc'^{i+1}) \mid (\loc, a, g, R, \loc') \in \EdgeSet \land a \in \ActionSet \big\}$.
		\end{itemize}
	\end{enumerate}
\end{defi}

\begin{figure}[tb]
   \centering
	\begin{tikzpicture}[pta, node distance=2cm, thin]

	\node[location, initial] (l0) at (-3, 0) {$\locinit^0$};
	\node[location] (li) at (-2.25, 0.5) {$\loci{i}^0$};
	\node[location] (lj) at (-1, 0.25) {$\loci{j}^0$};
	\node[rectangle, minimum width=3cm, minimum height=2cm, align=center, draw, black!60] (rectangleA) at (-2cm, 0cm) {$\TA_0$};
	\node[location, final] (lf0) at (-1.5,-0.5) {$\locfinal^0$};

	\node[location] (lk) at (+1.25, 0) {$\loci{k}^1$};
	\node[rectangle, minimum width=3cm, minimum height=2cm, align=center, draw, black!60] (rectangleB) at (2cm, 0cm) {$\TA_1$};
	\node[location, final] (lf1) at (2.5,-0.5) {$\locfinal^1$};

	\node (ln) at (4,0) {\dots};
	\node[location,final] (lm) at (+5, 0) {$\loci{m}^N$};
	\node[rectangle, minimum width=3cm, minimum height=2cm, align=center, draw, black!60] (rectangleB) at (6cm, 0cm) {$\TA_N$};

	\path
	(li) edge node[align=center, above]{$\styleact{\silentaction}$} (lj)
	(lj) edge node[align=center, below]{$\styleact{a_1}$%
	} (lk)
	(ln) edge[bend right] node[align=center, below]{\ \ \ \ \ $\styleact{a_N} $ %
	} (lm)
	;
	\end{tikzpicture}
  \caption{The construction of an $N$-observation unfolded TA}
  \label{figure:finite-obs:unfolding}
\end{figure}

The set of traces of the resulting automaton $\Unfold{\TA}$ is exactly the set of traces of~$\TA$ truncated after the \mathnth{N} letter, \ie{} $\projection{N}(\Trace{\TA})$.
Indeed, all runs from~$\TA$ also exist in~$\Unfold{\TA}$ but the number of the current copy increases by one each time a letter is observed during the run, and once it has reached~$N$, all observable transitions of~$\TA$ become non-observable so that only the first $N$~letters of each word can be seen.

Testing weak/full/$\exists$-opacity of a TA $\TA$ against the time selection 
function~$\strategy{N}$ is thus equivalent to testing weak/full/$\exists$-opacity
of $\Unfold{\TA} $. This remark allows us to handle $\exists$-opacity.

\begin{prop}
The $\exists$-opacity \wrt{} first $N$ observations decision problem is \PSPACE{}-complete.
\end{prop}

\begin{proof}
By \cref{theorem:decidability-exists-opacity} and the previous remark, the $\exists$-opacity problem against time selection function~$\strategy{N}$ for $\TA$ is inter-reducible
with the reachability problem for $\Unfold{\TA}$.
While $\Unfold{\TA}$ is exponentially larger than $\TA$, this blow up is limited to the number of states (and crucially not the number of clocks), meaning that 
applying the standard algorithm for reachability on $\Unfold{\TA}$ still provides
a \PSPACE{} algorithm.

Concerning \PSPACE{}-hardness, we adapt the reduction of \cref{figure:exists-opacity:PSPACE-hardness} to a single observation.
Formally, let $\TA = (\ActionSet, \LocSet, \locinit, \set{\locfinal}, \ClockSet, I, E)$ be a timed automaton.
Let $\TA' = (\set{a}, \LocSet \cup \setsmall{\loc_{\text{pub}}, \loc_{\text{priv}}, \locfinal'}, \locinit, \setsmall{\locfinal'}, \ClockSet, I', E')$ be the timed automaton (see \cref{fig:firstN:hardness-exists}) with
	\begin{itemize}
	\item $I'$ an extension of~$I$ such that $I'(\loc_{\text{pub}}) = I'(\loc_{\text{priv}}) = I'(\locfinal') = \BTrue{}$;
	\item $E' = \setsmall{(\loc, \silentaction, g, R, \loc') \mid \exists b \in \ActionSet~ (\loc, b, g, R, \loc') \in E} \\ \cup \setsmall{(\locfinal, \silentaction, \BTrue{}, \emptyset, \loc_{\text{pub}}), (\locfinal, \silentaction, \BTrue{}, \emptyset, \loc_{\text{priv}}), (\loc_{\text{pub}}, a, \BTrue{}, \emptyset, \locfinal'),(\loc_{\text{priv}}, a, \BTrue{}, \emptyset, \locfinal')}$.
	\end{itemize}
	\begin{figure}[h]
\begin{center}
\begin{tikzpicture}[pta, node distance=2cm, thin]
	\node[location, initial, rectangle, minimum width=3cm, minimum height=2cm] (l0) at (0, 0) {$\TA$};
	\node[location] (lf) at (1, 0) {$\locfinal$};
	\node[location] (lpub) at (2.5,1) {$\loc_{\text{pub}}$};
	\node[location] (lpriv) at (2.5, -1) {$\loc_{\text{priv}}$};
	\node[location, final] (lf') at (4, 0) {$\locfinal'$};
	\path (lf) edge node[above] {$\styleact{\silentaction}$} (lpub);
	\path (lf) edge node[below] {$\styleact{\silentaction}$} (lpriv);
	\path (lpub) edge node[above] {$\styleact{a}$} (lf');
	\path (lpriv) edge node[below] {$\styleact{a}$} (lf');
\end{tikzpicture}
\end{center}
	\caption{Reduction from reachability to $\exists$-opacity against $\sigma_1$}
	\label{fig:firstN:hardness-exists}
	\end{figure}
Then, $\TA'$ is $\exists$-opaque with respect to 1 observation if and only if $\locfinal$ is reachable in~$\TA$.
Indeed, the runs reaching~$\locfinal'$ are obtained by extending the runs reaching~$\locfinal$, which can be done both publicly (through $\loc_{\text{pub}}$) or privately (through $\loc_{\text{priv}}$) with the same observation.
This ensures that $\TA'$ is fully opaque. 
In order to be $\exists$-opaque, it however needs that there is at least one run reaching~$\locfinal'$ which is equivalent to reachability of~$\locfinal$.
\end{proof}

When we defined the tick automaton in \cref{section:opacity:discrete-time}, we added a tick clock which counted the time elapsed since the start of the run, and thanks to it we got within the untimed trace of the runs the integral part of all timestamps of observed actions.
By completing this with the order of the timestamps' fractional parts, one can get enough information in the untimed trace to determine the equivalence class of the associated timed words. However, to keep track of this order along the run in the~TA, we need to use one clock per observation, which is reset whenever it reaches~1.
We thus cannot do this for an unbounded number of observations, which explains why our construction solves the opacity problem for the first $N$~observations, but not for an unbounded number of observations---which is undecidable~\cite{Cassez09}.
We then add a gadget plugged after each final location of the unfolded~TA to make the outcome of this ordering of timestamps fractional parts appear explicitly in the untimed trace of the run.
This extends the run by less than two time units.
In the gadget, we make visible the (possibly simultaneous) reset of the added clocks with a tick alphabet $\setsmall{f_K \mid K \subseteq \ClockSet \land K \neq \emptyset}$, ``$f$'' standing for ``fractional part''.

This way, a suffix with untimed word $f_{\setsmall{0,2}}f_{\setsmall{3}}f_{\setsmall{1}}$ can be read after the observations $(a, \tau_1), (b, \tau_2) , (c, \tau_3)$ if and only if $0 = \fract{\tau_2} < \fract{\tau_3} < \fract{\tau_1}$.
With the global time tick clock labelling the integral times (as in \cref{section:opacity:discrete-time} for the discrete time setting), the timed word $(a,1.4)(b,2)(c,5.1)$ now produces the untimed word $tatbtttctf_{\setsmall{0,2}}f_{\setsmall{3}}f_{\setsmall{1}}$.

We denote by~$\ClockSet_K$ the set of added clocks with indices in $K \subseteq \interval{0}{N}$: $\ClockSet_K = \setsmall{x_k \mid k \in K}$.
We also denote for each $J \subseteq \interval{0}{N}$:
\begin{align*}
	g_J 			& = \bigwedge\limits_{x \in \ClockSet_J}(x = 1) \land \bigwedge\limits_{y \in \ClockSet_{\interval{1}{N}\setminus J}}(0<y<1) \\
	g^{(0)}_J 		& = \bigwedge\limits_{x \in \ClockSet_J}(x = 1) \land \bigwedge\limits_{y \in \ClockSet_{\interval{0}{N}\setminus J}}(0<y<1) \\
	g_{<1}			& = \bigwedge\limits_{x \in \ClockSet_{\interval{1}{N}}}(x < 1) \\
	g^{(0)}_{<1}	& = \bigwedge\limits_{x \in \ClockSet_{\interval{0}{N}}}(x < 1)
\end{align*}

\begin{defi}[The ``tick'' construction]\label{def:tick-construction}
	Let $\TA$ be a TA, $N \in \setN$ and let $\Unfold{\TA} = (\ActionSet, \LocSet', \loci{0}^0, \PrivSet', \FinalSet', \ClockSet, \invariant', \EdgeSet')$ be the $N$-unfolded TA of~$\TA$.
	We define the ``tick'' construction $\TickN{\TA} = (\ActionSet', \LocSet'', \loci{0}^0, \PrivSet', \setsmall{\loc_{G1}}, \ClockSet', \invariant', \EdgeSet'')$ where
	\begin{enumerate}
		\item $\ActionSet' = \ActionSet \cup \setsmall{t} \cup \ActionSet_N$ where $\ActionSet_N = \set{f_K \mid K \subseteq \interval{0}{N}, K \neq \emptyset}$ is the alphabet labelling the transitions resetting the set of added clocks~$\ClockSet_{\interval{0}{N}}$ (see below);
		\item $\LocSet'' = \LocSet' \cup \setsmall{\loc_{G0}, \loc_{\textit{G1}}}$
		\item $\ClockSet' = \ClockSet \cup \ClockSet_{\interval{0}{N}}$ where $\ClockSet_{\interval{0}{N}}$ is the set of added clocks: $x_0$ is the global time ``tick'' clock and each other $x_i$ is associated with the \mathnth{i} observation;
		\item $\EdgeSet'' = \bigcup\limits_{i=0}^{N-1} \EdgeSet^i \cup  \EdgeSet^{i \rightarrow i+1} \cup \EdgeSet_t \cup \EdgeSet_{\gamma} \cup \EdgeSet_G$ is the set of transitions where, given $0 \leq i<N$
		\begin{itemize}
		\item $\EdgeSet^i = \big\{ (\loc^i, \silentaction, g \land g^{(0)}_{<1}, R, \loc'^i) \mid (\loc^i, \silentaction, g, R, \loc'^i) \in \EdgeSet' \big\}$;
		\item  $\EdgeSet^{i \rightarrow i+1} = \big\{ (\loc^i, a, g \land g^{(0)}_{<1}, R \cup \lbrace x_{i+1} \rbrace, \loc'^{i+1}) \mid (\loc^i, a, g, R, \loc'^{i+1}) \in \EdgeSet' \big\}$;
		\item $ \EdgeSet_t = \big\{ (\loc, t, x_0 = 1, \setsmall{x_0}, \loc) \mid \loc \in \LocSet' \setminus \FinalSet' \big\}$;
		\item $ \EdgeSet_{\gamma} = \big\{ (\loc, \silentaction, g_I, \ClockSet_I, \loc) \mid \loc \in \LocSet' \land I \subseteq \interval{1}{N} \land I \neq \emptyset \big\}$;
		\item $ \EdgeSet_{G} = \big\{ (\loc_f, f_{K \cup \setsmall{0}}, g^{(0)}_{K \cup \set{0}}, \ClockSet_{K \cup \setsmall{0}}, \loc_{G0}) \mid \loc_f \in \FinalSet' \land K \subseteq \interval{1}{N} \big\} \cup
		\big\{ (\loc_{G0}, f_K, g^{(0)}_K, \ClockSet_K, \loc_{G0}) \mid K \subseteq \interval{1}{N} \big\} \cup
		\big\{ (\loc_{G0}, \silentaction, g^{(0)}_{K \cup \set{0}}, \emptyset, \loc_{G1}) \mid K \subseteq \interval{1}{N} \big\} $.
		\end{itemize}
	\end{enumerate}
\end{defi}

Note that the global time ``tick'' clock~$x_0$ is reset exactly every time unit, and so are the other added clocks since their first reset, occuring with the corresponding observation. Hence the fractional part of the timestamp of each observation is encoded by the associated added clock. 

See \cref{figure:finite-obs:tick} in \cref{appendix:tick-construction} for an illustration of the ``tick'' construction.

\begin{defi}[Ticked word]
Let $w = (a_1, \tau_1) (a_2, \tau_2) \dots (a_N, \tau_N)$ be a timed word.
Then the \emph{ticked word} of~$w$ is
\[\Tick{w} = t^{\intpart{\tau_1}} a_1 t^{\intpart{\tau_2} - \intpart{\tau_1}} a_2 \dots t^{\intpart{\tau_{N}} - \intpart{\tau_{N-1}}} a_N f_{K_0} f_{K_1} f_{K_2} \dots f_{K_m}\]
where $m = \vert \setsmall{ \fract{\tau_i} \mid i \leq N} \setminus \set{0}\vert$ (several observation timestamps may have the same fractional part, in which case $m < N$), $K_i = \set{i \in \interval{0}{N} \mid \fract{\tau_i} = \gamma_i}$ with $\gamma_0 = 0$ and $\gamma_i = \min \big(\set{\fract{\tau_j} \mid j \in \interval{1}{N}} \setminus \set{\gamma_k \mid k < i} \big)$ for $i \in \interval{0}{m}$.
\end{defi}

This definition extends to timed languages, defining what we call \emph{ticked languages}.
The untimed language of $\TickN{\TA}$ is $\Tick{\projection{N}(\Trace{\TA})}$, and is the ticked language of $\projection{N}(\Trace{\TA})$.
Indeed, the guard $g^{(0)}_{<1}$ on each transition that was already in~$\Unfold{\TA}$ ensures that the clocks~$x_i$ (for $0 \leq i \leq N$) have a value strictly less than~$1$.
This does not affect the behaviour of the system since the new loops with guard~$g_I$ (for each $I \subseteq \interval{1}{N}$) allow to reset these clocks as soon as they reach~$1$, and this is also done for the clock~$x_0$ by the $t$-labelled transitions.
This way, the letter~$t$ occurs only at integral times.
When the current copy of~$\TA$ is the copy number~$i$, the clocks~$x_j$ for $j > i$ are all synchronized with~$x_0$ since their only resets occurred at the same time.
The clocks $x_k$ for $k \leq i$ are equal to the fractional part of the time elapsed since the~\mathnth{k} observation because they are reset during this observation and then only at integral numbers of time units after that.

\begin{exa}
See in \cref{table:example-computation} the computation of the ticked word associated with the timed word $(a,1.2)(b,1.5)(c,2)(d,2.3)$, with the corresponding clock valuations in the ``tick'' construction for $N = 4$.
Note that this construction solely depends on the timed word, not on a given~TA, which explains why we do not need to assume any particular TA in this example.

\begin{table}
\caption{Computation of the ticked word associated with the timed word $(a,1.2)(b,1.5)(c,2)(d,2.3)$}
\label{table:example-computation}
\[\begin{array}{c|c|c|c|c|c|c|c}
\rowcolor{USPNsable}\textbf{global time} & x_0 & x_1 & x_2 & x_3 & x_4 & \textbf{ticked word} & \textbf{copy id} \\
\hline
0				   & 0	 &  0  &  0  &  0  &  0  &    & 0\\
1				   & 0	 &  0  &  0  &  0  &  0  &  t & 0\\
1.2				   & 0.2 &  0  & 0.2 & 0.2 & 0.2 &  a & 0\rightarrow 1\\
1.5				   & 0.5 & 0.3 &  0  & 0.5 & 0.5 &  b & 1\rightarrow 2\\
2				   & 0   & 0.8 & 0.5 &  0  &  0  &  tc& 2,2\rightarrow 3\\
2.2				   & 0.2 &  0  & 0.7 & 0.2 & 0.2 &    & 3\\
2.3				   & 0.3 & 0.1 & 0.8 & 0.3 &  0  &  d & 3 \rightarrow 4\\
2.5				   & 0.5 & 0.3 &  0  & 0.5 & 0.2 &    & 4 \\
3				   & 0   & 0.8 & 0.5 &  0  & 0.7 &  f_{\set{0,3}} & 4 \\
3.2				   & 0.2 &  0  & 0.7 &  0  & 0.2 &  f_{\set{1}} & 4 \\
3.3				   & 0.3 & 0.1 & 0.8 & 0.1 & 0.3 &  f_{\set{4}} & 4 \\
3.5				   & 0.5 & 0.3 &  0  & 0.3 & 0.5 &  f_{\set{2}} & 4 \\
\end{array}\]
\end{table}

\end{exa}

The idea is to bring enough information into the untimed traces to solve opacity in a finite (untimed) automaton.
As we prove in the following (\cref{lem:untimedtick}), this ticked trace corresponds exactly to the equivalence class of the timed word~$w$, with the relation~$\simeq$ presented in \cref{sec:techtools}.
A timed language being a union of such equivalence classes (\cref{lem:folklore:equivalence}), we can compare the languages $\projection{N}(\PrivateTr{\TA})$ and $\projection{N}(\PublicTr{\TA})$ by comparing their associated ticked languages.

Let $\Regions{\TickN{\TA}}$ be the region automaton of this automaton, whose language is exactly $\Tick{\projection{N}(\TA)}$, the untimed projected traces of~$\TA$ with~$\projection{N}$ enriched with timing information given by the ticks.

Thanks to these added ticks, paths of $\Regions{\TickN{\TA}}$ sharing the same trace correspond to runs of~$\TA$ for which the (at most) $N$ observations occurred within the same time intervals (due to the ticks representing the total time), and for which the fractional part of the timing of those $N$~observations have the same order (given by the sequence of~$f_i$ added at the end of the trace).

\begin{lem}\label{lem:untimedtick}
Let $w$ and~$v$ be two timed words.
Then $w \simeq v$ if and only if $\Tick{w} = \Tick{v}$.
\end{lem}

\begin{proof}

Let $w$ and~$v$ be two timed words.
If $|w| \neq |v|$ then $w$ and $v$ are not equivalent and $\Tick{w} \neq \Tick{v}$, so we can suppose now that $|w| = |v| = N$.
For convenience, we set $\tau_0^w = \tau_0^v = 0$.

For $u \in \set{w, v}$ we denote $u = (a_1^u, \tau_1^u) \dots (a_N^u, \tau_N^u)$ and we have
\[ \Tick{u} = t^{\intpart{\tau_1^u}} a_1^u t^{\intpart{\tau_2^u} - \intpart{\tau_1^u}} a_2^u \dots t^{\intpart{\tau_N^u} - \intpart{\tau_{N-1}^u}} a_N^u f_{K^u_0} f_{K^u_1} f_{K^u_2} \dots f_{K^u_{m^u}}\]
with $(K_i)_{0 \leq i \leq m^u}$ the partition of $\interval{0}{N}$ such that for all $j \in K_i^u$ and $k \in  K_l^u$, we have $\fract{\tau_j^u} \leq \fract{\tau_k^u}$ if and only if $i \leq l$. 

Then, according to the definitions:
\begin{equation*}
\begin{split}
w \simeq v & \iff \left\{ \begin{array}{l} \untimed{w} = \untimed{v} \\ \land \forall i \in \interval{1}{N} , \intpart{\tau_i^w} = \intpart{\tau_i^v} \\ \land \forall i \in \interval{1}{N} , \fract{\tau_i^w} = 0 \iff \fract{\tau_i^v} = 0 \\ \land \forall i,j \in \interval{1}{N} , \fract{\tau_i^w} \leq \fract{\tau_j^w} \iff \fract{\tau_i^v} \leq \fract{\tau_j^v} \end{array}\right.\\
& \iff \left\{ \begin{array}{l} \forall i \in \interval{1}{N} , a_i^w = a_i^v \\ \land \forall i \in \interval{1}{N} , \intpart{\tau_i^w} - \intpart{\tau_{i-1}^w} = \intpart{\tau_i^v} - \intpart{\tau_{i-1}^v} \\ \land K_0^w = K_0^v \\ \land \forall i \in \interval{1}{m^v} , K_i^w = K_i^v \end{array}\right. \\
& \iff \Tick{w} = \Tick{v} \qedhere
\end{split}
\end{equation*}%
\end{proof}

\begin{lem}\label{cor:inclusion-tick}
Timed language inclusion is equivalent to ticked language inclusion.
\end{lem}

\begin{proof}
Let $L$ be a timed language with alphabet $\ActionSet$, and let $s \in \Tick{L}$.
According to \cref{lem:untimedtick}, there is a unique timed word equivalence class~$W$ such that $s$ is the ticked word of all timed words in~$W$.
As established in \cref{lem:folklore:equivalence} a timed language is a disjoint union of such word equivalence classes.
So $W \subseteq L$, and therefore we have
\[ L = \biguplus\limits_{s \in \Tick{L}} \setsmall{w \in \TimedWords{\ActionSet} \mid \Tick{w} = s}\text{.}\]
It follows that if $L'$ is another timed language, $\Tick{L} \subseteq \Tick{L'}$ if and only if $L \subseteq L'$. 
\end{proof}

Recall from \cref{definition:N-bounded} that a language~$L$ is $N$-bounded if no word it contains has a length greater than~$N$.
Since the private and public languages of an $N$-unfolded TA are $N$-bounded, we address in \cref{prop:Nfirst:inclusion} the $N$-bounded language inclusion problem.
This result immediately answers weak and full opacity, as
the reductions of \cref{sec:reducWincl} still hold when considering $N$-bounded languages,
and thus inclusion (resp.\ equality) of $N$-bounded languages and weak (resp.\ full) opacity against the time selection function $\strategy{N}$ are inter-reducible.

We can decide $N$-bounded language inclusion by using the fact that the tick language of a bounded language is regular and finite. Instead of proceeding naively for this computation and get a \twoEXPSPACE{} algorithm, we present an algorithm with optimal theoretical complexity.

\begin{thm}\label{prop:Nfirst:inclusion}
Inclusion of $N$-bounded languages is a \coNEXPTIME{}-complete problem.
\end{thm}

\begin{cor}\label{cor:Nfirst:w-f-opacity}
Given $N \in \setN$, the weak and full opacity \wrt{} first $N$ observations decision problems are \coNEXPTIME{}-complete.
\end{cor}

The rest of this subsection is dedicated to the proof of \cref{prop:Nfirst:inclusion}.

\subsubsection{\coNEXPTIME{}-membership}

We present an algorithm with complexity \NEXPTIME{} to decide whether the language inclusion of two $N$-bounded timed languages, given by some TAs $\TA$ and~$\TB$, does \emph{not} hold.

First we build the ``ticked'' automata $\Tick{\TA}$ and $\Tick{\TB}$ of $\TA$ and~$\TB$.
They are exponential in size since $N$~is encoded in binary.
The region automata $\mathcal{R}_{\TA}^{\textit{Tick}}$ and $\mathcal{R}_{\TB}^{\textit{Tick}}$ of the respective TAs $\Tick{\TA}$ and $\Tick{\TB}$ define some rational languages $T_{\TA}$ and~$T_{\TB}$.
Thanks to \cref{cor:inclusion-tick}, showing $\Trace{\TA} \not\subseteq \Trace{\TB}$ amounts to show $T_{\TA} \not\subseteq T_{\TB}$.

However, in order to keep an exponential complexity, we cannot build the region automata in their entirety (as a region automaton is doubly exponential in the problem data size).
So, to study language inclusion, we have to be careful in our exploration of the region automata, and in particular bound the number of regions that are necessary to visit in order to compare languages.
We do this through an analysis of the shape of one of the smallest words separating the two languages.
In \cref{lemma:N-obs:weak-full:region-number}, we bound the number of regions that can appear while reading this word, then we provide in \cref{lemma:N-obs:weak-full:word-size} a doubly exponential bound to its length, and a representation of simply exponential size.
After that, we ``guess'' a word of the right shape using non-determinism, and we only have to test whether it belongs to~$T_{\TA}$ but not to~$T_{\TB}$.

\paragraph{Claim}
Due to the ``tick'' construction, the words accepted by the region automata are of the form
\begin{equation}
t^{k_1} a_1 t^{k_2} a_2 \dots t^{k_n} a_n f_{K_0} \dots f_{K_m}
\label{word-structure}
\end{equation}
for some integer $n \leq N$; with for all $i \in \interval{1}{n}$, $a_i \in \ActionSet$, $k_i \in \setN$ and $(K_i)_{0 \leq i \leq m}$ is a partition of~$\interval{0}{n}$ such that $0 \in K_0$.

\begin{lem}\label{lemma:N-obs:weak-full:region-number}
Let $\TA = (\ActionSet, \LocSet, \locinit, \set{\locfinal}, \ClockSet, I, E)$ be a TA. We denote $\ClockCard = \vert \ClockSet \vert$. Let $M$ be the greatest constant appearing in its guards and invariants. Let $\TickN{\TA}$ be the ``tick'' construction applied on the $N$-unfolding of~$\TA$ for some $N \in \setN$, and $\mathcal{R}_{\textit{tick}}$ be its region automaton.
Then the number of regions that can be visited while reading
a word of the form \cref{word-structure} in $\mathcal{R}_{\textit{tick}}$ is bounded by
\[B = (N+1)^3 \cdot \vert \LocSet \vert \cdot (2N + 3)\cdot (2M+2)^\ClockCard  \cdot 2^\ClockCard \cdot (N+\ClockCard+1)^\ClockCard.\]
\end{lem}

\begin{proof}
The tick unfolding being equipped with $\ClockCard+N+1$ clocks, a region is defined by
\begin{enumerate}
\item \emph{a location}: the total number of locations is bounded by $(N+1) \cdot | \LocSet |$;
\item \emph{an interval ($\set{m}$, $(m;m+1)$ or $(M, \infty)$) for each clock (of $\ClockSet \cup \ClockSet_t$) with $m \leq M$}: naively at most $(2M+2)^{\ClockCard+N+1}$ possibilities;
\item \emph{an ordering of the fractional parts of all clocks}: naively at most $2^{\ClockCard+N+1} \cdot (\ClockCard+N+1)!$ possibilities.
\end{enumerate}

The bound on the number of intervals and ordering possibilities given above can be refined to a polynomial number in~$N$ by relying on the structure of the words accepted by the region automaton---which we explain in the following.

\paragraph{On intervals.}
Noting that added clocks have values limited by~1, we can immediately reduce the bound to
$3^{N+1}\cdot (2M+2)^\ClockCard$ (an added clock can be valuated in $\set{0}$, in $(0;1)$ or in~$\set{1}$).
A slightly finer analysis of the added clocks gives a better bound than~$3^{N+1}$:
the ticks appearing in the word~$w$ group together the clocks having the same fractional part, hence only clocks from the same group can simultaneously have an integer value.
Within each copy of the unfolding of the initial TA, this reduces the number of possibilities to at most $2(N+1)+1$ (the worst case being when no pair of clocks share the same fractional part). Indeed, the groups of clocks with same fractional part are fixed within the copy.
There are at most $N+1$ possibilities for the group with integer value, and these clocks can either be equal to~0 or to~1.
The last possibility is that all clocks are valuated in~$(0;1)$.
As there are $N+1$ copies, we can bound the number of intervals by $(N+1)(2(N+1)+1)$.

\paragraph{On orderings.}
Once again, we can reduce the bound on the number of orderings thanks to the information related to the added clocks appearing in the word.
In fact, the comparison of the fractional parts of the $N+1$ added clocks is entirely determined by the partition~$(K_i)$ given by at the end of the ticked word.
We thus first assume the added clocks have been ordered thanks to the word, thus we do not need to consider any permutation for the added clocks.
Let us now count the number of orderings due to the remaining $\ClockCard$~clocks.
To choose an ordering of fractional parts, we can select a permutation of the clocks and determine between which ones equality holds.
Let us first illustrate this using an example.

\begin{exa}\label{exa:clock-ordering}
Let us assume a fictional example with (at least) 5~added clocks.
Assume that $\fract{x_1} < \fract{x_0} = \fract{x_3} = \fract{x_4} < \fract{x_2}$ is an ordering of the fractional parts of the first 5~added clocks; this ordering involves two equalities.
There are 7~manners to add a clock $y$ to this ordering: $\fract{y}< \fract{x_1}$, $\fract{y} = \fract{x_1}$, $\fract{x_1} < \fract{y} < \fract{x_0}$, $\fract{y} = \fract{x_0}$, $\fract{x_0} < \fract{y} < \fract{x_2}$, $\fract{y} = \fract{x_2}$ or $\fract{x_2} < \fract{y}$.%
\end{exa}

We insert (as sketched in \cref{exa:clock-ordering}) the fractional parts of the $\ClockCard$~clocks into the existing ordering, one after the other.
When adding the \mathnth{i} clock, there are $2(N+i-k)+1$ positions for this clock compared to the already ordered clock, with $k$ being the number of equality signs in the ordering before adding the \mathnth{i} clock.
Indeed, either the fractional part of the new clock is equal to one of the already ordered ones, which corresponds to $N+ i - k$ choices, or it is inserted in between two clocks, and there are $N+i -k +1$ such choices.
In the worst case, $k = 0$ (all fractional parts are different).
Moreover, we can roughly bound~$m$ by~$\ClockCard$ and $2(N+\ClockCard)+1$ by $2(N+\ClockCard+1)$ to get a general bound of $2^\ClockCard \cdot (N+\ClockCard+1)^\ClockCard$ for the number of possibilities of ordering of the fractional parts of the $\ClockCard$ clocks compared to those, already ordered, of the $N+1$ added clocks.

The $N+1$ added clocks are reset each time they reach~1.
The ordering of their fractional parts, previously considered fixed to add those of the $\ClockCard$~clocks, is determined by the part of the word that has already been read.
Moreover, since the behaviour of added clocks is cyclic, all their possible orderings (for a given word) are the same up to rotation.
Hence we encounter at most $N+1$ different possibilities.

Additionally, the $N+1$ added clocks appear progressively with each of the $N$~letters and thus the orderings can change and correspond to different regions for each, so we multiply our bound by~$N+1$. This multiplication corresponds to the multiplication by the number of copies we needed for the bound on the number of orderings, and can in fact be shared for both orderings and intervals.

We deduce from what precedes that the number of reachable regions during the reading of~$w$ is bounded by
\[(N+1) \cdot \vert \LocSet \vert \cdot (2N + 3)\cdot (2M+2)^\ClockCard  \cdot 2^\ClockCard \cdot (N+\ClockCard+1)^\ClockCard \cdot (N+1) \cdot(N+1)\]
which is polynomial in~$N$.
This concludes the proof of \cref{lemma:N-obs:weak-full:region-number}.
\end{proof}

Moreover, given two equivalent timed words, the over-approximation of the set of regions computed above is the same. According to \cref{lem:folklore:equivalence}, equivalent words are indeed produced by runs crossing the same locations and which valuations verify the same constraints; this means in particular that the associated runs in the region automaton cross the same regions.

As we plan to guess the smallest word distinguishing $T_{\TA}$ and~$T_{\TB}$, it is necessary to first determine an upper bound for its size.
In the following lemma, we show that it is at most doubly exponential in the sizes of $\TA$, $\TB$ and~$N$; but that, by expressing the $k_i$~exponents in binary, we obtain an exponential size description.

\begin{lem}\label{lemma:N-obs:weak-full:word-size}
The smallest word distinguishing $T_{\TA}$ and~$T_{\TB}$ (which is of the form given in \cref{word-structure}) is of size at most doubly exponential in the size of $\TA$ and~$N$ (binary), and can be described in exponential size.
\end{lem}
\begin{proof}
We bound the length of the shortest word of~$T_{\TA}$ which is not in~$T_{\TB}$ with an idea reminiscent of the pumping lemma.
Reading a letter in $\Tick{\TA}$ or~$\Tick{\TB}$ corresponds to reaching a set of regions, in both automata.
We represent the sets of regions reachable after reading a word~$w$ by a pair of Boolean vectors $V_w = (V_w^{\TA}, V_w^{\TB})$ of respective dimension
$\vert L_{\mathcal{R}_{\TA}^{\textit{Tick}}} \vert$ and $\vert L_{\mathcal{R}_{\TB}^{\textit{Tick}}} \vert$ where $L_{\mathcal{R}_{\TA}^{\textit{Tick}}}$ and $L_{\mathcal{R}_{\TB}^{\textit{Tick}}}$ are the respective sets of (reachable) regions of $\mathcal{R}_{\TA}^{\textit{Tick}}$ and $\mathcal{R}_{\TB}^{\textit{Tick}}$ during the reading of~$w$.
Hence, for any word~$w$, denoting $B$ the bound computed in \cref{lemma:N-obs:weak-full:region-number}, the dimension of~$V_w$ is bounded by~$2B$.

By the bound on the size of the vectors~$V_w$ and the fact it is Boolean, it can take at most $2^{2B}$ different assignments.
Moreover, the acceptance of a word~$w$ is entirely determined by the vector~$V_w$ as it incorporates the presence of a final location.

If a word~$w$ distinguishing both region automata is longer than~$2^{2 B}$, it visits twice the same set of regions so we can remove the factor of the word responsible for this loop to get a shorter word which also distinguishes $T_{\TA}$ and~$T_{\TB}$.
The bound~$B$ is exponential in the size of the binary writing of~$N$ and in the size of~$\TA$, hence the length of the word $\mathcal{O}(N) + \sum\limits_{i=1}^{N} k_i$, bounded by~$2^{2B}$, is of doubly exponential size.
But choosing to represent the $k_i$ exponents in binary gives an exponential description.
\end{proof}

Using non-determinism, we guess the exponentially long description of a word~$w$ among those of the form given in \cref{word-structure}.
We will check that it distinguishes $T_{\TA}$ and~$T_{\TB}$ as expected.
We can do this by reading the word~$w$ on both region automata and testing whether final regions have been reached on both sides or not.
Since we treat both $\mathcal{R}_{\TA}^{\textit{Tick}}$ and~$\mathcal{R}_{\TB}^{\textit{Tick}}$ the same way, we use $\Regions{}$ to denote one of them in the following.

Knowing the size of~$\Regions{}$, we need to avoid building the whole automaton: instead, we use a multiplication of adjacency matrices associated with each letter or word.
Each of these matrices has as many columns and as many lines as the number of reachable regions of~$\Regions{}$, hence bounded by~$B$.
Those regions are given by the over-approximation previously deduced by analysing the guessed word as in \cref{lemma:N-obs:weak-full:region-number}'s proof.
We define the matrices in the following.

Since we are concerned only with positivity, we assume that each matrix multiplication is followed by setting every positive coefficient to~1.
We begin by computing a first matrix~$M^{\Regions{}}_{\silentaction}$ in which each coefficient $(i,j)$ is~$1$ if there is a silent transition from the \mathnth{i} region to the \mathnth{j} region, and $0$ otherwise.
Similarly, we define, for each letter $a \in \ActionSet \cup \set{t} \cup \ActionSet_{N}$ that is in~$w$, a matrix $M^{\Regions{}}_a$ where the coefficient $(i,j)$ is~$1$ if there is a transition labelled by $a$ from region $i$ to region~$j$, and $0$ otherwise.
Then we compute a matrix $M^{\Regions{}}_{\silentaction^*}$ which coefficients are $\left(M^{\Regions{}}_{\silentaction^*}\right)_{(i,j)} = \max\limits_{0 \leq k \leq B} \left((M^{\Regions{}}_{\silentaction})^k\right)_{(i,j)}$. %
In this matrix, a coefficient~$1$ at position $(i,j)$ means that we can go from region $i$ to region~$j$ without taking any observable transition.

Thus, if we compute for each letter~$a$ the matrix $M^{\Regions{}}_{a_{\silentaction}} = M^{\Regions{}}_{\silentaction^*} \cdot M^{\Regions{}}_a \cdot M^{\Regions{}}_{\silentaction^*}$, we get $1$ at position $\left(M^{\Regions{}}_{a_{\silentaction}}\right)_{(i,j)}$ if and only if there is a path labelled by~$a$ from region $i$ to region~$j$.
We extend this process to words, by setting for any letter~$a$ and any non-empty word~$v$ the matrix $M^{\Regions{}}_{(va)_{\silentaction}} = M^{\Regions{}}_{v_{\silentaction}} \cdot M^{\Regions{}}_{a_{\silentaction}}$.
For instance, the word $abc$ would be represented by the matrix $M^{\Regions{}}_{abc_{\silentaction}} = M^{\Regions{}}_{a_{\silentaction}} \cdot M^{\Regions{}}_{b_{\silentaction}} \cdot M^{\Regions{}}_{c_{\silentaction}}$.
All this can be done in exponential time since the matrices are of exponential dimension in the data size.

We can detect a violation of opacity by finding a word that distinguishes $T_{\TA}$ and~$T_{\TB}$.
To do this, we have to check if their final regions are reached after reading the guessed word~$w$.
We compute the matrices corresponding to~$w$ for both region automata, and apply them to get the vector of reached regions after reading. However, this requires too many multiplications, thus we present below a more efficient alternative to compute these matrices.
To reduce the number of needed multiplications, we take advantage of the structure \cref{word-structure} of the word~$w$, relying on the binary writing of the $k_i$~exponents.
We compute the matrices associated with factors~$t^{2^j}$ by successively squaring the matrix of~$t$.
We limited the exponents~$k_i$ by~$2^{2B}$, thus the exponent~$j$ is limited by~$2B$ and an exponential number of multiplications is required.

We set as initial vector the vector $V^{\Regions{}}_{\silentaction}$ whose \mathnth{i} component is~$1$ only if it corresponds to an initial region.
The word~$w$ is accepted by the region automaton if at least one component associated with a final region is equal to~$1$ in the vector $V^{\Regions{}}_w = V^{\Regions{}}_{\silentaction} \cdot M^{\Regions{}}_{w_{\silentaction}}$.

Hence we can check whether the word~$w$ distinguishes $T_{\TA}$ and~$T_{\TB}$ by looking at the pair of vectors $(V^{\mathcal{R}_{\TA}^{\textit{Tick}}}_w, V^{\mathcal{R}_{\TB}^{\textit{Tick}}}_w)$.
As stated before, this algorithm requires exponential time for the computation of these vectors.
Moreover, the size of the guessed representation of~$w$ is exponential in the size of~$N$, since it consists of giving the exponents $k_i$ written in binary (thus of length~$2B$), the at most $N$ factors $a_i$ and~$f_{K_l}$.
The selection of the word witnessing non-opacity at the beginning uses non-determinism, so the total complexity is \NEXPTIME{}.
Since this algorithm solves the non-inclusion problem for $N$-bounded languages, we obtain a \coNEXPTIME{} algorithm for the inclusion problem.

\subsubsection{\coNEXPTIME{}-hardness}\label{sec:conexphard-Nfirst}

For the following reduction, we consider the problem of non-equivalence of rational expressions without Kleene star and with squaring, known to be \NEXPTIME{}-complete~\cite{SM73}.
For two such expressions $\sqexp_1$ and~$\sqexp_2$, and as done in \cref{section:opacity:discrete-time}, we build the discrete-time TAs $\TA_1$ and~$\TA_2$ representing the same expressions.
The construction of $\TA_1$ and~$\TA_2$ is polynomial, and they have the same language if and only if $\sqexp_1$ and~$\sqexp_2$ are equivalent.
We now need to show that the languages of $\TA_1$ and~$\TA_2$ are $N$-bounded for some~$N$ described polynomially in the size of the rational expressions with square.

Since $\sqexp_1$ and~$\sqexp_2$ are described without the Kleene star, their languages are finite and the length of the words of the languages of $\TA_1$ and~$\TA_2$ is bounded.
More precisely, if $\sqexp_1$ and~$\sqexp_2$ are of length denoted by~$n$, they cannot describe a word longer than $2^{n-1}$ (squaring being the only operation adding succinctness in the expression).
The languages of $\TA_1$ and~$\TA_2$ are hence $2^{n-1}$-bounded.
Furthermore, in the input of the (non-)equality problem of $N$-bounded languages, $N$ is expressed in binary, hence in logarithmic size compared to the number's value.
We thus write $N = 2^{n-1}$ in polynomial space in the size of the given rational expressions.

We have provided a polynomial-time reduction from the non-equality problem for $N$-bounded languages to the inequivalence problem for rational expressions without Kleene star and with squaring.

The $N$-bounded language equality problem is hence \coNEXPTIME{}-hard, and so are the weak and full opacity \wrt{} first $N$ observations decision problems.

\subsection{$N$ fixed times}\label{sec:Nfixed}

In this section, we suppose the attacker can observe the system a finite number~$N$ of times, but not necessarily limited to the first $N$ actions.
We define this new setting for opacity as follows.
Given a time sequence~$\tau$ of length~$N$, we consider that the attacker can only observe the first action immediately following each time from the sequence~$\tau$.
They can be seen as sensor switch-on times, the sensor switching off as soon as an action has been observed.

\begin{exa}
	Consider the following timed word:
		\[(a,1.2)(b,1.5)(c,2)(d,2.3)(a,2.5)(c,2.5)(c,4)(d,5)\]

		Consider the following time sequence~$\tau$ of length~3: $\tau = (0, 2.5, 3)$.
		Then the attacker will observe $(a,1.2)(a,2.5)(c,4)$.
		Now consider the time sequence~$\tau' = (0.7, 1, 6)$.
		Then the attacker will observe $(a,1.2)$.
\end{exa}

In this section, we adapt the definition of opacity to one time sequence, and we prove that it can decided in the same way as opacity with the first $N$ observations (\cref{sec:Nfirst}), and thus generalizes this former setting.

Since the attacker observes only the first action occurring after each observation time~$\tau_i$, given a time sequence~$\tau$, we define the time selection function 
representing this choice.

\begin{defi}[Time selection function following a time sequence]
Let $\tau$ be a time sequence.
We define the \emph{time selection function following~$\tau$} as follows:
\[\begin{array}{rrcl}
\strategy{\tau}: & \TimedWords{\ActionSet} & \longrightarrow & \setRgeqzero \\
			& w & \longmapsto & \tau_{ind_{\tau}(w)}
\end{array}\]
with $\textit{ind}_{\tau}(w \cdot (a,t)) = \min \setsmall{i \in \setN \mid i > \textit{ind}_{\tau}(w) \land \tau_i \geq t}$ and $\textit{ind}_{\tau}(\silentaction)=0$.
\end{defi}
\begin{rem}
The associated projection~$\projection{}$ (defined in \cref{definition:projection}) can be rephrased in the present case of a time sequence~$\tau$.
Indeed, it associates to a timed word $w = w_0 \dots w_n$, with $w_j = (a_j, t_j)$ for $j \in \interval{0}{n}$, the projected trace $w_{k_0} \dots w_{k_{N-1}}$ where for each $i \in \interval{0}{N-1}$, we set
\[k_i = \min \big( \set{k \in \interval{k_{i-1} +1}{n} \mid \tau_{i} \leq t_k \leq \tau_{i+1}} \cup \set{+ \infty} \big)\]
with $w_{\infty} = \silentaction$, $k_{-1} = -1$ and $\tau_{N} = +\infty$.
\end{rem}

Recall we say a TA~$\TA$ is fully (resp.\ weakly) opaque against a time selection function $\strategy{}$ if $\projection{}(\PrivateTr{\TA}) = \projection{}(\PublicTr{\TA})$ (resp.\ $\projection{}(\PrivateTr{\TA}) \subseteq \projection{}(\PublicTr{\TA})$).
In the remainder of this section, we address the following problem of opacity in the setting of time selection functions following time sequences.

\smallskip

\defProblem{Weak (resp.\ full) opacity \wrt{} $N$ static observations decision problem}
{A TA~$\TA$, a finite time sequence $\tau$}
{Is $\TA$ weakly (resp.\ fully) opaque against $\strategy{\tau}$?}

\

In the following lemma, we establish a bijective correspondence between the sets of projected traces associated with two equivalent time sequences.
We use for this a generalisation of the distortion function $\tweak{f}{f'}$ defined on times in the proof of \cref{prop:folklore:rec}.
This is defined for $f$ and~$f'$, the ordered sequences of distinct fractional parts (including 0 and~1) of times from two time sequences $\tau$ and~$\tau'$.
We extend the distortion function to time sequences by denoting $\tweak{f}{f'}((t_i)_{i\in \set I}) = (\tweak{f}{f'}(t_i))_{i \in I}$ for any time sequence~$t$ indexed by a set~$I$.
It is then naturally lifted to timed words by asserting $\untimed{\tweak{f}{f'}(w)} = \untimed{w}$ and $\timed{\tweak{f}{f'}(w)} = \tweak{f}{f'}(\timed{w})$.

\begin{lem}
Let $\tau$ and~$\tau'$ be two equivalent time sequences, and $f,f'$ their respective ordered sequences of distinct fractional parts including 0 and~1.
Then $\tweak{f}{f'} \circ \projection{\tau} = \projection{\tau'} \circ \tweak{f}{f'}$.
\end{lem}
\begin{proof}
Let $w$ be a timed word.
First, we notice that $\timed{w} \simeq \tweak{f}{f'}(\timed{w})$ since the distortion function preserves the integral parts of timestamps and the order of their fractional parts.
Let us denote by~$w'$ the timed word equivalent to~$w$ and such that $\timed{w'} = \tweak{f}{f'}(\timed{w})$, \ie{} $w' = \tweak{f}{f'}(w)$.
Then, using the definition of $\tweak{f}{f'}$, it is easy to show that the times in~$\tau$ and~$\timed{w}$ are ordered in the same way as those in~$\tau'$ and~$\timed{w'}$.
Moreover the definition of the projections $\projection{\tau}$ and~$\projection{\tau'}$ only depends on this order.
Hence a letter~$w_i$ is kept in~$\projection{\tau}(w)$ if and only if the corresponding letter~$w'_i$ is kept in~$\projection{\tau'}(w')$.
Consequently, $\tweak{f}{f'}(\projection{\tau}(w)) = \projection{\tau'}(w')$.
In other words, $\tweak{f}{f'} (\projection{\tau}(w)) = \projection{\tau'}(\tweak{f}{f'}(w))$---and this for any timed word~$w$.
\end{proof}
\begin{prop}[Equivalent time sequences define the same opacity properties]\label{prop:equivalent-time-seq:opacity}
Let $\TA$ be a timed automaton.
Let $\tau$ and~$\tau'$ be two equivalent time sequences.
Then, $\TA$ is fully (resp.\ weakly) opaque following~$\tau$ if and only if $\TA$ is fully (resp.\ weakly) opaque following~$\tau'$.
\end{prop}
\begin{proof} We only treat the case of weak opacity, since it remains inter-reducible with full opacity (\cref{th:weak-full-eq}).

Let $\tau$ and $\tau'$ be two equivalent time sequences, and $f$, $f'$ their respective ordered sequences of distinct fractional parts including 0 and~1.

Assume $\TA$ is fully opaque following~$\projection{\tau}$.
We prove that it is also fully opaque following~$\projection{\tau'}$.

The distortion function $\tweak{f}{f'}$ is a bijection from $\projection{\tau}(\Trace{\TA})$ to $\projection{\tau'}(\Trace{\TA})$.
If we show that the image by $\tweak{f}{f'}$ of an opaque trace is opaque, then it is also shown for its inverse function $\tweak{f'}{f}$.
Hence if there exists some non opaque projected trace with~$\tau$, the corresponding projected trace with~$\tau'$ is well defined and is also non opaque, and \emph{vice-versa}. 
Thus $\projection{\tau}(\Trace{\TA})$ contains a non opaque trace if and only if $\projection{\tau'}(\Trace{\TA})$ also does, which is the expected result.
We now detail the proof of this property of $\tweak{f}{f'}$.

Let $\run , \run'$ be runs on~$\TA$ such that $\run$ is private, $\run'$~is public and $\projection{\tau}(\Trace{\run}) = \projection{\tau}(\Trace{\run'})$.
Since $\tweak{f}{f'}$ preserves both the integral parts and the order of fractional parts of time sequences, we have $\timed{\Trace{\run}} \simeq \tweak{f}{f'}(\timed{\Trace{\run}})$ and the timed word $\tweak{f}{f'}(\Trace{\run})$, denoted by~$w$, is equivalent to $\Trace{\run}$.
According to \cref{lem:folklore:equivalence}, $w$ is hence the trace of a private run in~$\TA$.
We proceed similarly to get $w' = \tweak{f}{f'}(\Trace{\run'})$ the trace of a public run in~$\TA$.
Then, since $\projection{\tau}(\Trace{\run}) = \projection{\tau}(\Trace{\run'})$, we can apply the distortion function $\tweak{f}{f'}$ on both sides and get:
\[\tweak{f}{f'}(\projection{\tau}(\Trace{\run})) = \tweak{f}{f'}(\projection{\tau}(\Trace{\run'})). \]
Applying the previous lemma gives us the following identity.
\[ \projection{\tau'}(\tweak{f}{f'}(\Trace{\run})) = \projection{\tau'}(\tweak{f}{f'}(\Trace{\run'}))\]
\ie{} $\projection{\tau'}(w) = \projection{\tau'}(w')$, this projected trace being thus produced both by a private and by a public run.
This proves that if a trace $v$ is opaque for $\projection{\tau}$ then $\tweak{f}{f'}(v)$ is an opaque trace for~$\projection{\tau'}$.

This concludes the proof that opacity notions are the same when defined by projections associated with equivalent time sequences.
\end{proof}

A consequence of \cref{prop:equivalent-time-seq:opacity} is that we only need to study opacity properties for \emph{some} time sequences to decide them on \emph{all} equivalent time sequences.
Given a time sequence~$\tau$ of size~$N$, we choose a representing equivalent time sequence~$\tausimple$, with only rational times and of reasonable size.
We proceed as follows.
We denote by~$N_f$ the number of different fractional parts appearing in~$\tau$ also different from~$0$.
We deduce from~$\tau$ the order of the fractional parts of the times in~$\tausimple$ and represent it by a function $s : \interval{0}{N-1} \longrightarrow \interval{0}{N_f}$, such that for all $i,j \in \interval{0}{N-1}$,$s(i) \leq s(j)$ iff $\fract{\tau_i} \leq \fract{\tau_j}$; and $s(i) = 0$ iff $\fract{\tau_i} = 0$.
Notice that $(\frac{s(0)}{N_f+1}, \dots, \frac{s(N-1)}{N_f+1}) \simeq (\fract{\tau_0}, \dots, \fract{\tau_{N-1}})$.
Then for each $i \in \interval{0}{N-1}$ we set $\tausimple_i = \intpart{\tau_i} + \frac{s(i)}{N_f+1}$.
This way, we have $\tausimple \simeq \tau$.
Moreover, if we take another $\tau' \simeq \tau$ as starting point, this construction gives the same time sequence $\tausimple' = \tausimple$.

Time sequences of the above chosen form are called \emph{simple} time sequences in the following because all their values are rational and of reasonable size.
These are the time sequences we use to study opacity properties of all equivalent time sequences.

Given a simple time sequence~$\tausimple$, we can now build a timed automaton~$\TA_{\tausimple}$ which simulates the observations with the sequence~$\tausimple$ (in a similar way as the unfolding of the first $N$ observations), and show that $\TA$ is fully opaque with observations following~$\tausimple$ if and only if $\TA_{\tausimple}$ is fully opaque with respect to the first $N$ observations.
The TA $\TA_{\tausimple}$ is obtained by making $2N+1$ copies of~$\TA$, among which $N+1$ copies correspond to the sensor turned off, and $N$ to a sensor turned on (waiting for the next observable action of the system).
The choice of~$\tausimple$ guarantees that only rational constants will occur in the constraints, and we can ensure that $\TA_{\tausimple}$ respects our definition of timed automata by multiplying all constants by $(N_f + 1)$ to obtain only integral constants without changing the opacity properties (indeed, languages comparisons remain the same when all timestamps are multiplied by a factor $(N_f + 1)$).
For this purpose, if $k$ is an integer, we denote by $k\cdot g$ the constraint~$g$ in which all constants have been multiplied by~$k$.

We illustrate this construction in \cref{figure:finite-obs:Nfixed} and formalize it in the following definition.

\begin{defi}[$\tausimple$-unfolded TA]\label{def:unfoldtau}
Let $\TA = \TAprivextend$ be a TA and let $\tausimple$ be a simple time sequence of size~$N$, with $N_f = \vert\set{\fract{\tausimple_i}\mid 0 \leq i < N  }\setminus \set{0}\vert$ the number of different fractional parts in $\tausimple$ also different from~$0$.
We call \emph{$\tausimple$-unfolding of~$\TA$} the timed automaton $\TA_{\tausimple} = (\ActionSet, \LocSet', \loc^0_{0~\textit{Off}}, \PrivSet', \FinalSet', \ClockSet \cup \set{z}, \invariant', \EdgeSet')$ where
	\begin{enumerate}
		\item $\LocSet' = \bigcup\limits_{i=0}^{N-1} \LocSet^i_{\textit{On}} \cup \bigcup\limits_{j=0}^{N} \LocSet^j_{\textit{Off}}$ where the sets $\LocSet^i_{\textit{On}}$ and $\LocSet^j_{\textit{Off}}$ are $2N+1$ disjoint copies of $\LocSet$ where each location $\loc \in \LocSet$ has been renamed into $\loc^i_{\textit{On}} \in \LocSet^i_{\textit{On}}$ and $\loc^j_{\textit{Off}} \in \LocSet^j_{\textit{Off}}$: for $0 \leq i \leq N-1$, $\LocSet^i_{\textit{On}} = \lbrace \loc^i_{\textit{On}} \mid \loc \in \LocSet \rbrace$, and for $0 \leq j \leq N$, $\LocSet^j_{\textit{Off}} = \lbrace \loc^j_{\textit{Off}} \mid \loc \in \LocSet \rbrace$;
		\item $\loc^0_{0~\textit{Off}} \in \LocSet^0_{\textit{Off}}$ is the initial location;
		\item $\PrivSet' = \bigcup\limits_{i=0}^{N-1} \PrivSetOn^i \cup \bigcup\limits_{j=0}^{N} \PrivSetOff^i$
where $\PrivSetOn^i$ and $\PrivSetOff^j$ are the copies respectively within $\LocSet^i_{\textit{On}}$ and $\LocSet^j_{\textit{Off}}$ of the set of private locations of~$\TA$;
		\item $\FinalSet' = (\bigcup\limits_{i=0}^{N-1} \LocSet_{\textit{On}, f}^i) \cup (\bigcup\limits_{j=0}^{N-1}\LocSet_{\textit{Off}, f}^j) \cup \LocSet_{\textit{Off}}^N$ where $\LocSet_{\textit{On}, f}^i$ and $\LocSet_{\textit{Off}, f}^j$ are the copies respectively within $\LocSet^i_{\textit{On}}$ and $\LocSet^j_{\textit{Off}}$ of the set of final locations of~$\TA$;
		\item $\invariant' (\loc^i_{\textit{On}}) = N_f \cdot (\invariant (\loc) \land (z \leq \tau_{i+1}))$ and $\invariant' (\loc^j_{\textit{Off}}) = N_f \cdot (\invariant (\loc) \land (z \leq \tau_{j}))$ for $\loc \in \LocSet$ and $0 \leq i \leq N-1$ and $0 \leq j \leq N$f
extends $\invariant$ to each $\LocSet_{i, \textit{On}}$ and $\LocSet_{j, \textit{On}}$ while keeping all constants integral;
		\item $\EdgeSet' = \bigcup\limits_{i=0}^{N-1} \EdgeSet^i_{\textit{On}} \cup \bigcup\limits_{j=0}^{N} \EdgeSet^j_{\textit{Off}} \cup \bigcup\limits_{i=0}^{N-2} \EdgeSet^i_{\textit{On/On}} \cup \bigcup\limits_{i=0}^{N-1} \EdgeSet^i_{\textit{Off/On}} \cup \bigcup\limits_{i=0}^{N-1} \EdgeSet^i_{\textit{obs}}$ is the set of transitions where, given $0 \leq i \leq N-1$ and $0 \leq j \leq N$
		\begin{itemize}
		\item $\EdgeSet^i_{\textit{On}} = \big\lbrace (\loc^i_{\textit{On}}, \silentaction, N_f \cdot g, R, \loc'^i_{\textit{On}}) \mid (\loc, \silentaction, g, R, \loc') \in \EdgeSet \big\rbrace$;
		\item  $\EdgeSet^j_{\textit{Off}} = \big\lbrace (\loc^j_{\textit{Off}}, \silentaction, N_f \cdot g, R, \loc'^j_{\textit{Off}}) \mid (\loc, a, g, R, \loc') \in \EdgeSet \land a \in \ActionSet \cup \set{\silentaction}\big\rbrace$;
		\item $\EdgeSet^i_{\textit{On/On}} = \big\lbrace (\loc^i_{\textit{On}}, \silentaction, N_f \cdot (z = \tau_{i+1}), \emptyset, \loc^{i+1}_{\textit{On}}) \mid \loc \in \LocSet \big\rbrace$;
		\item $\EdgeSet^i_{\textit{Off/On}} = \big\lbrace (\loc^i_{\textit{Off}}, \silentaction, N_f \cdot (z = \tau_{i}), \emptyset, \loc^{i}_{\textit{On}}) \mid \loc \in \LocSet \big\rbrace$;
		\item $\EdgeSet^i_{\textit{obs}} = \big\lbrace (\loc^i_{\textit{On}}, a, N_f \cdot g, R, \loc'^{i+1}_{\textit{Off}}) \mid (\loc, a, g, R, \loc') \in \EdgeSet, a \in \ActionSet \big\rbrace$.
		\end{itemize}
	\end{enumerate}
\end{defi}
\begin{prop}
Let $\TA$ be a TA, $\tausimple$ a simple time sequence and $\TA_{\tausimple}$ the $\tausimple$-unfolding of~$\TA$.
Then $\TA_{\tausimple}$ is weakly (resp.\ fully) opaque if and only if $\TA$ is weakly (resp.\ fully) opaque against~$\strategy{\tausimple}$.
\end{prop}
\begin{proof}

There is a bijective correspondence between the runs of~$\TA$ and those of~$\TA_{\tausimple}$.
In particular, every run~$\run$ on~$\TA$ is associated with a run~$\run'$ on~$\TA_{\tausimple}$ such that every location visited in~$\run'$ is the unique copy in~$\TA_{\tausimple}$ of the corresponding location in~$\TA$ with same private or public property and that keeps in memory the observation time~$\tau_i$ which currently must be reached before the next observation.

This is exactly the information contained in the projection~$\projection{\tausimple}$.
For instance, if an action has already been observed since the last observation time~$\tausimple_{i}$, the current action is not observed and the current location is in the copy $\LocSet^{i+1}_{\textit{Off}}$.
This corresponds to the case when a letter $(a,t)$ does not appear in the projected trace $\projection{\tausimple}(w \cdot (a,t))$ in the definition of the projection.
Conversely, the letter is kept in the projected trace when no action was observed after the last observation time $\tausimple_i$, which happens in the copy $\LocSet^{i}_{\textit{On}}$.
From this copy, any observable transition leads to copy $\LocSet^{i+1}_{\textit{Off}}$.
Once the last copy $\LocSet^{N+1}_{\textit{Off}}$ is reached, no observation can be performed any more by the attacker.
This corresponds to the $N$-boundedness of the language of the projected traces $\projection{\tausimple}(\Trace{\TA})$, given by the length of the time sequence~$\tausimple$.

The timed automaton $\TA_{\tausimple}$ produces a set of traces which is exactly $\projection{\tausimple}(\Trace{\TA})$ where all times are multiplied by a factor $N_f +1$, which ensure all constants of the TA are integers. 
\end{proof}

\begin{thm}\label{static-N}
Let $N \in \setN$, and let $\tau$ be a $N$-finite time sequence.
The weak and full opacity against the time selection function~$\strategy{\tau}$ decision problem are \coNEXPTIME-complete.
\end{thm}
\begin{proof}
Again, we solely focus on weak opacity since both problems remain inter-reducible.

As a consequence of \cref{prop:equivalent-time-seq:opacity}, $\TA$~is weakly opaque against the time selection function~$\strategy{\tau}$ if and only if it is weakly opaque against the time selection function following the simple time sequence equivalent to~$\tau$.
It thus suffices to test weak opacity of the unfolded TA $\TA_{\tausimple}$.
As stated above, the language of this TA is $N$-bounded, hence we can apply the \coNEXPTIME algorithm of \cref{sec:Nfirst} to check the weak opacity of~$\TA_{\tausimple}$.

We get the hardness result by considering the particular case of the time sequence of length $N$ whose elements are all~0, whose associated projection is exactly~$\projection{N}$, the one for which the hardness result obtained in \cref{sec:Nfirst} holds.
\end{proof}
\begin{figure}[tb]
	\begin{subfigure}[b]{0.2\textwidth}
		\begin{tikzpicture}[pta, node distance=2cm, thin]

		\node[location, initial] (li) at (-8, -2) {$\locinit$};
		\node[location, final] (lf) at (-6, -2) {$\locfinal$};

		\path
		(li) edge[loop above] node{$\styleact{a}$} (li)
		(li) edge node[align=center]{$\styleclock{x}=1$} node[below]{$\styleact{b}$} (lf)
		;

		\end{tikzpicture}
		\caption{TA}
		\label{figure:finite-obs:Nfixed:TA}
	\end{subfigure}

	\begin{subfigure}[b]{0.78\textwidth}
   \centering
	\begin{tikzpicture}[pta, node distance=2cm, thin]

	\node[location, initial] (li0off) at (-3, 0) {$\locinit^0$};
	\node[location, final] (lf0off) at (-1, 0) {$\locfinal^0$};
	
	\path
	(li0off) edge[loop above, align=center] node{$\styleclock{z} < \tau_0$ \\ $\styleact{\silentaction}$ } (li0off)
	(li0off) edge node[align=center]{$\styleclock{x}=1$\\$\land \styleclock{z}<\tau_0$ } node[below]{$\styleact{\silentaction}$} (lf0off)
	;

	\node[location] (li1off) at (1, 0) {$\locinit^1$};
	\node[location, final] (lf1off) at (3, 0) {$\locfinal^1$};
	
	\path
	(li1off) edge[loop above, align=center] node{$\styleclock{z} < \tau_1$ \\ $\styleact{\silentaction}$} (li1off)
	(li1off) edge node[align=center]{$\styleclock{x}=1$\\$\land \styleclock{z}<\tau_1$} node[below]{ $\styleact{\silentaction}$} (lf1off)
	;

	\node[location] (li2off) at (5, 0) {$\locinit^2$};
	\node[location, final] (lf2off) at (7, 0) {$\locfinal^2$};
	
	\path
	(li2off) edge[loop above, align=center] node{$\styleclock{z} < \tau_2$ \\ $\styleact{\silentaction}$} (li2off)
	(li2off) edge node[align=center]{$\styleclock{x}=1$\\$\land \styleclock{z}<\tau_2$} node[below]{$\styleact{\silentaction}$} (lf2off)
	;

	\node[location] (li0on) at (-3, -3) {$\locinit^0$};
	\node[location, final] (lf0on) at (-1, -3) {$\locfinal^0$};
	
	\path
	(li0on) edge node{$\styleact{a}$} (li1off)
	(li0on) edge node[align=center, below, xshift=-3pt, yshift=-2.5pt]{$\styleclock{x}=1$ \\ $\styleact{b}$} (lf1off)
	;

	\node[location] (li1on) at (1, -3) {$\locinit^1$};
	\node[location, final] (lf1on) at (3, -3) {$\locfinal^1$};
	
	\path
	(li1on) edge node{$\styleact{a}$} (li2off)
	(li1on) edge node[align=center, below right, yshift=-1em]{$\styleclock{x}=1$ \\ $\styleact{b}$} (lf2off)
	;

	\path
	(li0off) edge node[align=center, left]{$\styleclock{z}= \tau_0$ \\ $\styleact{o_0}$} (li0on)
	(li1off) edge node[align=center, right]{$\styleclock{z}= \tau_1$ \\ $\styleact{o_1}$} (li1on)
	;

	\path
	(li0on) edge[bend right] node[align=center, below=+0.5pt]{$\styleclock{z}=\tau_1$ \\ $\styleact{o_1}$} (li1on)
	;

 	\node[left of=li0off] {\textit{Off}};
 	\node[left of=li0on] {\textit{On}};

	\node[draw, dashed, fit=(li0off)(lf0off)(li0on)(lf0on), inner xsep=8pt, inner ysep=3.5em, label={[label distance=0.5em]above:{\textbf{0 observation}}} ] {};
	\node[draw, dashed, fit=(li1off)(lf1off)(li1on)(lf1on), inner xsep=8pt, inner ysep=3.5em, label={[label distance=0.5em]above:{\textbf{1 observation}}} ] {};
	\node[draw, dashed, fit=(li2off)(lf2off), inner xsep=8pt, inner ysep=3.5em, label={[label distance=0.5em]above:{\textbf{2 observations}}} ] {};

	\end{tikzpicture}
  \caption{$\tau$-unfolding for $N=2$}

  \end{subfigure}
  \caption{Timed automaton and its $\tau$-unfolding}
  \label{figure:finite-obs:Nfixed}
\end{figure}
\subsection{$N$ times with dynamical strategy}\label{ss:N-dynamic}

In the previous two sections, the strategy of the attacker did not depend on  what they observed: they either observed the first $N$ observations or had $N$~times that were \emph{a~priori} selected.
However, an attacker may wish to adapt their strategy based on the timed word they observed so far.
In this section, we only assume that we are given~$N$ such that the attacker's time selection function is $N$-finite.

\begin{defi}
Let $\TA$ be a TA and $N \in \setN$.
We say that $\TA$ is \emph{weakly (resp.\ fully) opaque \wrt{} $N$ dynamic observations} iff for all $N$-finite time selection function $\strategy{}$
$\TA$ is weakly (resp.\ fully) opaque against $\strategy{}$.
\end{defi}
\defProblem{Weak (resp.\ full) opacity \wrt{} $N$ dynamic observations decision problem}
{A TA~$\TA$, $N\in \setN$}
{Is $\TA$ weakly (resp.\ fully) opaque \wrt{} $N$ dynamic observations?}
\begin{thm}\label{dyn-N}
The weak (resp.\ full) opacity \wrt{} $N$ dynamic observations decision problem is
\coNEXPTIME{}-complete.
\end{thm}
\begin{proof}
Let us start by discussing the hardness.
In \cref{sec:conexphard-Nfirst}, we showed that the 
weak (resp.\ full) opacity against the time selection function $\strategy{N}$ (detecting the first $N$ observations) is \coNEXPTIME{}-hard.
This result was obtained via a reduction from the inequivalence problem for rational expressions with squaring, but without the Kleene star operator.
These expressions produce a language that is exponentially bounded, thus the reduction consisted in selecting $N$ large enough so that everything could be observed.
The same reduction can immediately be used in our dynamic setting: as the traces of runs produced by the system have at most $N$~letters, then there is no better strategy than detecting the first $N$ observations.
Hence, the \coNEXPTIME{}-hardness immediately carries over to the weak and full opacity \wrt{} $N$ dynamic observations decision problems.

We now move to establishing \coNEXPTIME{}-membership.
To do so, we will reduce this
dynamic problem to opacity against the time selection function $\strategy{2N}$ (observing the first $2N$~actions).
Intuitively, we turn the selection of the next timestamp into an observable action of the system.
To do so, we imitate the $\tau$-unfolding of \cref{def:unfoldtau}, though instead of
going to an ``On'' component (where the attacker observes the next action of the system)
when the next timestamp of~$\tau$ is reached, we make this transition observable.
In practice, it mainly means removing every mention of~$z$ from the definition of the $\tau$-unfolded TA and keeping the coefficients as they are in the initial~TA.
Every possible choice of when the attacker starts observing again is allowed in this TA, hence representing every possible time selection function.

Formally, given a TA $\TA = \TAprivextend$ and $N\in\setN$, 
we call \emph{$N$-free-unfolding of~$\TA$} the TA $\TA_{N} = (\ActionSet\cup\{o_i\mid i = 0,\dots, N-1\}, \LocSet', \loc^0_{0~\textit{Off}}, \PrivSet', \FinalSet', \ClockSet, \invariant', \EdgeSet')$
where

\begin{enumerate}
		\item $\LocSet' = \bigcup\limits_{i=0}^{N-1} \LocSet^i_{\textit{On}} \cup \bigcup\limits_{j=0}^{N} \LocSet^j_{\textit{Off}}$ where the sets $\LocSet^i_{\textit{On}}$ and $\LocSet^j_{\textit{Off}}$ are $2N+1$ disjoint copies of $\LocSet$ where each location $\loc \in \LocSet$ has been renamed $\loc^i_{\textit{On}} \in \LocSet^i_{\textit{On}}$ and $\loc^j_{\textit{Off}} \in \LocSet^j_{\textit{Off}}$: for $0 \leq i \leq N-1$, $\LocSet^i_{\textit{On}} = \lbrace \loc^i_{\textit{On}} \mid \loc \in \LocSet \rbrace$, and for $0 \leq j \leq N$, $\LocSet^j_{\textit{Off}} = \lbrace \loc^j_{\textit{Off}} \mid \loc \in \LocSet \rbrace$;
		\item $\loc^0_{0~\textit{Off}} \in \LocSet^0_{\textit{Off}}$ is the initial location;
		\item $\PrivSet' = \bigcup\limits_{i=0}^{N-1} \PrivSetOn^i \cup \bigcup\limits_{j=0}^{N} \PrivSetOff^i$
where $\PrivSetOn^i$ and $\PrivSetOff^j$ are the copies respectively within $\LocSet^i_{\textit{On}}$ and $\LocSet^j_{\textit{Off}}$ of the set of private locations of~$\TA$;
		\item $\FinalSet' = (\bigcup\limits_{i=0}^{N-1} \LocSet_{\textit{On}, f}^i) \cup (\bigcup\limits_{j=0}^{N-1}\LocSet_{\textit{Off}, f}^j) \cup \LocSet_{\textit{Off}}^N$ where $\LocSet_{\textit{On}, f}^i$ and $\LocSet_{\textit{Off}, f}^j$ are the copies respectively within $\LocSet^i_{\textit{On}}$ and $\LocSet^j_{\textit{Off}}$ of the set of final locations of~$\TA$;
		\item $\invariant' (\loc^i_H) = \invariant (\loc)$ 
		for $\loc \in \LocSet$, $H\in\{\textit{Off},\textit{On}\}$ and $0 \leq i \leq N$
extends $\invariant$ to each $\LocSet_{i, \textit{On}}$ and $\LocSet_{j, \textit{On}}$;
		\item $\EdgeSet' = \bigcup\limits_{i=0}^{N-1} \EdgeSet^i_{\textit{On}} \cup \bigcup\limits_{j=0}^{N} \EdgeSet^j_{\textit{Off}} \cup \bigcup\limits_{i=0}^{N-2} \EdgeSet^i_{\textit{On/On}} \cup \bigcup\limits_{i=0}^{N-1} \EdgeSet^i_{\textit{Off/On}} \cup \bigcup\limits_{i=0}^{N-1} \EdgeSet^i_{\textit{obs}}$ is the set of transitions where, given $0 \leq i \leq N-1$ and $0 \leq j \leq N$
		\begin{itemize}
		\item $\EdgeSet^i_{\textit{On}} = \big\lbrace (\loc^i_{\textit{On}}, \silentaction, g, R, \loc'^i_{\textit{On}}) \mid (\loc, \silentaction, g, R, \loc') \in \EdgeSet \big\rbrace$;
		\item  $\EdgeSet^j_{\textit{Off}} = \big\lbrace (\loc^j_{\textit{Off}}, \silentaction,  g, R, \loc'^j_{\textit{Off}}) \mid (\loc, a, g, R, \loc') \in \EdgeSet \land a \in \ActionSet \cup \set{\silentaction} \big\rbrace$;
		\item $\EdgeSet^i_{\textit{On/On}} = \big\lbrace (\loc^i_{\textit{On}}, o_i, true, \emptyset, \loc^{i+1}_{\textit{On}}) \mid \loc \in \LocSet \big\rbrace$;
		\item $\EdgeSet^i_{\textit{Off/On}} = \big\lbrace (\loc^i_{\textit{Off}}, o_i, true, \emptyset, \loc^{i}_{\textit{On}}) \mid \loc \in \LocSet \big\rbrace$; and
		\item $\EdgeSet^i_{\textit{obs}} = \big\lbrace (\loc^i_{\textit{On}}, a, g, R, \loc'^{i+1}_{\textit{Off}}) \mid (\loc, a, g, R, \loc') \in \EdgeSet, a \in \ActionSet \big\rbrace$.
		\end{itemize}
	\end{enumerate}

	We illustrate the aforementioned construction in \cref{figure:finite-obs:Nfree}.
\begin{figure}[tb]
   \centering
	\begin{tikzpicture}[pta, node distance=2cm, thin]

	\node[location, initial] (li0off) at (-3, 0) {$\locinit^0$};
	\node[location, final] (lf0off) at (-1, 0) {$\locfinal^0$};
	
	\path
	(li0off) edge[loop above, align=center] node{$\styleact{\silentaction}$} (li0off)
	(li0off) edge node[above]{$\styleclock{x}=1$} node[below]{$\styleact{\silentaction}$} (lf0off)
	;

	\node[location] (li1off) at (1, 0) {$\locinit^1$};
	\node[location, final] (lf1off) at (3, 0) {$\locfinal^1$};
	
	\path
	(li1off) edge[loop above, align=center] node{$\styleact{\silentaction}$} (li1off)
	(li1off) edge node[above]{$\styleclock{x}=1$} node[below]{$\styleact{\silentaction}$} (lf1off)
	;

	\node[location] (li2off) at (5, 0) {$\locinit^2$};
	\node[location, final] (lf2off) at (7, 0) {$\locfinal^2$};
	
	\path
	(li2off) edge[loop above, align=center] node{$\styleact{\silentaction}$} (li2off)
	(li2off) edge node[above]{$\styleclock{x}=1$} node[below]{$\styleact{\silentaction}$} (lf2off)
	;

	\node[location] (li0on) at (-3, -3) {$\locinit^0$};
	\node[location, final] (lf0on) at (-1, -3) {$\locfinal^0$};
	
	\path
	(li0on) edge node{$\styleact{a}$} (li1off)
	(li0on) edge node[align=center, below, xshift=-3pt, yshift=-2.5pt]{$\styleclock{x}=1$ \\ $\styleact{b}$} (lf1off)
	;

	\node[location] (li1on) at (1, -3) {$\locinit^1$};
	\node[location, final] (lf1on) at (3, -3) {$\locfinal^1$};
	
	\path
	(li1on) edge node{$\styleact{a}$} (li2off)
	(li1on) edge node[align=center, below right, yshift=-1em]{$\styleclock{x}=1$ \\ $\styleact{b}$} (lf2off)
	;

	\path
	(li0off) edge node[align=center, left]{$\styleact{o_0}$} (li0on)
	(li1off) edge node[align=center, right]{$\styleact{o_1}$} (li1on)
	;

	\path
	(li0on) edge[bend right] node[align=center, below=+0.5pt]{$o_1$ } (li1on)
	;

 	\node[left of=li0off] {\textit{Off}};
 	\node[left of=li0on] {\textit{On}};

	\node[draw, dashed, fit=(li0off)(lf0off)(li0on)(lf0on), inner xsep=8pt, inner ysep=3.5em, label={[label distance=0.5em]above:{\textbf{0 observation}}} ] {};
	\node[draw, dashed, fit=(li1off)(lf1off)(li1on)(lf1on), inner xsep=8pt, inner ysep=3.5em, label={[label distance=0.5em]above:{\textbf{1 observation}}} ] {};
	\node[draw, dashed, fit=(li2off)(lf2off), inner xsep=8pt, inner ysep=3.5em, label={[label distance=0.5em]above:{\textbf{2 observations}}} ] {};

	\end{tikzpicture}
  \caption{$N$-free-unfolding of the TA from \cref{figure:finite-obs:Nfixed:TA} for $N=2$}
  \label{figure:finite-obs:Nfree}
\end{figure}

Now, let $\strategy{}$ be an $N$-finite time selection function.
We have the immediate following correspondence:
there exists a run~$\rho$ of~$\TA$ with observation (under~$\strategy{}$)
$w=(a_1,t_1)\dots(a_n,t_n)$ iff 
there exists a run $\rho'$ of~$\TA_{N}$ with observation
$w' = (o_0,\tau_0)(a_1,t_1)(o_1,\tau_1)\dots(a_n,t_n)$ where 
for all $i\geq 0$, $\tau_i=\strategy{}\big((a_1,t_1)\dots(a_i,t_i)\big)$.
As a consequence,   $w \in \projection{}(\PrivateTr{\TA})$ (resp.\ $w\in \projection{}(\PublicTr{\TA})$)
iff $w'\in \projection{2N}(\PrivateTr{\TA})$ (resp.\ $w'\in \projection{2N}(\PublicTr{\TA})$).
Hence, $\TA$ is weakly (resp.\ fully) opaque \wrt{} $N$ dynamic observations
iff $\TA_{N}$ is weakly (resp.\ fully) opaque against the time selection function~$\strategy{2N}$.

The complexity of the \coNEXPTIME{} algorithm provided for \cref{cor:Nfirst:w-f-opacity}
lies mainly in the number of clocks.
So while the number of states in~$\TA_{N}$ is exponentially greater (as $N$ is in binary), testing $\TA_{N}$ against~$\strategy{2N}$ can be done in \coNEXPTIME{}.
Therefore, the weak (resp.\ full) opacity \wrt{} $N$ dynamic observations decision problem lies in \coNEXPTIME{} as well.
\end{proof}
\subsection{Combining restrictions}\label{ss:combined}

In \cref{section:opacity:TA:res}, we considered weak and full opacity for subclasses of~TAs.
In \cref{sec:Nfirst,sec:Nfixed,ss:N-dynamic}, we removed the restrictions on the~TAs, therefore considering the full TA class, and we instead limited the power of the attacker.
It is natural to now consider combining these two kinds of restrictions.
As the weak/full opacity \wrt{} first $N$ observations decision problems are decidable in general, the decidability obviously remains when considering subclasses of TAs, but we will show that the complexity of the algorithms may differ.

Let us now discuss the different combinations.

\begin{thm}\label{oneaction-N}
The weak (resp.\ full) opacity \wrt{} the first $N$ observations decision problem
for one-action TAs is \coNEXPTIME{}-complete.
\end{thm}
\begin{proof}
The general case algorithm for the weak (resp.\ full) opacity \wrt{} first $N$ observations decision problem can be applied to one-action TAs, with its \coNEXPTIME{} complexity.

For the hardness, the reduction discussed in \cref{sec:conexphard-Nfirst} starts with rational expressions with an alphabet of at least two letters, so it cannot be directly applied.
However, within the TA we build for our reduction, two time units elapse between each letter read.
As such, a two-letter alphabet $\{a,b\}$ can be represented with a one letter alphabet $\{a\}$, where $b$ is represented by observing two $a$ in 0-time.
Then, the rest of the reduction applies exactly as well, ensuring \coNEXPTIME{}-hardness.
\end{proof}
\begin{thm}\label{oneclock-N}
The weak (resp.\ full) opacity \wrt{} first $N$ observations decision problem
for one-clock TAs is in \coNEXPTIME{} and is \PSPACE{}-hard.
\end{thm}
\begin{proof}
The general case algorithm for the weak (resp.\ full) opacity \wrt{} first $N$ observations decision problem can be applied to one-clock TAs, with its \coNEXPTIME{} complexity.

For the hardness, consider the reduction discussed in \cref{section:opacity:discrete-time}.
If one assumes the rational expressions do not rely on the squaring operand, then the construction produces a one-clock TA. 
Moreover, given two rational expressions over ``+'', ``$\cdot$'' and the Kleene star, the 
smallest word differentiating the two languages, if it exists, is of exponential size in the rational expressions (if the expressions have size $r_1$ and $r_2$, this word is bounded 
by $2^{r_1+r_2+2}$). By selecting $N$ as the binary encoding of this bound, we thus have
that the two rational expressions are equivalent iff the built TA is fully opaque 
\wrt{} first $N$ observations.
As testing the equality of two rational expressions is \PSPACE{}-complete~\cite{SM73}, 
the weak (resp.\ full) opacity \wrt{} first $N$ observations decision problem
for one-clock TAs is \PSPACE{}-hard.
\end{proof}
\begin{thm}\label{discrete-N}
The weak (resp.\ full) opacity \wrt{} the first $N$ observations decision problem
for discrete time TAs is \coNEXPTIME{}-complete.
\end{thm}
\begin{proof}
The general case algorithm for the weak (resp.\ full) opacity \wrt{} first $N$ observations decision problem can be applied again, with its \coNEXPTIME{} complexity (noting that by forbidding the system from taking a transition when the ``tick'' clock is not~0, we force a general TA to play as if it was over discrete time).
The hardness of the general case directly applies, as it was first designed for discrete-time TAs.
\end{proof}

\begin{thm}\label{oera-N}
The weak (resp.\ full) opacity \wrt{} the first $N$ observations decision problem
for oERAs is \PSPACE{}-complete.
\end{thm}
\begin{proof}
In this case, we can adapt the algorithm for weak and full opacity for oERAs, and limit the search of a path violating the opacity to paths with less than $N$ observations.

The \PSPACE{}-hardness originates from the reachability problem~\cite{AD94}.
The number of observations of the shortest path that reaches a final location in a~TA is at most exponential (bounded by the number of regions of the~TA), and as $N$ is written in binary, one can select $N$ greater than the length of this path, \eg{} $N$ equal to the number of regions. It is hence exponential, but written in polynomial space.
Thus the hardness applies for the opacity \wrt{} first $N$ observations decision problems.
\end{proof}

Similar results can be obtained for the opacity \wrt{} $N$ static (resp.\ dynamic) observations decision problems thanks to the arguments of \cref{sec:Nfixed,ss:N-dynamic}.
The case where the number of clocks is bounded remains partially open. 
Our algorithm for opacity \wrt{} first $N$ observations
increases the number of clocks. As such, the limitation does not immediately reduce the complexity of the algorithm.
However the hardness proof does not apply any more, as we rely on the original clocks of the~TA to encode the squaring operand.

\paragraph{Summary}

We summarize the results of this section in \cref{table-summary-6}.

\begin{table}[tb]
	\centering
	\caption{Summary of the results in \cref{section:finite}}
	\label{table-summary-6}
	\begin{tabular}{| l | c | c | c |}
		\hline
		\rowHeader{} Subclass & First $N$ observations & $N$ static observations & $N$ dynamic observations \\
		\hline
		\cellHeader{}TAs& 
		\multicolumn{3}{c |}{\cellcolor{greenColorBlind!75}\cref{cor:Nfirst:w-f-opacity}, \cref{static-N,dyn-N} (coNEXPTIME{}-complete)}\\
		\cline{1-1}
		\cellHeader{}$|\ActionSet| = 1$ & \multicolumn{3}{c |}{\cellcolor{greenColorBlind!75} \cref{oneaction-N} (coNEXPTIME{}-complete)}\\
		\cline{1-1}
			\cellHeader{}$\Time = \setN$ & \multicolumn{3}{c |}{\cellcolor{greenColorBlind!75} \cref{discrete-N} (coNEXPTIME{}-complete)}  \\
		\hline	
		\cellHeader{}$|\ClockSet| = 1$ & \multicolumn{3}{c |}{\cellcolor{yellow!75}
		\cref{oneclock-N} (\PSPACE{}-hard, and in \coNEXPTIME{})}\\
		\hline
		\cellHeader{}oERAs & \multicolumn{3}{c |}{\cellcolor{greenColorBlind!75} \cref{oera-N} (\PSPACE{}-complete)} \\
		\hline
	\end{tabular}
\end{table}

\section{Conclusion and Perspectives}\label{section:conclusion}

In this paper, we addressed three definitions of opacity on subclasses of TAs, to circumvent the undecidability from~\cite{Cassez09}.
We first proved the inter-reducibility of weak and full opacity.
Then, while undecidability remains for one-action TAs, we retrieve decidability for one-clock TAs without $\silentaction$-transitions, or over discrete time, or for observable ERAs.
Our result for one-clock TAs without $\silentaction$-transitions is tight, since we showed that increasing the number of clocks or adding $\silentaction$-transitions leads to undecidability.
Recall that we summarized the results from \cref{section:opacity:TA:res} in \cref{table-summary} page~\pageref{table-summary}.
Note that all of our undecidability results also hold without invariants, as our proofs do not specifically need them.

We then studied the case of an attacker with an observational power with a limited budget, \ie{} that can only perform a finite set of observations.
We considered three different settings:
\begin{enumerate}%
	\item when the attacker can only see the first $N$ events (letters) of the system;
	\item when the attacker can decide \emph{a~priori} a set of timestamps, and they will observe the first event following each of these timestamps; and
	\item when the attacker can decide this set of timestamps at runtime depending on what observations they made until now.
\end{enumerate}%
We proved all three settings to be decidable and \coNEXPTIME{}-complete on the full TA formalism.
We also studied the combination of the reductions of the attacker power introduced in \cref{section:finite} on the one hand with the restrictions of the model from \cref{section:opacity:TA:res} on the other hand.
Recall that we summarized the results from \cref{section:finite} in \cref{table-summary-6}.

Finally, as proof ingredients for the aforementioned results, we proved that timed language inclusion is \EXPSPACE{}-complete for TAs over discrete time, and that inclusion of $N$-bounded timed languages of TAs is \coNEXPTIME{}-complete (over dense time).

\paragraph{Future work}
Perspectives include being able to build a controller to ensure a~TA is opaque.
This was investigated for execution-time opacity with an untimed controller in~\cite{ABLM22} and a timed controller in~\cite{ADLL25journal}, as well as for non-interference~\cite{GMR07}---but remains to be done within our framework.

Our results in \cref{section:finite} consider an attacker with a fixed attack budget; an interesting future work would be to design efficient procedures attempting to \emph{derive} a maximum attack budget such that the system remains opaque.
These procedures would be semi-algorithms (without guarantee of termination), since this problem reduces to opacity \emph{à~la} Cassez---which is undecidable.

Finally, we aim at investigating parametric versions of these problems, where timing constants considered as parameters (\ie{} unknown constants, \emph{à~la}~\cite{AHV93}) can be tuned to ensure opacity.
This was for example addressed in the setting of execution-time opacity~\cite{ALLMS23,ALM23}.

\paragraph{Acknowledgement}
We thank Marie Duflot for suggesting a proof idea that inspired the definition of the function $\tweak{f}{f'}$.

\appendix

\newpage
\section{Figure: Tick construction}\label{appendix:tick-construction}
\begin{figure}[b]
   \vspace*{-10pt}
   \centering
   \rotatebox{90}{ 
   \scalebox{.88}{
	\begin{tikzpicture}[pta, node distance=2cm, thin, scale = 1.7]
	\node[location, initial] (l0) at (-4, 0) {$\locinit^0$};
	\node[location] (li) at (-3.25, 0.75) {$\loci{i}^0$};
	\node[location] (lj) at (-2, 0.5) {$\loci{j}^0$};

	\node[rectangle, minimum width=7cm, minimum height=8cm, align=center, draw, black!60, dashed] (rectangleA) at (-2.75cm, 0cm) {};
	\node[black!60] at (-3cm, 2.5cm) {$\TA_0$};
	\node[location] (lf0) at (-2.25,-0.75) {$\locfinal^0$};

	\node[location] (lk) at (+1.25, 0) {$\loci{k}^1$};
	\node[rectangle, minimum width=5cm, minimum height=8cm, align=center, draw, black!60, dashed] (rectangleB) at (2cm, 0cm) {};
	\node[black!60] at (2cm, 2.5cm) {$\TA_1$};
	\node[location] (lf1) at (2.5,-0.5) {$\locfinal^1$};

	\node (ln) at (4.5,0) {\dots};
	\node[location] (lm) at (+7, 0) {$\loci{m}^N$};
	\node[rectangle, minimum width=3cm, minimum height=8cm, align=center, draw, black!60, dashed] (rectangleB) at (7cm, 0cm) {};
	\node[black!60] at (7cm, 2.5cm) {$\TA_N$};

	\node[location] (lG0) at (2,-4) {$\loc_{G0}$};
	\node[location, final] (lG1) at (4,-4) {$\loc_{G1}$};
	
	\node[align=left] at (-2,-4) {with $K$ and $I$ resp.\ describing \\ subsets of $\interval{1}{N}$ and $\interval{0}{N}$};

	\path
	(li) edge node[align=center, above]{$g^{(0)}_{<1}$ \\ $\styleact{\silentaction}$} (lj)
	(lj) edge node[align=center, above, xshift=1em]{$g^{(0)}_{<1}$ \\ $\styleact{a_1}$ \\ $\styleclock{\clock_1} \assign 0$\ \ \
	} (lk)
	(ln) edge[bend left] node[align=center, above left]{$g^{(0)}_{<1}$ \\ $\styleact{a_N} $ \\ $\styleclock{\clock_N} \assign 0$
	} (lm)
	;

	\path
	(li) edge[loop above] node[align=center]{$g_{I}$ \\ $\styleact{\silentaction}$ \\ $\styleclock{x_i} \assign 0$ \\ if $i \in I$} ()
	(l0) edge[loop above] node[align=center]{$g_{I}$ \\ $\styleact{\silentaction}$ \\ $\styleclock{x_i} \assign 0$ \\ if $i \in I$} ()
	(lj) edge[loop above] node[align=center]{$g_{I}$ \\ $\styleact{\silentaction}$ \\ $\styleclock{x_i} \assign 0$ \\ if $i \in I$} ()
	(lf0) edge[loop left] node[align=center, below]{$g_{I}$ \\ $\styleact{\silentaction}$ \\ $\styleclock{x_i} \assign 0$ \\ if $i \in I$} ()
	(lk) edge[loop above] node[align=center]{$g_{I}$ \\ $\styleact{\silentaction}$ \\ $\styleclock{x_i} \assign 0$ \\ if $i \in I$} ()
	(lf1) edge[loop above] node[align=center]{$g_{I}$ \\ $\styleact{\silentaction}$ \\ $\styleclock{x_i} \assign 0$ \\ if $i \in I$} ()
	(lm) edge[loop above] node[align=center]{$g_{I}$ \\ $\styleact{\silentaction}$ \\ $\styleclock{x_i} \assign 0$ \\ if $i \in I$} ()
	;

	\path
	(li) edge[loop below] node[align=center]{$\styleclock{x_0}=0$ \\ $\styleact{t}$ \\ $\styleclock{x_0} \assign 0$} ()
	(l0) edge[loop below] node[align=center]{$\styleclock{x_0}=0$ \\ $\styleact{t}$ \\ $\styleclock{x_0} \assign 0$} ()
	(lj) edge[loop below] node[align=center, right]{$\styleclock{x_0}=0$ \\ $\styleact{t}$ \\ $\styleclock{x_0} \assign 0$} ()
	(lf0) edge[loop right] node[align=center]{$\styleclock{x_0}=0$ \\ $\styleact{t}$ \\ $\styleclock{x_0} \assign 0$} ()
	(lk) edge[loop below] node[align=center]{$\styleclock{x_0}=0$ \\ $\styleact{t}$ \\ $\styleclock{x_0} \assign 0$} ()
	(lf1) edge[loop right] node[below, align=center]{$\styleclock{x_0}=0$ \\ $\styleact{t}$ \\ $\styleclock{x_0} \assign 0$} ()
	(lm) edge[loop below] node[align=center]{$\styleclock{x_0}=0$ \\ $\styleact{t}$ \\ $\styleclock{x_0} \assign 0$} ()
	;

	\path
	(lf0) edge node[align=center, below left]{$\styleact{f_{K \cup \set{0}}}$ \\ $g^{(0)}_{K \cup \set{0}}$} (lG0)
	(lf1) edge node[align=center, left, yshift=-2em, xshift=-1em]{$\styleact{f_{K \cup \set{0}}}$ \\ $g^{(0)}_{K \cup \set{0}}$} (lG0)
	(lm) edge node[align=center]{$\styleact{f_{K \cup \set{0}}}$ \\ $g^{(0)}_{K \cup \set{0}}$} (lG0)
	(lG0) edge node[align=center, below]{$\styleclock{x_0} = 1$ \\ $\land g_{<1}$ \\ $\styleact{\silentaction}$} (lG1)
	(lG0) edge[loop below] node[align=center, left]{$g^{(0)}_{K}$\\$\styleact{f_{K}}$} (lG0);
	\end{tikzpicture}
	}}

  \caption{The ``tick'' construction $\TickN{\TA}$ (see \cref{def:tick-construction})}
  \label{figure:finite-obs:tick}
\end{figure}
\newpage

	\newcommand{\CCIS}{Communications in Computer and Information Science}
	\newcommand{\CSUR}{{ACM} Computing Surveys}
	\newcommand{\ENTCS}{Electronic Notes in Theoretical Computer Science}
	\newcommand{\FAC}{Formal Aspects of Computing}
	\newcommand{\FundInf}{Fundamenta Informaticae}
	\newcommand{\FMSD}{Formal Methods in System Design}
	\newcommand{\IJFCS}{International Journal of Foundations of Computer Science}
	\newcommand{\IJSSE}{International Journal of Secure Software Engineering}
	\newcommand{\IJC}{International Journal of Control}
	\newcommand{\IseCure}{International Journal of Information Security}
	\newcommand{\IPL}{Information Processing Letters}
	\newcommand{\IC}{Information and Computation}
	\newcommand{\JAIR}{Journal of Artificial Intelligence Research}
	\newcommand{\JLAP}{Journal of Logic and Algebraic Programming}
	\newcommand{\JLAMP}{Journal of Logical and Algebraic Methods in Programming} %
	\newcommand{\JISA}{Journal of Information Security and Applications}
	\newcommand{\JLC}{Journal of Logic and Computation}
	\newcommand{\JALC}{Journal of Automata, Languages and Combinatorics}
	\newcommand{\LMCS}{Logical Methods in Computer Science}
	\newcommand{\LNCS}{Lecture Notes in Computer Science}
	\newcommand{\RESS}{Reliability Engineering \& System Safety}
	\newcommand{\RTS}{Real-Time Systems}
	\newcommand{\SCP}{Science of Computer Programming}
	\newcommand{\SOSYM}{Software and Systems Modeling ({SoSyM})}
	\newcommand{\STTT}{International Journal on Software Tools for Technology Transfer}
	\newcommand{\TCS}{Theoretical Computer Science}
	\newcommand{\TOPLAS}{{ACM} Transactions on Programming Languages and Systems ({ToPLAS})}
	\newcommand{\ToPNoC}{Transactions on {P}etri Nets and Other Models of Concurrency}
	\newcommand{\TOSEM}{{ACM} Transactions on Software Engineering and Methodology ({ToSEM})}
	\newcommand{\TSE}{{IEEE} Transactions on Software Engineering}
	\newcommand{\TCAD}{{IEEE} Transactions on Computer-Aided Design of Integrated Circuits and Systems}
\bibliographystyle{alphaurl}
\bibliography{PTA}

@article{AA23survey,
	author    = {Arcile, Johan and
				Andr{\'e}, {\'E}tienne},
	title     = {Timed automata as a formalism for expressing security: A survey on theory and practice},
	journal   = {\CSUR{}},
	doi       = {10.1145/3534967},
	year      = {2023},
	volume    = {55},
	number    = {6},
	day       = {31},
	month     = jul,
	pages	  = {1-36},
	coreRank  = {Astarrank},
}

@inproceedings{AAL24,
	author       = {Andr{\'e}, {\'E}tienne and Arcile, Johan and Lefaucheux, Engel},
	author+an    = {1=highlight},
	editor       = {Siddharth Barman and Sławomir Lasota},
	title        = {Execution-time opacity problems in one-clock parametric timed automata},
	booktitle = {Proceedings of the 44th {IARCS} Annual Conference on Foundations of Software Technology and Theoretical Computer Science ({FSTTCS} 2024)},
	SHORTbooktitle    = {{FSTTCS}},
	location     = {Gandhinagar, India},
	eventdate    = {2024-12-16/2024-12-18},
	year         = {2024},
  	month        = dec,
	pages        = {3:1--3:22},
	ISBN         = {978-3-95977-355-3},
	ISSN         = {1868-8969},
	doi          = {10.4230/LIPIcs.FSTTCS.2024.3},
	volume       = {323},
	publisher    = {Schloss Dagstuhl -- Leibniz-Zentrum für Informatik},
	series       = {Leibniz International Proceedings in Informatics (LIPIcs)},
  	lipn-category = {intc},
}

@inproceedings{AB26,
  	author       = {Andr{\'e}, {\'E}tienne and Bakiri, Lydia},
	author+an    = {1=highlight},
  	title        = {Opacity Problems in Multi-energy Timed Automata},
  	month        = mar,
	editor       = {Sun, Youcheng and Hu, Jingtong and Hsiu, Pi-Cheng},
	booktitle = {Proceedings of the 41st ACM/SIGAPP Symposium On Applied Computing ({SAC} 2026)},
	SHORTbooktitle    = {{SAC}},
	location     = {Thessaloniki, Greece},
	eventdate    = {2026-03-23/2026-03-27},
	pages        = {317-324},
	year         = {2026},
	doi          = {10.1145/3748522.3779881},
	publisher    = {{ACM}},
}

@inproceedings{ABLM22,
  	author    = {Andr{\'e}, {\'E}tienne and Bolat, Shapagat and Lefaucheux, Engel and Marinho, Dylan},
	editor    = {Cyrille Artho and Peter Ölveczky},
	title     = {{strategFTO}: Untimed control for timed opacity},
	year      = {2022},
	booktitle = {Proceedings of the 8th {ACM} {SIGPLAN} International Workshop on Formal Techniques for Safety-Critical Systems ({FTSCS} 2022)},
	SHORTbooktitle = {{FTSCS}},
	location  = {Auckland, New Zealand},
	eventdate = {2022-12-07},
	publisher = {{ACM}},
	pages     = {27-33},
	doi       = {10.1145/3563822.3568013},
	coreRank  = {absent},
}

@article{AD94,
	author    = {Alur, Rajeev and Dill, David L.},
	title     = {A theory of timed automata},
	doi       = {10.1016/0304-3975(94)90010-8},
	journal   = {\TCS{}},
	volume    = {126},
	month     = apr,
	number    = {2},
	numpages  = {53},
	year      = {1994},
	pages     = {183--235},
	publisher = {Elsevier Science Publishers Ltd.},
	address   = {Essex, UK},
}

@inproceedings{ADL24,
	author       = {Andr{\'e}, {\'E}tienne and Dépernet, Sarah and Lefaucheux, Engel},
	editor       = {Kazuhiro Ogata and Meng Sun and Dominique Méry},
	title        = {The Bright Side of Timed Opacity},
	booktitle = {Proceedings of the 25th International Conference on Formal Engineering Methods ({ICFEM} 2024)},
	SHORTbooktitle    = {{ICFEM}},
	location     = {Hiroshima, Japan},
	eventdate    = {2024-12-02/2024-12-06},
	year         = {2024},
  	month        = dec,
	series       = {\LNCS{}},
	volume       = {15394},
	pages        = {51--69},
	publisher    = {Springer},
	doi          = {10.1007/978-981-96-0617-7_4},
  	lipn-category = {intc},
	acceptanceRate = {44\,\%},
	coreRank     = {Crank},
}

@article{ADLL25journal,
	author       = {Andr{\'e}, {\'E}tienne and Duflot, Marie and Laetitia Laversa and Engel Lefaucheux},
	author+an    = {1=highlight},
	title        = {Execution-time opacity control for timed automata},
	journal      = {Software and Systems Modeling},
	year         = {2026},
	publisher    = {Springer},
	toappear     = {true},
	note         = {To appear},
}

@article{ADOQW08,
	author       = {Abdulla, Parosh Aziz and
				Johann Deneux and
				Joël Ouaknine and
				Karin Quaas and
				James Worrell},
	title        = {Universality Analysis for One-Clock Timed Automata},
	journal      = {\FundInf{}},
	volume       = {89},
	number       = {4},
	pages        = {419--450},
	year         = {2008},
	bibsource    = {dblp computer science bibliography, https://dblp.org},
}

@article{AEYM21,
	author       = {Ikhlass Ammar and El Touati, Yamen and Moez Yeddes and John Mullins},
	title        = {Bounded opacity for timed systems},
	journal      = {\JISA{}},
	volume       = {61},
	pages        = {1-13},
	year         = {2021},
	month     = sep,
	issn      = {2214-2126},
	doi          = {10.1016/j.jisa.2021.102926},
}

@Article{AFH99,
	author  = {Alur, Rajeev and Fix, Limor and Henzinger, Thomas A.},
	journal = {\TCS{}},
	title   = {Event-Clock Automata: {A} Determinizable Class of Timed Automata},
	year    = {1999},
	number  = {1-2},
	pages   = {253--273},
	volume  = {211},
	doi     = {10.1016/S0304-3975(97)00173-4},
}

@inproceedings{AGWZH24,
  	author       = {Jie An and Qiang Gao and Lingtai Wang and Naijun Zhan and Ichiro Hasuo},
	editor       = {André Platzer and Rozier, Kristin-Yvonne and Matteo Pradella and Matteo Rossi},
	title        = {The Opacity of Timed Automata},
	year         = {2024},
	booktitle = {Proceedings of the 26th International Symposium on Formal Methods ({FM} 2024)},
	SHORTbooktitle    = {{FM}},
	location     = {Milano, Italy},
	eventdate    = {2024-09-09/2024-09-13},
	series       = {\LNCS{}},
	volume       = {14933},
	pages        = {620-637},
	publisher    = {Springer},
	doi          = {10.1007/978-3-031-71162-6_32},
	coreRank     = {absent},
}

@inproceedings{AHV93,
	author       = {Alur, Rajeev and Henzinger, Thomas A. and Vardi, Moshe Y.},
	title        = {Parametric real-time reasoning},
	SHORTbooktitle    = {STOC},
	booktitle = {Proceedings of the 25th annual {ACM} symposium on Theory of computing ({STOC} 1993)},
	editor       = {S. Rao Kosaraju and
               Johnson, David S. and
               Alok Aggarwal},
	year         = {1993},
	doi          = {10.1145/167088.167242},
	eventdate    = {1993-05-16/1993-05-18},
	location     = {San Diego, California, USA},
	pages        = {592--601},
	publisher    = {ACM},
	address   = {New York, NY, USA},
}

@inproceedings{AK20,
  	author       = {Andr{\'e}, {\'E}tienne and Kryukov, Aleksander},
  	title        = {Parametric non-interference in timed automata},
	pages        = {37-42},
	year         = {2020},
	booktitle = {Proceedings of the 25th International Conference on Engineering of Complex Computer Systems ({ICECCS} 2020)},
	SHORTbooktitle    = {{ICECCS}},
	editor       = {Li, Yi and Liew, Alan},
	location     = {Singapore},
	eventdate    = {2021-03-04/2021-03-06},
	doi          = {10.1109/ICECCS51672.2020.00012},
	coreRank     = {Arank},
	acceptanceRate = {33\,\%},
  	lipn-category = {intc},
  	lipn-time    = {ant},
}

@inproceedings{ALLMS23,
	author       = {Andr{\'e}, {\'E}tienne and Lefaucheux, Engel and Lime, Didier and Marinho, Dylan and Sun, Jun},
	title        = {Configuring Timing Parameters to Ensure Execution-Time Opacity in Timed Automata},
	booktitle = {Proceedings of the First Workshop on Trends in Configurable Systems Analysis ({TiCSA} 2023)},
	SHORTbooktitle    = {{TiCSA}},
	location     = {Paris, France},
	eventdate    = {2023-04-23},
	series       = {Electronic Proceedings in Theoretical Computer Science},
	volume       = {392},
	pages        = {1--26},
	year         = {2023},
	editor       = {ter Beek, Maurice H. and Dubslaff, Clemens},
	doi          = {10.4204/EPTCS.392.1},
	note         = {Invited paper.},
}

@inproceedings{ALM23,
  	author       = {Andr{\'e}, {\'E}tienne and Lefaucheux, Engel and Marinho, Dylan},
	editor       = {Ait-Ameur, Yamine and Khendek, Ferhat},
	title        = {Expiring opacity problems in parametric timed automata},
	year         = {2023},
	booktitle = {Proceedings of the 27th International Conference on Engineering of Complex Computer Systems ({ICECCS} 2023)},
	SHORTbooktitle    = {{ICECCS}},
	location     = {Toulouse, France},
	eventdate    = {2023-06-14/2023-06-16},
	pages        = {89-98},
	acceptanceRate = {31\,\%},
	doi          = {10.1109/ICECCS59891.2023.00020},
	coreRank     = {Brank},
  	lipn-category = {intc},
}

@article{ALMS22,
	author       = {Andr{\'e}, {\'E}tienne and
				Lime, Didier and
				Marinho, Dylan and
				Sun, Jun},
	title        = {Guaranteeing timed opacity using parametric timed model checking},
	journal      = {\TOSEM{}},
	volume       = {31},
	month        = oct,
	number       = {4},
	year         = {2022},
	doi          = {10.1145/3502851},
	pages        = {1-36},
	coreRank     = {Astarrank},
}

@article{BCLR15,
	author    = {Gilles Benattar and
				Franck Cassez and
				Didier Lime and
				Roux, Olivier H.},
	title     = {Control and synthesis of non-interferent timed systems},
	journal   = {\IJC{}},
	volume    = {88},
	number    = {2},
	pages     = {217--236},
	year      = {2015},
	doi       = {10.1080/00207179.2014.944356},
	bibsource = {dblp computer science bibliography, https://dblp.org},
}

@article{BDST02,
	author       = {Barbuti, Roberto and
				Nicoletta De Francesco and
				Santone, Antonella and
				Tesei, Luca},
	title        = {A Notion of Non-Interference for Timed Automata},
	journal      = {\FundInf{}},
	volume       = {51},
	number       = {1-2},
	pages        = {1--11},
	year         = {2002},
	bibsource    = {dblp computer science bibliography, https://dblp.org},
}

@article{BKMR08,
	author       = {Bryans, Jeremy W. and
				Maciej Koutny and
				Laurent Mazaré and
				Ryan, Peter Y. A.},
	title        = {Opacity generalised to transition systems},
	journal      = {\IseCure{}},
	volume       = {7},
	number       = {6},
	pages        = {421--435},
	year         = {2008},
	doi          = {10.1007/s10207-008-0058-x},
	bibsource    = {dblp computer science bibliography, https://dblp.org},
}

@article{BT03,
	author       = {Roberto Barbuti and
				Luca Tesei},
	title        = {A Decidable Notion of Timed Non-Interference},
	journal      = {\FundInf{}},
	volume       = {54},
	number       = {2-3},
	pages        = {137--150},
	year         = {2003},
	bibsource    = {dblp computer science bibliography, https://dblp.org},
}

@InProceedings{Cassez09,
	author        = {Cassez, Franck},
	SHORTbooktitle     = {{ISA}},
	eventdate    = {2009-06-25/2009-06-27},
	location      = {Seoul, Korea},
	booktitle = {Proceedings of the Third International Conference and Workshops on Advances in Information Security and Assurance ({ISA} 2009)},
	title         = {The Dark Side of Timed Opacity},
	year          = {2009},
	editor        = {Park, Jong Hyuk and Chen, Hsiao{-}Hwa and Mohammed Atiquzzaman and Lee, Changhoon and Kim, Tai{-}Hoon and Yeo, Sang{-}Soo},
	pages         = {21--30},
	publisher     = {Springer},
	series        = {\LNCS{}},
	volume        = {5576},
	bibsource     = {dblp computer science bibliography, https://dblp.org},
	doi           = {10.1007/978-3-642-02617-1_3},
}

@article{CG00,
	author       = {Christian Choffrut and
					Massimiliano Goldwurm},
	title        = {Timed Automata with Periodic Clock Constraints},
	journal      = {Journal of Automata, Languages and Combinatorics},
	volume       = {5},
	number       = {4},
	pages        = {371--403},
	year         = {2000},
	doi          = {10.25596/JALC-2000-371},
	bibsource    = {dblp computer science bibliography, https://dblp.org},
}

@article{CHSJX22,
	author       = {Aidong Chen and
					Chen Hong and
					Xinna Shang and
					Hongyuan Jing and
					Sen Xu},
	title        = {Timing leakage to break {SM2} signature algorithm},
	journal      = {Journal of Information Security and Applications},
	volume       = {67},
	pages        = {103210},
	year         = {2022},
	doi          = {10.1016/J.JISA.2022.103210},
	bibsource    = {dblp computer science bibliography, https://dblp.org},
}

@article{Dima01,
	author       = {Catalin Dima},
	title        = {Real-Time Automata},
	journal      = {\JALC{}},
	volume       = {6},
	number       = {1},
	pages        = {3--23},
	year         = {2001},
	doi          = {10.25596/jalc-2001-003},
	bibsource    = {dblp computer science bibliography, https://dblp.org},
}

@techreport{DQY25,
	author       = {Weilin Deng and Daowen Qiu and Jingkai Yang},
	title        = {New Insights into the Decidability of Opacity in Timed Automata},
	institution  = {{arXiv}},
	number       = {abs/2504.00625},
	year         = {2025},
	url          = {https://arxiv.org/abs/2504.00625},
	eprinttype   = {arXiv},
	eprint       = {2504.00625},
}

@article{FJ15,
	author       = {John Fearnley and
					Marcin Jurdzinski},
	title        = {Reachability in two-clock timed automata is {PSPACE}-complete},
	journal      = {\IC{}},
	volume       = {243},
	pages        = {26--36},
	year         = {2015},
	doi          = {10.1016/J.IC.2014.12.004},
	bibsource    = {dblp computer science bibliography, https://dblp.org},
}

@article{GMR07,
	author    = {Guillaume Gardey and
				John Mullins and
				Roux, Olivier H.},
	title     = {Non-Interference Control Synthesis for Security Timed Automata},
	journal   = {\ENTCS{}},
	volume    = {180},
	number    = {1},
	pages     = {35--53},
	year      = {2007},
	doi       = {10.1016/j.entcs.2005.05.046},
	bibsource = {dblp computer science bibliography, https://dblp.org},
}

@inproceedings{KKG24,
	author       = {Julian Klein and
					Paul Kogel and
					Sabine Glesner},
	editor       = {Nico Plat and
					Stefania Gnesi and
					Furia, Carlo A. and
					Antónia Lopes},
	title        = {Verifying Opacity of Discrete-Timed Automata},
	booktitle    = {Proceedings of the 12th {IEEE/ACM} International Conference on Formal Methods in Software Engineering ({FormaliSE} 2024)},
	SHORTbooktitle    = {{FormaliSE}},
	eventdate    = {2024-04-14/2024-04-15},
	location     = {Lisbon, Portugal},
	pages        = {55--65},
	publisher    = {{ACM}},
	year         = {2024},
	doi          = {10.1145/3644033.3644376},
	bibsource    = {dblp computer science bibliography, https://dblp.org},
}

@inproceedings{LMS04,
	author       = {François Laroussinie and
				Nicolas Markey and
				Philippe Schnoebelen},
	editor       = {Philippa Gardner and
				Nobuko Yoshida},
	title        = {Model Checking Timed Automata with One or Two Clocks},
	booktitle = {Proceedings of the 15th International Conference on Concurrency Theory ({CONCUR} 2004)},
	SHORTbooktitle    = {{CONCUR}},
	eventdate    = {2004-08-31/2004-09-03},
	location     = {London, UK},
	series       = {\LNCS{}},
	volume       = {3170},
	pages        = {387--401},
	publisher    = {Springer},
	year         = {2004},
	doi          = {10.1007/978-3-540-28644-8_25},
	bibsource    = {dblp computer science bibliography, https://dblp.org},
}

@article{LRNA17,
	author       = {Florian Lorber and
					Amnon Rosenmann and
					Dejan Nickovic and
					Aichernig, Bernhard K.},
	title        = {Bounded determinization of timed automata with silent transitions},
	journal      = {Real-Time Systems},
	volume       = {53},
	number       = {3},
	pages        = {291--326},
	year         = {2017},
	doi          = {10.1007/S11241-017-9271-X},
	bibsource    = {dblp computer science bibliography, https://dblp.org},
}

@inproceedings{Mazare04,
	author		= {Mazaré, Laurent},
	editor		= {Ryan, Peter},
	title		= {Using unification for opacity properties},
	booktitle	= {Proceedings of the Fourth {IFIP} {WG} 1.7, {ACM} {SIGPLAN} and {GI} {FoMSESS} Workshop on Issues in the Theory of Security ({WITS} 2004)},
	SHORTbooktitle	= {{WITS}},
	eventdate	= {2004-04-03/2004-04-04},
	location	= {Barcelona, Spain},
	month		= apr,
	pages		= {165-176},
	year		= {2004},
}

@inproceedings{MS72,
	author       = {Meyer, Albert R. and
					Stockmeyer, Larry J.},
	title        = {The Equivalence Problem for Regular Expressions with Squaring Requires Exponential Space},
	booktitle = {Proceedings of the 13th Annual Symposium on Switching and Automata Theory ({SWAT} 1972)},
	SHORTbooktitle    = {{SWAT}},
	eventdate    = {1972-10-25/1972-10-27},
	location     = {College Park, Maryland, USA},
	pages        = {125--129},
	publisher    = {{IEEE} Computer Society},
	year         = {1972},
	doi          = {10.1109/SWAT.1972.29},
	bibsource    = {dblp computer science bibliography, https://dblp.org},
}

@inproceedings{OW04,
	author       = {Joël Ouaknine and
				James Worrell},
	title        = {On the Language Inclusion Problem for Timed Automata: Closing a Decidability Gap},
	booktitle = {Proceedings of the 19th {IEEE} Symposium on Logic in Computer Science ({LICS} 2004)},
	SHORTbooktitle    = {{LICS}},
	eventdate    = {2004-07-14/2004-07-17},
	location     = {Turku, Finland},
	pages        = {54--63},
	publisher    = {{IEEE} Computer Society},
	year         = {2004},
	doi          = {10.1109/LICS.2004.1319600},
	bibsource    = {dblp computer science bibliography, https://dblp.org},
}

@inproceedings{SLR23,
	author       = {Anthony Spriet and
					Didier Lime and
					Roux, Olivier H.},
	editor       = {Laure Petrucci and
					Jeremy Sproston},
	title        = {Timed Non-interference Under Partial Observability and Bounded Memory},
	booktitle  = {Proceedings of the 21st International Conference on Formal Modeling and Analysis of Timed Systems ({FORMATS} 2023)},
	SHORTbooktitle    = {{FORMATS}},
	eventdate    = {2023-09-19/2023-09-21},
	location     = {Antwerp, Belgium},
	series       = {\LNCS{}},
	volume       = {14138},
	pages        = {122--137},
	publisher    = {Springer},
	year         = {2023},
	doi          = {10.1007/978-3-031-42626-1_8},
	bibsource    = {dblp computer science bibliography, https://dblp.org},
}

@inproceedings{SM73,
	author       = {Stockmeyer, Larry J. and
					Meyer, Albert R.},
	editor       = {Aho, Alfred V. and
					Allan Borodin and
					Constable, Robert L. and
					Floyd, Robert W. and
					Harrison, Michael A. and
					Karp, Richard M. and
					H. Raymond Strong},
	title        = {Word Problems Requiring Exponential Time: Preliminary Report},
	booktitle = {Proceedings of the 5th Annual {ACM} Symposium on Theory of Computing ({STOC} 1973)},
	SHORTbooktitle    = {{STOC}},
	eventdate    = {1973-04-30/1973-05-02},
	location     = {Austin, Texas, {USA}},
	pages        = {1--9},
	publisher    = {{ACM}},
	year         = {1973},
	doi          = {10.1145/800125.804029},
	bibsource    = {dblp computer science bibliography, https://dblp.org},
}

@Inbook{Standaert2010,
	author       = {Standaert, François-Xavier},
	title        = {Introduction to Side-Channel Attacks},
	booktitle    = {Secure Integrated Circuits and Systems},
	series       = {Integrated Circuits and Systems},
	pages        = {27--42},
	publisher    = {Springer},
	year         = {2010},
	doi          = {10.1007/978-0-387-71829-3_2},
	bibsource    = {dblp computer science bibliography, https://dblp.org},
}

@incollection{WZ18,
	author    = {Lingtai Wang and
	Naijun Zhan},
	editor    = {Jones, Cliff B. and
	Ji Wang and
	Naijun Zhan},
	title     = {Decidability of the Initial-State Opacity of Real-Time Automata},
	booktitle = {Symposium on Real-Time and Hybrid Systems - Essays Dedicated to Professor Chaochen Zhou on the Occasion of His 80th Birthday},
	series    = {\LNCS{}},
	volume    = {11180},
	pages     = {44--60},
	publisher = {Springer},
	year      = {2018},
	doi       = {10.1007/978-3-030-01461-2_3},
	bibsource = {dblp computer science bibliography, https://dblp.org},
}

@article{WZA18,
	author    = {Lingtai Wang and
	Naijun Zhan and
	Jie An},
	title     = {The Opacity of Real-Time Automata},
	journal   = {\TCAD{}},
	volume    = {37},
	number    = {11},
	pages     = {2845--2856},
	year      = {2018},
	doi       = {10.1109/TCAD.2018.2857363},
	bibsource = {dblp computer science bibliography, https://dblp.org},
}

@article{Zhang24,
	author       = {Kuize Zhang},
	title        = {State-based opacity of labeled real-time automata},
	journal      = {\TCS{}},
	volume       = {987},
	pages        = {114373},
	year         = {2024},
	doi          = {10.1016/J.TCS.2023.114373},
	bibsource    = {dblp computer science bibliography, https://dblp.org},
}

\end{document}